\documentclass[papersize,journal]{IEEEtran}
\usepackage{amsmath,amsfonts}
\usepackage{algorithmic}
\usepackage{algorithm}
\usepackage{array}
\usepackage[caption=false,font=normalsize,labelfont=sf,textfont=sf]{subfig}
\usepackage{textcomp}
\usepackage{stfloats}
\usepackage{url}
\usepackage{verbatim}
\usepackage{graphicx}
\usepackage{cite}
\usepackage{mathtools}
\usepackage{amsthm}
\usepackage{xcolor}
\usepackage{fancybox}
\usepackage{stmaryrd}
\usepackage{ytableau}
\usepackage{dsfont}
\usepackage{physics}
\usepackage{hyperref}

\ytableausetup{boxsize=1em}

\newtheorem{Def}{Definition}
\newtheorem{Thm}[Def]{Theorem}
\newtheorem{Cor}[Def]{Corollary}
\newtheorem{Prop}[Def]{Proposition}
\newtheorem{Lem}[Def]{Lemma}

\newcommand{\1}{\mathds{1}}
\newcommand{\dbra}[1]{\langle\!\langle {#1} \vert}
\newcommand{\dket}[1]{\vert {#1} \rangle\!\rangle}

\newcommand{\dketbra}[1]{\vert {#1} \rangle\!\rangle\!\langle\!\langle {#1} \vert}

\newcommand{\STab}{{\mathrm{STab}}}
\newcommand{\young}[2]{\mathbb{Y}^{#1}_{#2}}
\newcommand{\SU}{\mathrm{SU}}

\newcommand{\mcL}{\mathcal{L}}
\newcommand{\mcX}{\mathcal{X}}
\newcommand{\mcY}{\mathcal{Y}}
\newcommand{\mcZ}{\mathcal{Z}}
\newcommand{\mcU}{\mathcal{U}}
\newcommand{\mcS}{\mathcal{S}}
\newcommand{\mcI}{\mathcal{I}}
\newcommand{\mcO}{\mathcal{O}}
\newcommand{\mcP}{\mathcal{P}}
\newcommand{\mcF}{\mathcal{F}}

\newcommand{\mcA}{\mathcal{A}}
\newcommand{\mcB}{\mathcal{B}}
\newcommand{\mcM}{\mathcal{M}}
\newcommand{\CC}{\mathbb{C}}
\newcommand{\mfS}{\mathfrak{S}}
\newcommand{\sfT}{{\sf T}}

\newcommand{\range}[1]{\{1, \ldots, #1\}}

\begin{document}

\title{One-to-One Correspondence between Deterministic Port-Based Teleportation and Unitary Estimation}

\author{Satoshi Yoshida, Yuki Koizumi, Micha\l{} Studzi\'{n}ski, Marco T\'{u}lio Quintino, Mio Murao
\thanks{S.~Yoshida is with Department of Physics, Graduate School of Science, The University of Tokyo, Tokyo 113-0033, Japan (email:\href{mailto:satoshiyoshida.phys@gmail.com}{satoshiyoshida.phys@gmail.com}).}
\thanks{Y.~Koizumi is with Department of Physics, Faculty of Science, The University of Tokyo, Tokyo 113-0033, Japan, and also with Department of Applied Physics, Graduate School of Engineering, The University of Tokyo, Tokyo 113-8656, Japan.}
\thanks{M.~Studzi\'{n}ski is with Institute of Theoretical Physics and Astrophysics, Faculty of Mathematics, Physics and Informatics, University of Gda\'{n}sk, Wita Stwosza 57, 80-308 Gda\'{n}sk, Poland and International Centre for Theory of Quantum Technologies, University of Gda\'{n}sk, 80-309 Gda\'{n}sk, Poland.}
\thanks{M.~T.~Quintino is with Sorbonne Universit\'{e}, CNRS, LIP6, F-75005 Paris, France.}
\thanks{M.~Murao is with Department of Physics, Graduate School of Science, The University of Tokyo, Tokyo 113-0033, Japan, and also with Trans-scale Quantum Science Institute, The University of Tokyo, Tokyo 113-0033, Japan.}
\thanks{This paper was presented at Quantum Innovation 2024 and the 28th Annual Quantum Information Processing Conference (QIP 2025).}
}

% \markboth{Journal of \LaTeX\ Class Files,~Vol.~1, No.~2, December~2023}%

% \IEEEpubid{0000--0000~\copyright~2023 IEEE}

\maketitle

\begin{abstract}
Port-based teleportation is a variant of quantum teleportation, where the receiver can choose one of the ports in his part of the entangled state shared with the sender, but cannot apply other recovery operations. 
We show that the optimal fidelity of deterministic port-based teleportation (dPBT) using $N=n+1$ ports to teleport a $d$-dimensional state is equivalent to the optimal fidelity of $d$-dimensional unitary estimation using $n$ calls of the input unitary operation.
From any given dPBT, we can explicitly construct the corresponding unitary estimation protocol achieving the same optimal fidelity, and \emph{vice versa}.
Using the obtained one-to-one correspondence between dPBT and unitary estimation, we derive the asymptotic optimal fidelity of port-based teleportation given by $1-O(d^4)N^{-2}\leq F \leq 1-\Omega(d^4)N^{-2}$, which improves the previously known result given by $1-O(d^5)N^{-2} \leq F \leq 1-\Omega(d^2) N^{-2}$.
We also show that the optimal fidelity of unitary estimation for the case $n\leq d-1$ is $F = {n+1 \over d^2}$, 
and this fidelity is equal to the optimal fidelity of unitary inversion with $n\leq d-1$ calls of the input unitary operation even if we allow indefinite causal order among the calls.
\end{abstract}


\section{Introduction}

\IEEEPARstart{Q}{uantum} teleportation \cite{bennett1993teleporting} is one of the fundamental protocols in quantum information, which is widely used in quantum communication \cite{bennett1992communication} and quantum computation \cite{gottesman1999demonstrating, raussendorf2001one}.
It also opens up a way to understand fundamental properties of quantum mechanics, such as quantum entanglement \cite{horodecki2009quantum}.
In the original protocol of quantum teleportation, Alice sends an unknown quantum state to Bob by using a shared entanglement with the Bell measurement on the unknown state and her share of the entangled state and classical communication. Depending on the outcome of Alice's measurement, Bob applies a Pauli operation on his part of the shared entangled state to recover her state.

Port-based teleportation (PBT)~\cite{ishizaka2008asymptotic, ishizaka2009quantum} is a variant of the teleportation protocol, where Alice and Bob share an entangled $2N$-qudit state, and each of Alice and Bob possesses $N$ qudits (called ports).  Alice performs a joint measurement on an unknown state and Alice's $N$ ports.
Instead of applying a Pauli operation, Bob can only choose a port from $N$ ports of the shared entanglement to recover Alice's state.
Since Bob's recovery operation commutes with parallel calls of the same quantum operation applied to Bob's ports, we can use port-based teleportation to apply a quantum operation $\mathcal{E}$ to a quantum state prepared after the application of $\mathcal{E}$ in the following way.
We first apply a quantum operation $\mathcal{E}$ in parallel to Bob's ports (storage) and use it for port-based teleportation of the quantum state $\rho$.
As a result, Bob obtains the quantum state $\mathcal{E}(\rho)$ (retrieval).
Such a task is called storage-and-retrieval \cite{sedlak2019optimal}, quantum learning \cite{bisio2010optimal}, or universal programming \cite{yang2020optimal,nielsen1997programmable}.
Beyond storage-and-retrieval, port-based teleportation has wide applications in quantum cryptography \cite{beigi2011simplified}, Bell nonlocality \cite{buhrman2016quantum}, holography \cite{may2019quantum, may2022complexity}, and higher-order quantum transformations \cite{quintino2019probabilistic, quintino2019reversing, quintino2022deterministic, yoshida2023universal}.
Port-based teleportation is also extended to the scenario where Alice's state is given by multi-qudit states (multi port-based teleportation) \cite{kopszak2021multiport, studzinski2022efficient} and continuous-variable states \cite{pereira2023continuous}.

Since the no-programming theorem prohibits a deterministic and exact implementation of storage-and-retrieval \cite{nielsen1997programmable}, port-based teleportation is also impossible in a deterministic and exact way.
This no-go theorem led researchers to consider probabilistic or approximate protocol \cite{ishizaka2008asymptotic, ishizaka2009quantum, ishizaka2015some, wang2016higher, studzinski2017port, mozrzymas2018optimal, mozrzymas2018simplified, christandl2021asymptotic, studzinski2022square, leditzky2022optimality, strelchuk2023minimal, grinko2023gelfand, grinko2023efficient, wills2023efficient, fei2023efficient, mozrzymas2024port, kim2024asymptotic}.
The former is called probabilistic port-based teleportation (pPBT), and the latter is called deterministic port-based teleportation (dPBT).
The optimal success probability of pPBT is explicitly given in Refs.~\cite{ishizaka2009quantum, studzinski2017port}, and it is known to be utilized to achieve the optimal success probability of storage-and-retrieval of unitary operation \cite{sedlak2019optimal}.
On the other hand, less is known for the optimal fidelity of dPBT, and no closed formula is known except for the qubit case \cite{ishizaka2008asymptotic}.
Its asymptotic form is shown in Ref.~\cite{christandl2021asymptotic}, but its lower and upper bounds do not coincide with each other.
Also, dPBT is known not to provide the optimal fidelity of storage-and-retrieval of unitary operation \cite{yang2020optimal, quintino2022deterministic}.
\begin{figure*}
    \centering
    \includegraphics[width=\linewidth]{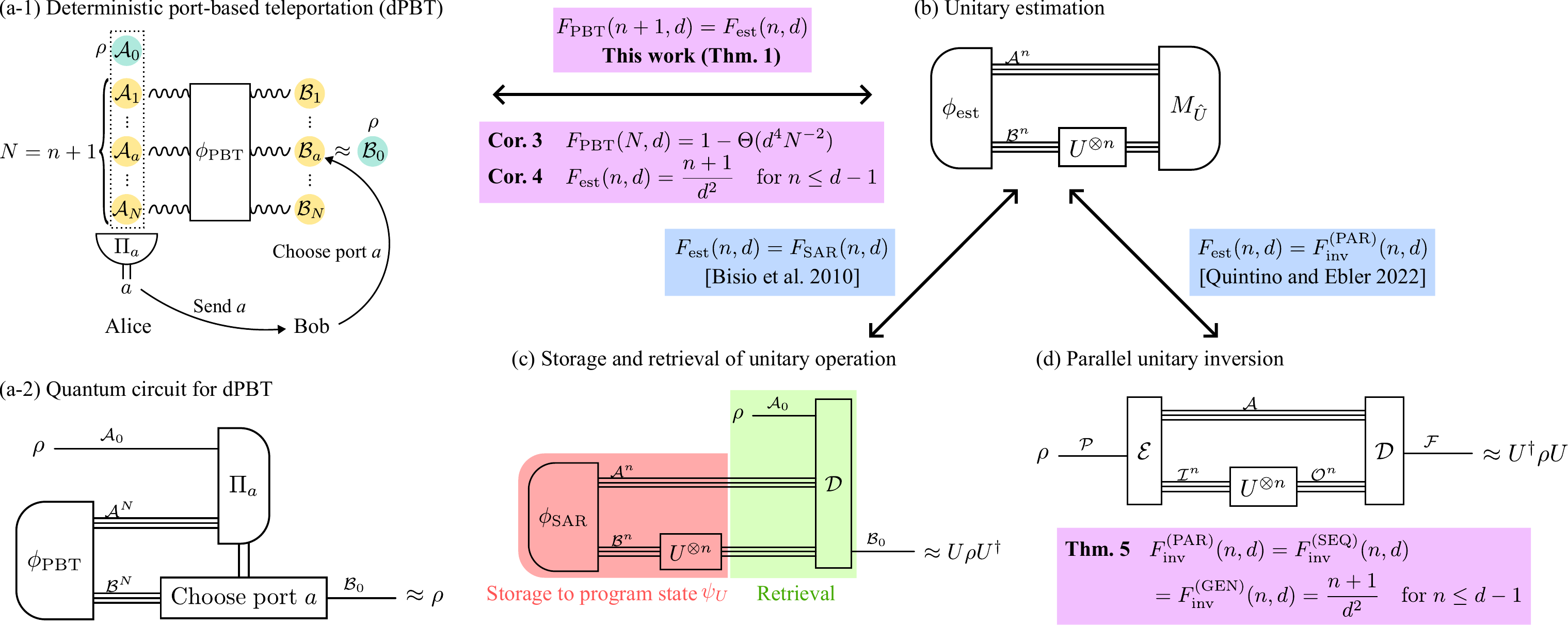}
    \caption{This work shows the one-to-one correspondence between deterministic port-based teleportation (dPBT) and unitary estimation (see Thm.~\ref{thm:equivalence}).
    Combining this result with Ref.~\cite{bisio2010optimal} and Ref.~\cite{quintino2022deterministic}, we also show the one-to-one correspondence with deterministic storage-and-retrieval (dSAR) of unitary operation and deterministic parallel unitary inversion.\\
    (a-1) In dPBT, Alice approximately teleports her unknown state $\rho$ by sending the measurement outcome $a$ of the joint measurement on the target state $\rho$ and half of a shared $2N$-qudit entangled state $\phi_\mathrm{PBT}$.
    Bob chooses the port $a$ of the other half to obtain a quantum state close to $\rho$.
    Its circuit description is shown in (a-2).
    By using the one-to-one correspondence, its asymptotic optimal fidelity is shown to be $F_\mathrm{PBT}(N,d) = 1-\Theta(d^4N^{-2})$ (see Cor.~\ref{cor:asymptotically_optimal_PBT}). \\
    (b) In unitary estimation, $n$ calls of unitary operator $U$ are applied in parallel to the resource state $\phi_\mathrm{est}$, and the output state is measured in a positive operator-valued measure (POVM) measurement $\{M_{\hat{U}} \dd \hat{U}\}_{\hat{U}}$ to obtain the estimated unitary $\hat{U}$.
    By using the one-to-one correspondence, its optimal fidelity is shown to be $F_\mathrm{est}(n,d) = {n+1 \over d^2}$ for $n\leq d-1$ (see Cor.~\ref{cor:optimal_unitary_estimation_small_n}).\\
    (c) In dSAR, $n$ calls of $U$ are applied into a resource state $\rho_\text{SAR}$ to obtain a program state $\psi_U$. In a later moment, the action of $U$ is retrieved by applying a quantum operation $\mathcal{D}$ on the joint system of the input state $\rho$ and the program state $\psi_U$. \\
    (d) In deterministic parallel unitary inversion, an operation $\mathcal{E}$ is applied before $n$ calls of an unknown unitary operator $U$, and an operation $\mathcal{D}$ is applied just after, in a way that the resulting composition is approximately $U^{-1}$.
    For $n\leq d-1$, the parallel protocol achieves the optimal fidelity even if we consider the most general protocol including the ones with indefinite causal order (see Thm.~\ref{thm:optimality_of_parallel_unitary_inversion}).}
    \label{fig:equivalence}
\end{figure*}

In this paper, we show a one-to-one correspondence of the dPBT with another widely explored task called the unitary estimation in the Bayesian framework \cite{holevo2011probabilistic, acin2001optimal, dariano2001using, fujiwara2001estimation, peres2002covariant, bagan2004entanglement, bagan2004quantum, ballester2004estimation, chiribella2004efficient, chiribella2005optimal, hayashi2006parallel, kahn2007fast, bisio2010optimal, yang2020optimal, haah2023query,holevo2011probabilistic}.
We show that the optimal fidelity of $d$-dimensional unitary estimation using $n$ calls of the input unitary operation is equivalent to that of dPBT using $N=n+1$ ports to teleport a $d$-dimensional state.
Our proof is made in a constructive way, once the optimal protocol for dPBT is given, we can explicitly construct the corresponding unitary estimation protocol and \emph{vice versa}. Combining our results with previously known the one-to-one correspondence between unitary estimation, deterministic storage-and-retrieval (dSAR)~\cite{bisio2010optimal} and deterministic parallel unitary inversion~\cite{quintino2022deterministic}, we then obtain a one-to-one correspondence between the four tasks (see Fig.~\ref{fig:equivalence}).
Note that this one-to-one correspondence does not imply that the optimal dPBT protocol provides the optimal fidelity of dSAR; rather, our result implies the opposite.
If we use the dPBT protocol to perform the dSAR using $n$ calls of the input unitary operation $U$, it provides the fidelity given by $F_\mathrm{PBT}(n,d)$, which is strictly smaller than $F_\mathrm{SAR}(n,d) = F_\mathrm{PBT}(n+1,d)$ from our one-to-one correspondence.

The one-to-one correspondence proved here allows us to translate previous results on the optimal fidelity of unitary estimation to that of dPBT and \emph{vice versa}. This allows us to prove that the asymptotic fidelity of optimal dPBT is precisely $F_\text{PBT}(N,d)=1-\Theta(d^4 N^{-2})$, from a corresponding result for unitary estimation~\cite{yang2020optimal, haah2023query}.
This result improves upon the previous result \cite{christandl2021asymptotic} and gives a tight scaling with respect to $d$.
Also, results from dPBT~\cite{mozrzymas2018optimal}, allow us to show that the fidelity of unitary estimation is given by ${n+1 \over d^2}$ for $n\leq d-1$.
Finally, we prove that the optimal fidelity of unitary inversion using $n$ calls of the input unitary operation is given by that of unitary estimation for $n\leq d-1$ even when considering adaptive circuits or protocols without a definite causal order \cite{hardy2007towards, oreshkov2012quantum, chiribella2013quantum}.

\section{One-to-one correspondence between dPBT and unitary estimation}

In dPBT, Alice wants to teleport an arbitrary qudit state $\rho\in\mcL(\mcA_0)$ to Bob using a shared entangled $2N$-qudit state $\phi_\mathrm{PBT}\in \mcL(\mcA^N \otimes \mcB^N)$, where $\mcL(\mcX)$ represents the set of linear operators on a Hilbert space $\mcX$, $\mcA^N$ and $\mcB^N$ are the joint Hilbert spaces defined by
$\mcA^N = \bigotimes_{i=1}^{N} \mcA_i$ and $\mcB^N = \bigotimes_{i=1}^{N} \mcB_i$ for $\mcA_0 \simeq \mcA_i \simeq  \mcB_i \simeq \CC^d$.
Alice measures the quantum state $\rho$ and half of the shared entangled state $\phi_\mathrm{PBT}$ with a positive operator-valued measure (POVM) measurement $\{\Pi_a\}_{a=1}^{N}$, and sends the measurement outcome $a$ to Bob.
Bob chooses the port $a$ from the other half of the shared entangled state $\phi_\mathrm{PBT}$, approximating Alice's state $\rho$.
The quantum state Bob obtains is given by $\Lambda(\rho)$, where $\Lambda: \mcL(\mcA_0) \to \mcL(\mcB_0)$ is the teleportation channel \cite{ishizaka2008asymptotic, horodecki1999general}
\begin{align}
\label{eq:def_teleportation_channel}
    \Lambda(\rho)\coloneqq \sum_{a=1}^{N} \Tr_{\mcA_0 \mcA^N \overline{\mcB_a}} [(\Pi_a \otimes \1_{\mcB^N})(\rho\otimes \phi_\mathrm{PBT})],
\end{align}
with $\overline{\mcB_a}\coloneqq \bigotimes_{i\neq a} \mcB_i$ and $\1$ being the identity operator.
The performance of the dPBT protocol may be evaluated in terms of channel fidelity~\cite{raginsky2001fidelity} between a quantum channel with map $\Lambda:\mathcal{L}(\CC^d)\to\mathcal{L}(\CC^d)$ and a unitary operator  $U:\CC^d\to\CC^d$, quantity defined by
\begin{align}
    f( U, \Lambda) \coloneqq {1\over d^2}\sum_{i} \abs{\Tr(K_i U^\dagger)}^2,
\end{align}
where $\{K_i\}_i$ is the set of Kraus operators of $\Lambda$ satisfying $\Lambda(\rho) = \sum_i K_i \rho K_i^\dagger$ \cite{nielsen2010quantum}. 
The performance of a dPBT protocol is then the fidelity between the teleportation channel of Eq.~\eqref{eq:def_teleportation_channel} and the identity channel represented by the identity operator $\1_d\in\mathcal{L}(\mathbb{C}^d)$, i.e, 
\begin{align}
    \label{eq:def_teleportation_fidelity}
    F_\mathrm{PBT}\coloneqq f( \1_d, \Lambda).
\end{align}

Unitary estimation is the task to obtain the classical description of the estimated unitary
$\hat{U}$ of the input unitary operator $U\in \SU(d)$, which can be called $n$ times.
We use the Bayesian framework \cite{holevo2011probabilistic} to evaluate the performance of unitary estimation, and we assume a Haar-random prior distribution of the input unitary operator.
The performance of an estimation protocol is given by its \emph{average fidelity}
\begin{align}
    \label{eq:average-case_fidelity_UE}
    F^{(\mathrm{ave})}_\mathrm{est}\coloneqq \int \dd U \int \dd \hat{U} p(\hat{U}|U) f( \hat{U}, \mcU),
\end{align}
where $p(\hat{U}|U)$ is the probability distribution of obtaining the estimator $\hat{U}$ given the input unitary operator $U$, $\dd U$ and $\dd \hat{U}$ are the Haar measures \cite{mele2024introduction} on $\SU(d)$, $f$ is the channel fidelity,  and $\mcU:\mathcal{L}(\CC^d)\to\mathcal{L}(\CC^d)$ is the unitary operation defined by $\mcU(\cdot)\coloneqq U(\cdot)U^\dagger$.
The optimal fidelity is shown to be implementable by a parallel covariant protocol \cite{bisio2010optimal}. 
In the  parallel covariant protocol [see Fig.~\ref{fig:equivalence}~(b)], the $n$-fold unitary operator $U^{\otimes n}$ is applied to the resource state $\phi_\mathrm{est} \in \mcL(\mcA^n \otimes \mcB^n)$, and the POVM measurement $\left\{M_{\hat{U}} \dd \hat{U}\right\}_{\hat{U}}$ is applied to obtain the estimated unitary $\hat{U}$, hence 
$p(\hat{U}|U)=\tr\left[M_{\hat{U}}\, (U^{\otimes n}\otimes \1_{\mcA^n})\phi_\mathrm{est}(U^{\otimes n}\otimes \1_{\mcA^n})^\dagger \right]$.
Also, in the covariant protocol, the average-case fidelity coincides with the worst-case fidelity~\cite{holevo2011probabilistic} $F_\mathrm{est}^{(\mathrm{wc})} \coloneqq \inf_U \int \dd \hat{U} p(\hat{U}|U) f( \hat{U}, \mcU)$.

In the following, we present a one-to-one correspondence between dPBT and unitary estimation.

\begin{Thm}
\label{thm:equivalence}
The optimal fidelity of deterministic port-based teleportation (dPBT) using $N=n+1$ ports to teleport a $d$-dimensional state [denoted by $F_\mathrm{PBT}(N, d)$] coincides with the optimal fidelity of unitary estimation using $n$ calls of an input $d$-dimensional unitary operation [denoted by $F_\mathrm{est}(n, d)$], i.e.,
\begin{align}
    F_\mathrm{PBT}(n+1, d) = F_\mathrm{est}(n,d)
\end{align}
holds.
In addition, given any dPBT protocol using $N=n+1$ ports, we can construct a unitary estimation protocol with $n$ calls attaining the fidelity more than or equal to the fidelity of the dPBT, and \emph{vice versa}.  
\end{Thm}
\begin{proof}[Proof sketch]
    We show this theorem by explicitly constructing a unitary estimation protocol from any given dPBT protocol and \emph{vice versa} (see Fig.~\ref{fig:equivalence_protocol}).
    We first convert any given dPBT and unitary estimation protocols into covariant protocol without decreasing the fidelity~\cite{mozrzymas2018optimal, leditzky2022optimality,chiribella2005optimal, bisio2010optimal}, as shown in (a)$\to$(b) and (c)$\to$(d) of Fig.~\ref{fig:equivalence_protocol} (see Appendixes~\ref{appendix_sec:covariant_PBT} and \ref{appendix_sec:covariant_unitary_estimation} for the details).
    The covariant protocol\footnote{A function $f$ is called covariant with respect to the unitary representations $U, V$ of the group $G$ if and only if $f(U_g \cdot U_g^{\dagger}) = V_g f(\cdot) V_g^\dagger$ holds for all $g\in G$ [see, e.g., Eq.~(1) in \cite{holevo1993note}].} are protocols with the resource states and the POVMs in specific forms that are covariant with respect to certain representations of the unitary group, see e.g., Def.~1 of Ref.~\cite{yang2020optimal}.
    The covariant dPBT and unitary estimation protocols are shown to be parametrized by non-negative weights $\vec{w} = \{w_\mu\}_\mu$ and $\vec{v} = \{v_\alpha\}_\alpha$ associated with irreducible representations of the unitary group $\SU(d)$, respectively, satisfying the normalization conditions given by
    \begin{align}
        \sum_{\mu} w_\mu^2 = \sum_{\alpha} v_\alpha^2 = 1.
    \end{align}
    In addition, the corresponding fidelities are given by
    \begin{align}
        F_\mathrm{PBT} &= \vec{w}^\sfT M_\mathrm{PBT}(N,d) \vec{w},\\
        F_\mathrm{est} &= \vec{v}^\sfT M_\mathrm{est}(n,d) \vec{v},
    \end{align}
    using the teleportation matrix $M_\mathrm{PBT}(N,d)$~\cite{mozrzymas2018optimal} and the estimation matrix $M_\mathrm{est}(n,d)$~\cite{bagan2004quantum,bagan2004entanglement,chiribella2005optimal,yang2020optimal}.
    These two matrices are shown to be related by introducing a matrix $R(n,d)$ as shown below:
    \begin{align}
        M_\mathrm{PBT}(N,d) &= {1\over d^2} R(N-1,d)^\sfT R(N-1,d),\\
        M_{\mathrm{est}}(n,d)&= {1\over d^2} R(n,d) R^\sfT(n,d).
    \end{align}
    Then, we present an explicit recipe to convert between covariant protocols for dPBT with $N=n+1$ calls and unitary estimation with $n$ calls keeping the fidelity, given by multiplying $R(n,d)$ or $R(n,d)^\sfT$ to the weight vectors $\vec{w}$ or $\vec{v}$, respectively, as shown in (b) $\leftrightarrow$ (d) of Fig.~\ref{fig:equivalence_protocol}.
    A detailed proof of Thm.~\ref{thm:equivalence} is presented in Appendix~\ref{appendix_sec:proof_equivalence}.
\end{proof}

\begin{figure*}
    \centering
    \includegraphics[width=.7\linewidth]{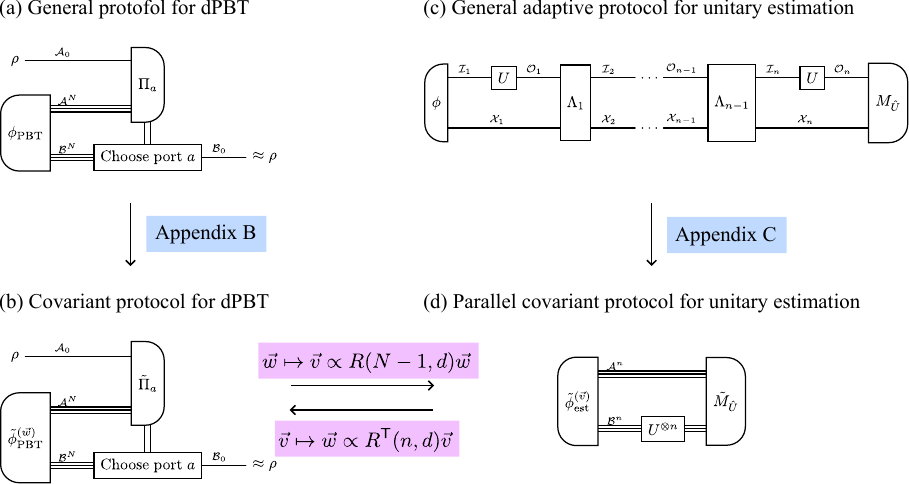}
    \caption{Explicit construction of a unitary estimation protocol from any given dPBT protocol and \emph{vice versa}.\\
    (a) General protocol for dPBT.\\
    (b) A covariant protocol for dPBT is constructed by using Eqs.~\eqref{eq:wmu}, \eqref{eq:covariant_resource_state} and \eqref{eq:srm} to the general protocol (a) \cite{mozrzymas2018optimal, leditzky2022optimality} (see the details in Appendix~\ref{appendix_sec:covariant_PBT}).\\
    (c) General adaptive protocol for unitary estimation.\\
    (d) A parallel covariant protocol for unitary estimation is constructed by using Eqs.~\eqref{eq:def_of_probe_state}, \eqref{eq:v_alpha} and \eqref{eq:covariant_POVM} to the general protocol (c) \cite{bisio2010optimal} (see the details in Appendix~\ref{appendix_sec:covariant_unitary_estimation}).\\
    This work shows a transformation between the covariant protocols for dPBT and unitary estimation given in Eqs.~\eqref{eq:PBT_to_untiary_estimation} and \eqref{eq:untiary_estimation_to_PBT}.
    Combining this transformation with the constructions $\text{(a)} \to \text{(b)}$ and $\text{(c)} \to \text{(d)}$, we obtain a transformation between the sets of general protocols for unitary estimation and dPBT.
    }
    \label{fig:equivalence_protocol}
\end{figure*}

\section{One-to-one correspondence between four different tasks}

In Ref.~\cite{bisio2010optimal}, the optimal unitary estimation using $n$ calls is shown to be equivalent to the optimal deterministic storage-and-retrieval (dSAR) protocol using $n$ calls, and in Ref.~\cite{quintino2022deterministic}, the optimal unitary estimation using $n$ calls is shown to be equivalent to the optimal deterministic parallel unitary inversion protocol using $n$ calls. Combining these with Thm.~\ref{thm:equivalence} proved in this work, we obtain the one-to-one correspondence among the optimal fidelity of these four tasks (see Fig.~\ref{fig:equivalence}). For completeness, we now describe the task of dSAR and parallel unitary inversion.

The dSAR (see Fig.~\ref{fig:equivalence}), also referred to as unitary learning, considers the problem of storing the usage of an arbitrary input unitary operation on some state in a way that the usage of this operation may be retrieved in a later moment. When the usage of the input operations is made in parallel, an assumption which can be made without loss in performance~\cite{bisio2010optimal}, the task is described as follows. An arbitrary input unitary operator $U\in\SU(d)$ is called $n$ times on a resource state $\phi_\mathrm{SAR}\in\mathcal{L}(\mathcal{A}\otimes\mathcal{I}^{n})$, where $\mathcal{A}$ is an auxiliary space and the unitary operation $\mcU^{\otimes n}$ maps an input space $\mathcal{I}^{n}$ into an output space $\mathcal{O}^{n}$ to obtain a program state 
\begin{align}
    \psi_U\coloneqq\left(\1 \otimes U^{\otimes n}\right)  \phi_\mathrm{SAR} \left(\1 \otimes U^{\otimes n}\right)^\dagger\in \mathcal{L}(\mathcal{A}\otimes\mathcal{O}^{n}).
\end{align}
Then, an arbitrary input state $\rho\in\mathcal{L}(\mathcal{P})$ is subjected to a quantum operation $\mathcal{D}:\mathcal{L}(\mathcal{P}\otimes\mathcal{A}\otimes\mathcal{O}^{n})\to \mathcal{L}(\mathcal{F})$, often referred to as decoder \cite{chiribella2008transforming}, on the joint system of $\rho$ and $\psi_U$ to obtain the output state   $\Lambda_U(\rho)\coloneqq \mathcal{D}(\rho\otimes \psi_U)\in\mathcal{F}$, which we desired to be approximately the state $U\rho U^\dagger$.
The performance of dSAR is evaluated by the average-case fidelity given by
 $   F_\mathrm{SAR}^{(\mathrm{ave})}\coloneqq \int \dd U f(U, \Lambda_U)$,
whose optimal value is always attainable by a covariant protocol, and respects~\cite{yang2020optimal} 
 $F_\mathrm{SAR}^{(\mathrm{ave})} = 1-{1\over 2} \|\Lambda_U-\mathcal{U}\|_\diamond$ for the covariant protocol,
where $\|\cdot \|_\diamond$ is the diamond norm \cite{kitaev1997quantum, watrous2005notes}.
Reference~\cite{bisio2010optimal} shows that optimal dSAR protocol can always be implemented by using an estimation protocol, i.e., $F_\mathrm{est}(n,d) = F_\mathrm{SAR}(n,d)$ holds for optimal protocols. From our results, it implies that $F_\mathrm{SAR}(n,d)=F_\mathrm{PBT}(n+1,d)$.

In the task of transforming unitary operations, one is allowed to make $n$ calls of unitary operation $\mathcal{U}(\cdot)=U(\cdot) U^\dagger:\mathcal{L}(\mathcal{I}_i)\to\mathcal{L}(\mathcal{O}_i)$, where $i\in\{1,\ldots,n\}$, $U\in \mathrm{SU}(d)$ and we aim to design a quantum circuit which aims to output an operation $\mathcal{U}_g(\cdot)\coloneqq g(U)\cdot g(U)^\dagger$, where $g: \mathrm{SU}(d) \to \mathrm{SU}(d)$.
When the $n$ calls of the input operation are made in parallel (see Fig.~\ref{fig:equivalence}), the task consists in finding an encoder operation $\mathcal{E}:\mathcal{L}(\mathcal{P})\to \mathcal{L}(\mathcal{A}\otimes\mathcal{I}^{n})$, where $\mathcal{A}$ is an auxiliary space and $\mathcal{I}^{n}\coloneqq\bigotimes_{i=1}^n \mathcal{I}_i$, and a decoder operation $\mathcal{D}:\mathcal{L}(\mathcal{A}\otimes \mathcal{O}^{n})\to\mathcal{F}$, where $\mathcal{O}^{n}\coloneqq \bigotimes_{i=1}^n \mathcal{O}_i$, such that the resulting output operation 
\begin{align}
    \Lambda_U: =\mathcal{D} \circ \left(\1_{\mathcal{A}} \otimes \mathcal{U}^{\otimes n} \right) \circ \mathcal{E},
\end{align}
is approximately the operation $\mathcal{U}_g$, with $\1_{\mathcal{A}}$ here being the identity map in the auxiliary space. Similarly to dSAR, we evaluate the performance of a given protocol in terms of its average-case fidelity, 
$F^\mathrm{(PAR)}_{\mathrm{g}}\coloneqq \int \dd U f(g(U), \Lambda_U)$.

Unitary estimation protocols can always be used to transform $n$ calls of a unitary $U$ into $g(U)$ with some fidelity $F_\text{est}^{(\mathrm{ave})}$, for that one may simply estimate the unitary $U$ as $\hat{U}$ and then to output $g(\hat{U})$. In Ref.~\cite{quintino2022deterministic}, it is shown that if the target function $g:\SU(d)\to\SU(d)$ is antihomomorphic, i.e., $g(UV) = g(V)g(U)$ holds for all $U, V\in \mathrm{SU}(d)$, unitary estimation attains optimal performance for parallel unitary transformations. Since unitary inversion $g(U)=U^{-1}$ and unitary transposition $g(U)=U^T$ are antihomomorphisms, it follows from Thm.~\ref{thm:equivalence} that the optimal parallel unitary inversion and parallel unitary transposition with $n$ calls is precisely $F_{\text{PBT}}(n+1,d)$.

\section{Applications}

Theorem~\ref{thm:equivalence} allows us to translate results on unitary estimation to results on dPBT, and \emph{vice versa}.
We now illustrate some of these applications.
This section utilizes the big-O notation $O(\cdot)$, $\Omega(\cdot)$ and $\Theta(\cdot)$, defined as follows \cite{arora2009computational}:
\begin{align}
    f(x) = O(g(x)) &\Leftrightarrow \limsup_{x\to \infty}{\abs{f(x) \over g(x)}} <\infty,\\
    f(x) = \Omega(g(x)) &\Leftrightarrow g(x) = O(f(x)),\\
    f(x) = \Theta(g(x)) &\Leftrightarrow f(x)=O(g(x)) \text{ and } f(x) = \Omega(g(x)).
\end{align}
Intuitively, $O(\cdot)$ represents a lower bound, $\Omega(\cdot)$ represents an upper bound, and $\Theta(\cdot)$ represents a tight bound.

Reference \cite{yang2020optimal} shows a unitary estimation protocol asymptotically achieving $F_\text{est}(n,d)\geq 1-O(d^4) n^{-2}$.
As shown in the following lemma, this lower bound is tight:

\begin{Lem}
\label{lemma:asymptotically_optimal_unitary_estimation}
The optimal fidelity of unitary estimation using $n$ calls of an input $d$-dimensional unitary operation is upper bounded as
\begin{align}
    F_\text{est}(n,d) \leq 1-\Omega(d^4) n^{-2}.
\end{align}
\end{Lem}
\begin{proof}[Proof sketch]
Reference~\cite{haah2023query} shows a lower bound on the query complexity $n$ to achieve a diamond-norm error $\epsilon$ for unitary estimation, given by $n = \Omega(d^2)/\epsilon$.
The proof is based on constructing $\exp(\Omega(d^2))$ unitary channels that are $\epsilon$-distant from each other and $O(\epsilon)$-close to the identity channel in terms of the diamond norm.
We adapt their proof for the Bures distance~\cite{zyczkowski2006geometry} induced from the channel fidelity instead of the diamond norm to show this lemma.
A detailed proof is presented in Appendix~\ref{appendix_sec:asymptotically_optimal_unitary_estimation}.
\end{proof}

Thanks to Thm.~\ref{thm:equivalence}, we have the following corollary.
\begin{Cor} The asymptotic fidelity of optimal dPBT is given by
\label{cor:asymptotically_optimal_PBT}
\begin{align}
    1-O(d^4)N^{-2} \leq F_\mathrm{PBT}(N,d) \leq 1-\Omega(d^4) N^{-2}.
\end{align}
\end{Cor}
Corollary~\ref{cor:asymptotically_optimal_PBT} may be compared with previous results~\cite{christandl2021asymptotic} which proved that
\begin{align}
    1-O(d^5) N^{-2} \leq F_\mathrm{PBT}(N,d) \leq 1-\Omega(d^2) N^{-2}.
\end{align}
In particular, we can explicitly construct the dPBT protocol whose fidelity scales with $1-O(d^4)N^{-2}$ by combining the unitary estimation protocol in Ref.~\cite{yang2020optimal} with the obtained one-to-one correspondence, which improves over the protocol exhibiting the fidelity $1-O(d^5)N^{-2}$ shown in \cite{christandl2021asymptotic}.
This result also gives a partial answer to an open problem raised in Ref.~\cite{christandl2021asymptotic} to determine
$h(d)\coloneqq \lim_{N\to \infty} N^2 [1-F_\mathrm{PBT}(N,d)]$, which is shown to be $h(d) = \Theta(d^4)$ by Cor.~\ref{cor:asymptotically_optimal_PBT}.

Reference~\cite{mozrzymas2018optimal} shows that, when 
$N\leq d$, optimal dPBT is given by
$    F_\mathrm{PBT}(N,d) = {N\over d^2}$. Hence, from Thm.~\ref{thm:equivalence}, we obtain the following corollary.
\begin{Cor}
\label{cor:optimal_unitary_estimation_small_n}
When the number of the calls of the input unitary operation, denoted by $n$, satisfies $n\leq d-1$, the optimal fidelity of unitary estimation is given by
\begin{align}
    F_{\mathrm{est}}(n,d) = {n+1\over d^2}.
\end{align}
\end{Cor}
Also, in Appendix~\ref{appendix_sec:optimal_unitary_estimation}, we show that the optimal resource state for dPBT and unitary estimation are the same for the case $N=n+1\leq d$.
Hence, the optimal resource state is useful as a universal resource state for dPBT and unitary estimation.

\section{Optimal unitary inversion for $n\leq d-1$}

As previously mentioned, the task of transforming $n$ calls of an arbitrary unitary operation into its inverse is attainable by an estimation strategy when only parallel uses of the input operation are allowed. However, already when $n=2$ calls are allowed, the performance of sequential circuits for qubit unitary inversion outperforms parallel ones, and the performance of processes without a definite causal order~\cite{hardy2007towards, oreshkov2012quantum, chiribella2013quantum} outperforms sequential ones~\cite{quintino2022deterministic}.
Also, when $n=4$ calls are allowed, there exists a sequential circuit which transforms an arbitrary qubit-unitary operation into its inverse operation with fidelity one~\cite{yoshida2023reversing}, which was recently extended to an arbitrary dimension using $n=O(d^2)$ calls~\cite{Chen2024Reversing} (see also~\cite{Odake2024LowerBound}). 

In Appendix~\ref{appendix_sec:feasilibity_proof}, we show an upper bound of the fidelity of unitary inversion given by ${n+1\over d^2}$ for $n\leq d-1$ even if we consider the most general protocols including the ones with indefinite causal order.
Corollary~\ref{cor:optimal_unitary_estimation_small_n} implies that this bound is achievable by an estimation-based strategy.
More precisely, we prove the following theorem. 
\begin{Thm}
\label{thm:optimality_of_parallel_unitary_inversion}
  When $n\leq d-1$, the optimal fidelity of $d$-dimensional unitary inversion with $n$ calls of a unitary operation is given by
  \begin{align}
    F_\mathrm{inv}^{(\mathrm{PAR})}(n,d)=F_\mathrm{inv}^{(\mathrm{SEQ})}(n,d)=F_\mathrm{inv}^{(\mathrm{GEN})}(n,d)=\frac{n+1}{d^2},\label{eq:optimal_fidelity_n_less_than_d}
  \end{align}
  where $\mathrm{PAR}$, $\mathrm{SEQ}$, and $\mathrm{GEN}$ refer to parallel, sequential, and general protocols including the ones with indefinite causal order, respectively.
\end{Thm}

\section{Conclusion}

In this paper, we show the equivalence between the optimal fidelity of dPBT for $N=n+1$ and unitary estimation with $n$ calls of the unitary operation by giving an explicit transformation between the optimal protocols.
Combining this result with the previous results on the equivalence of the optimal fidelity between unitary estimation and dSAR~\cite{bisio2010optimal}, and also between parallel unitary inversion and unitary estimation~\cite{quintino2022deterministic},  we obtain a one-to-one correspondence among four important quantum information tasks.
Our correspondence result has the potential to accelerate the research on dPBT and unitary estimation since obtaining the optimal fidelity of dPBT in any setting, i.e., asymptotic or finite, can be immediately translated into that of unitary estimation, and \emph{vice versa}.
Here, we use this one-to-one correspondence to derive the asymptotically optimal dPBT protocol, a non-trivial result which improves the scaling of the fidelity derived in a previous work \cite{christandl2021asymptotic}.
In a related topic, we prove that when only $n\leq d-1$ calls are available, deterministic unitary inversion can always be obtained by an estimation protocol, showing that sequential circuits and strategies without definite causal order are not useful in the regime of $n\leq d-1$ calls.

Our proof of the one-to-one correspondence is based on the covariant forms of dPBT and unitary estimation, which is shown to be optimal in terms of fidelity.
On the other hand, Ref.~\cite{haah2023query} shows that a non-covariant adaptive protocol for unitary estimation with no auxiliary system achieves the asymptotically optimal performance in terms of the diamond-norm distance.
Since the transformation from a general protocol to a covariant protocol may induce the increase of the resources other than the number of calls $n$, e.g., the number of required qudits or non-Clifford gates to run the protocols, and the amount of entanglement in the resource states, the corresponding dPBT protocol constructed by our methods may not be implemented efficiently for some other notions of efficiency.
We leave it an open problem to construct a conversion between the fidelities of unitary estimation and dPBT in terms of resources other than the number of calls.

\section*{Acknowledgments}
We acknowledge Yuxiang Yang for fruitful discussions and Jisho Miyazaki and Wataru Yokojima for comments on the manuscript.
S.Y. acknowledges support by Japan Society for the Promotion of Science (JSPS) KAKENHI Grant Number 23KJ0734, FoPM, WINGS Program, the University of Tokyo, and DAIKIN Fellowship Program, the University of Tokyo.
Y.K. acknowledges support by FY2023 Study and Visit Abroad Program (SVAP 2023), the School of Science, the University of Tokyo.
M.S. acknowledges support by grant Sonata 16, UMO-2020/39/D/ST2/01234 from the Polish National Science Centre.
M.M. acknowledges support by MEXT Quantum Leap Flagship Program (MEXT QLEAP) JPMXS0118069605, JPMXS0120351339, JSPS KAKENHI Grant Number 23K21643 and IBM Quantum.
The quantum circuits shown in this paper are drawn using {\sc quantikz} \cite{kay2018tutorial}.

\appendix
Appendix~\ref{appendix_sec:shur_weyl_duality} reviews mathematical tools based on
the Young diagrams and the Schur-Weyl duality.
Appendix~\ref{appendix_sec:covariant_PBT} shows a construction of a covariant protocol for deterministic port-based teleportation (dPBT) from any given dPBT protocol, following Refs.~\cite{ishizaka2008asymptotic, mozrzymas2018optimal, leditzky2022optimality}.
The teleportation fidelity of the obtained covariant protocol is also shown (see Lem.~\ref{lem:covariant_PBT}).
Appendix~\ref{appendix_sec:covariant_unitary_estimation} shows a construction of a parallel covariant protocol, for unitary estimation from any given adaptive protocol \cite{chiribella2005optimal, bisio2010optimal}.
The fidelity of the obtained covariant protocol is also shown (see Lem.~\ref{lem:covariant_unitary_estimation}).
Appendix~\ref{appendix_sec:proof_equivalence} proves Thm.~\ref{thm:equivalence}, stating the one-to-one correspondence between dPBT and unitary estimation, by constructing a transformation between covariant protocols for dPBT and unitary estimation.
Appendix~\ref{appendix_sec:asymptotically_optimal_unitary_estimation} shows the asymptotically optimal fidelity of unitary estimation using the previous results from Refs.~\cite{yang2020optimal, haah2023query} (Lem.~\ref{lemma:asymptotically_optimal_unitary_estimation}).
Appendix~\ref{appendix_sec:optimal_unitary_estimation} shows an optimal protocol for unitary estimation for $n\leq d-1$ corresponding to Cor.~\ref{cor:optimal_unitary_estimation_small_n}.
Appendix~\ref{appendix_sec:feasilibity_proof} proves Thm.~\ref{thm:optimality_of_parallel_unitary_inversion}, stating that estimation-based protocol for unitary inversion is optimal for $n\leq d-1$ even when we consider indefinite causal order \cite{hardy2007towards, oreshkov2012quantum, chiribella2013quantum}.

\appendices
\section{Review of the Schur-Weyl duality, Young diagrams, and the Young-Yamanouchi basis}
\label{appendix_sec:shur_weyl_duality}
This section reviews mathematical tools based on the Young diagrams and the Schur-Weyl duality.
We suggest the standard textbooks, e.g. Refs.~\cite{fulton1997young, georgi2000lie, ceccherini2010representation}, for more detailed reviews.

A Young diagram is defined in association with a partition of $n$ for a non-negative integer $n$.
A partition of $n$ is a non-decreasing sequence of positive integers $(\alpha_1, \ldots, \alpha_d)$ for any non-negative integer $d$ with $\alpha_1\geq \cdots \geq \alpha_d > 0$ which sums up to $\sum_{i=1}^{d} \alpha_i = n$.
A partition of $n$ can be represented as a Young diagram with $n$ boxes, which has $\alpha_i$ boxes in the $i$-th row.
For instance, $(2, 1)$ can be represented as a Young diagram given by
\begin{align}
\label{eq:example_young_diagram}
    \alpha = \ydiagram{2,1}.
\end{align}
We call $d$, the number of positive integers in the sequence associated with the Young diagram $\alpha$, to be the depth of each Young diagram.
We denote the set of Young diagrams with $n$ boxes and at most depth $d$ by $\young{d}{n}$, i.e., 
\begin{align}
\label{eq:def_Ydn}
    \young{d}{n}\coloneqq \left\{(\alpha_1, \ldots, \alpha_{d'}) \; \middle| \; \substack{0\leq d'\leq d, \alpha_1\geq \cdots \geq \alpha_{d'}>0, \\\sum_{i=1}^{d'}\alpha_i = n}\right\}.
\end{align}
For a Young diagram $\alpha$, we define the following two sets:
\begin{align}
\label{eq:add_box}
    &\alpha+\square \coloneqq \big\{(\alpha_1, \ldots, \alpha_i+1, \ldots, \alpha_d)\;\big|\;\substack{i\in\{1, \ldots, d\} \\\text{s.t. } \alpha_{i-1}\geq \alpha_i+1}\big\} \nonumber\\
    &\hspace{120pt} \cup\{(\alpha_1,\ldots, \alpha_d, 1)\},\\
\label{eq:remove_box}
    &\alpha-\square \nonumber\\
    &\coloneqq 
    \begin{cases}
        \big\{(\alpha_1, \ldots, \alpha_i-1, \ldots, \alpha_{d})\;\big|\;\substack{i\in\{1, \ldots, d\} \\\text{s.t. } \alpha_{i}-1\geq \alpha_{i+1}}\big\} & (\alpha_d >1)\\
        \big\{(\alpha_1, \ldots, \alpha_i-1, \ldots, \alpha_{d})\;\big|\;\substack{i\in\{1, \ldots, d-1\} \\\text{s.t. } \alpha_{i}-1\geq \alpha_{i+1}}\big\} \nonumber\\
        \hspace{90pt} \cup \{(\alpha_1, \ldots, \alpha_{d-1})\} & (\alpha_d = 1)
    \end{cases}
    ,
\end{align}
where $\alpha_0$ and $\alpha_{d+1}$ are taken to be $\alpha_0 = \infty$ and $\alpha_{d+1} = 0$, and $\alpha+\square (\alpha-\square)$ represents the set of Young diagrams adding a box to (removing a box from) $\alpha$.
For instance, for the Young diagram $\alpha$ given in Eq.~\eqref{eq:example_young_diagram}, $\alpha+\square$ and $\alpha-\square$ are given by
\begin{align}
    \alpha+\square &= \left\{ \ydiagram{3,1}, \ydiagram{2,2}, \ydiagram{2,1,1} \right\},\\
    \alpha-\square &= \left\{ \ydiagram{1,1}, \ydiagram{2} \right\}.
\end{align}
For a Young diagram $\alpha$ with $n$ boxes, we consider a sequence of Young diagrams given by
\begin{align}
\label{eq:sequence_of_young_diagram}
    (\alpha^{(0)} = \emptyset, \alpha^{(1)}, \ldots, \alpha^{(n-1)}, \alpha^{(n)} = \alpha),
\end{align}
where $\emptyset$ is the Young diagram with zero boxes and $\alpha^{(i)}$ satisfies $\alpha^{(i+1)}\in \alpha^{(i)}+\square$.
A sequence of Young diagrams can be represented by a standard Young tableau, which is given by filling $i$ on the box of the Young diagram $\alpha$ that is added when we obtain $\alpha^{(i)}$ from $\alpha^{(i-1)}$.
The Young diagram to be filled with numbers is called the frame of a corresponding standard tableau.
For instance, a sequence
\begin{align}
    \left(\emptyset, \ydiagram{1}, \ydiagram{2}, \ydiagram{2,1}\right)
\end{align}
is represented by a standard tableau
\begin{align}
    \ytableaushort{1 2, 3},
\end{align}
with the frame
\begin{align}
    \ydiagram{2,1}.
\end{align}
The number of the standard tableaux with frame $\alpha$, denoted by $m_\alpha$, is given by the hook-length formula:
\begin{align}
    \label{eq:hook-length}
    m_\alpha &= {n! \over \prod_{(i,j)\in \alpha} \mathrm{hook}_\lambda (i,j)}.
\end{align}
Here, $(i,j)$ represents a box in the $i$-th row and the $j$-th column in the Young diagram $\alpha$, $\mathrm{hook}_\lambda (i,j)$ is defined by
\begin{align}
    \mathrm{hook}_\lambda (i,j)\coloneqq \alpha_i + \alpha_j' - i - j + 1,
\end{align}
and $\alpha_i$ and $\alpha'_j$ are the numbers of boxes in the $i$-th row and the $j$-th column of $\alpha$, respectively.
We denote the set of the standard tableaux with frame $\alpha$ by
\begin{align}
\label{eq:stab}
    \STab(\alpha) \coloneqq \{s^{\alpha}_1, \ldots s^{\alpha}_{m_\alpha}\},
\end{align}
where each standard tableau is indexed by $a\in\{1, \ldots, m_{\alpha}\}$ and the $a$-th standard tableau is denoted by $s^{\alpha}_a$.
For a Young diagrams $\alpha$ and a standard tableau $s^\alpha_a$ with frame $\alpha$ associated with a sequence \eqref{eq:sequence_of_young_diagram},
a standard tableau with frame $\mu\in\alpha+\square$ is referred to be obtained by adding a box \fbox{$n+1$} to $s^{\alpha}_{a}$, when its associated sequence is given by
\begin{align}
    (\alpha^{(0)} = \emptyset, \alpha^{(1)}, \ldots, \alpha^{(n-1)}, \alpha^{(n)} = \alpha, \alpha^{(n+1)} = \mu).
\end{align}
We denote the index of thus obtained standard tableau with frame $\mu$ by $a^\alpha_{\mu}\in\{1, \ldots, m_\mu\}$ so that the standard tableau itself is represented by $s^{\mu}_{a^\alpha_{\mu}}$.
We also call $s^{\alpha}_a$ to be the standard tableau obtained by removing a box \fbox{$n+1$} from $s^{\mu}_{a^\alpha_\mu}$.
For instance, in the case of
\begin{align}
    \alpha = \ydiagram{2,1},\quad s^\alpha_a = \ytableaushort{1 2, 3},\quad \\
    \alpha+\square = \left\{ \mu = \ydiagram{3,1},\quad \nu = \ydiagram{2,2},\quad  \lambda = \ydiagram{2,1,1}\right\},
\end{align}
the standard tableaux obtained by adding a box \fbox{$4$} to $s^{\alpha}_{a}$ are given by
\begin{align}
    s^{\mu}_{a^\alpha_{\mu}} = \ytableaushort{1 2 4, 3},\quad s^{\nu}_{a^\alpha_\nu} = \ytableaushort{1 2, 3 4},\quad s^{{ \lambda}}_{a^\alpha_\lambda} = \ytableaushort{1 2, 3, 4}.
\end{align}

We consider the following representations on $(\CC^d)^{\otimes {n} }$ of the special unitary group $\SU(d)$ and the symmetric group $\mfS_n$:
\begin{align}
\label{eq:representation_unitary}
    &\SU(d)\ni U \mapsto U^{\otimes n} \in \mcL(\CC^d)^{\otimes { n} },\\
    &\mfS_{n} \ni \sigma \mapsto V_{\sigma} \in \mcL(\CC^d)^{\otimes {n} },
\end{align}
where $\mcL(\mcX)$ is the set of linear operators on a Hilbert space $\mcX$, and $V_{\sigma}$ is a permutation operator defined as
\begin{align}
\label{eq:def_permutation_operator}
    V_{\sigma}\ket{i_1\cdots i_{n}} \coloneqq \ket{i_{\sigma^{-1}(1)}\cdots i_{\sigma^{-1}(n)}}
\end{align}
for the computational basis $\{\ket{i}\}_{i=1}^{d}$ of $\CC^d$.
Then, these representations are decomposed simultaneously as follows:
\begin{align}
    (\CC^d)^{\otimes { n}} &= \bigoplus_{\alpha\in\young{d}{{n}}} \mcU_{\alpha} \otimes \mcS_{\alpha},\label{eq:decomp_space}\\
    U^{\otimes {n}} &= \bigoplus_{\alpha \in \young{d}{{ n}}} U_{\alpha} \otimes \1_{\mcS_{\alpha}},\\
    V_{\sigma} &= \bigoplus_{\alpha \in\young{d}{{ n}}} \1_{\mcU_{\alpha}} \otimes \sigma_{\alpha},\label{eq:def_sigma_mu}
\end{align}
where $\alpha$ runs in the set $\young{d}{n}$ defined in Eq.~\eqref{eq:def_Ydn}, $\SU(d)\ni U\mapsto U_{\alpha}\in \mcL(\mcU_{\alpha})$  is an irreducible representation of $\SU(d)$, and $\mfS_{n}\ni \sigma \mapsto \sigma_{\alpha}\in \mcL(\mcS_{\alpha})$ is an irreducible representation of $\mfS_n$.
This relation shows that any operator commuting with $U^{\otimes n}$ for all $U\in\SU(d)$ can be written as a linear combination of $\{V_\sigma\}_{\sigma\in\mfS_n}$, which is called the Schur-Weyl duality.
The orthogonal projector onto $\mcU_\alpha \otimes \mcS_\alpha$ is called the Young projector \cite{studzinski2022efficient}, denoted by $P_\alpha$.
The dimension of $\mcU_\alpha$, denoted by $d_\alpha$, is given by the hook-content formula:
\begin{align}
\label{eq:hook-content}
    d_\alpha &= \prod_{(i,j)\in \alpha} {d+j-i \over \mathrm{hook}_\lambda (i,j)}.
\end{align}
The dimension of $\mcS_\alpha$ is given by $m_\alpha$, which is the number of the standard tableaux with frame $\alpha$ given in Eq.~\eqref{eq:hook-length}.
The dimension $m_\alpha$ satisfies the following Lemma.

\begin{Lem}
\label{lem:symmetric_group_induced_irrep_dimension}
    For any given Young diagram $\alpha$ with $n$ boxes,
    \begin{align}
        \sum_{\mu\in\alpha+\square} m_\mu = (n+1) m_\alpha
    \end{align}
    holds.
\end{Lem}
\begin{proof}
    Since the irreducible representations $\alpha$ and $\mu\in\alpha+\square$ are related by
    \begin{align}
        \mathrm{Ind}_{\mfS_n}^{\mfS_{n+1}} \mcS_\alpha = \bigoplus_{\mu\in\alpha+\square} \mcS_\mu,
    \end{align}
    where $\text{Ind}_H^G \pi$ represents the induced representation for a finite group $G$, its subgroup $H$ and a representation $\pi$ of the group $H$.
    Then, we obtain
    \begin{align}
        \sum_{\mu\in\alpha+\square} m_\mu &= \dim \bigoplus_{\mu\in\alpha+\square} \mcS_\mu = \dim \mathrm{Ind}_{\mfS_n}^{\mfS_{n+1}} \mcS_\alpha \\
        &= {|\mfS_{n+1}| \over |\mfS_n|} \dim \mcS_\alpha = (n+1)m_\alpha.
    \end{align}
\end{proof}

Due to Schur's lemma, any operator commuting with $U^{\otimes n}$ can be written as a linear combination of the operators $E^{\alpha}_{ab}$ defined by \cite{studzinski2022efficient}
\begin{align}
    E^{\alpha}_{ab}\coloneqq \1_{\mcU_\alpha} \otimes \ketbra{\alpha, a}{\alpha, b}_{\mcS_{\alpha}},\label{eq:def_E}
\end{align}
for $a,b\in \range{m_{\alpha}}$, where $\{\ket{\alpha, a}\}$ is an orthonormal basis of $\mcS_{\alpha}$.
In particular, we take the Young-Yamanouchi basis \cite{young1931quantitative, yamanouchi1937construction} (or Young's orthogonal form \cite{ceccherini2010representation}) of $\mcS_{\alpha}$.
Each element in the Young-Yamanouchi basis $\{\ket{\alpha, a}\}$ is associated with the standard tableaux $s^{\alpha}_a$ in the set $\STab(\alpha)$ [see Eq.~\eqref{eq:stab}].
The Young-Yamanouchi basis is a subgroup-adapted basis, i.e., the action of $\mfS_n\subset \mfS_{n+1}$ on the Young-Yamanouchi basis $\{\ket{\mu,i}\}$ for $\mu\in\young{d}{n+1}$ is unitarily equivalent to the action of $\mfS_n$ on the Young-Yamanouchi basis $\{\ket{\alpha, a}\}$ for $\alpha\in\young{d}{n}$ as follows:
\begin{align}
    \bra{\mu, i} \sigma_\mu \ket{\mu, j} = \delta_{\alpha, \beta} \bra{\alpha, a}\sigma_\alpha \ket{\beta, b}\;\;\;\forall \sigma\in\mfS_n, \label{eq:subgroup_adapted}
\end{align}
where $s^{\alpha}_a$ and $s^{\beta}_b$ are the standard tableaux obtained by removing a box \fbox{$n+1$} from $s^{\mu}_i$ and $s^{\mu}_j$, respectively.

The basis $\{E^\mu_{ij}\}$ satisfies the following Lemmas (see, e.g.,  Refs.~\cite{studzinski2022efficient,yoshida2023reversing} for the proofs):

\begin{Lem}
  The basis $\{E^{\mu}_{ij}\}$ satisfies
  \begin{align}
      (E^{\mu}_{ij})^*&=E^{\mu}_{ij},\label{eq:lem2_1}\\
      \Tr E^{\mu}_{ij} &= d_{\mu} \delta_{i, j},\label{eq:lem2_2}\\
      E^{\mu}_{ij}E^{\nu}_{kl} &= \delta_{\mu, \nu} \delta_{j, k} E^{\mu}_{il},\label{eq:lem2_3}
  \end{align}
  where $X^*$ is the complex conjugate of $X$ in the computational basis and $\delta_{i, j}$ is the Kronecker delta defined as $\delta_{i, i}=1$ and $\delta_{i, j}=0$ for $i\neq j$.
\end{Lem}
  
\begin{Lem}
  Let $a^\alpha_{\mu}$ be the index of a standard tableau $s^{\mu}_{a^\alpha_{\mu}}$ obtained by adding a box \fbox{$n+1$} to a standard tableau $s^{\alpha}_a$ for $\alpha\in\young{d}{n}$ and $\mu\in\alpha+\square$. Then, $E^{\alpha}_{ab}\otimes \1_d$ can be written as
  \begin{align}
      E^{\alpha}_{ab} \otimes \1_d = \sum_{\mu\in\alpha+\square} E^{\mu}_{a^\alpha_{\mu} b^\alpha_{\mu}},
  \end{align}
  where $\1_d$ is the identity operator on $\CC^d$.
\end{Lem}
  
\begin{Lem}
  Let $s^{\alpha}_a$ and $s^{\beta}_b$ be the standard tableaux obtained by removing a box \fbox{$n+1$} from $s^{\mu}_i$ and $s^{\mu}_j$, respectively, for $\mu, \nu\in\young{d}{n+1}$. The partial trace of $E^{\mu}_{ij}$ in the last system is given by
  \begin{align}
      \Tr_{n+1} E^{\mu}_{ij} = \delta_{\alpha, \beta}\frac{d_\mu}{d_\alpha} E^{\alpha}_{ab}.
  \end{align}
\end{Lem}

\section{Covariant protocol for dPBT}
\label{appendix_sec:covariant_PBT}
This section shows that any dPBT protocol can be converted into a canonical form of the protocol, which we call the covariant protocol, without decreasing the teleportation fidelity, following Refs.~\cite{ishizaka2008asymptotic, mozrzymas2018optimal, leditzky2022optimality} (see Eqs.~\eqref{eq:wmu}, \eqref{eq:covariant_resource_state} and \eqref{eq:srm} and Figs.~\ref{fig:equivalence_protocol}~(a) and (b)).
The teleportation fidelity of the covariant protocol is shown in Lem.~\ref{lem:covariant_PBT}, which is shown in Refs.~\cite{mozrzymas2018optimal, leditzky2022optimality}.

We consider a general protocol for dPBT given by Eq.~(1) of the main text.
We take an eigendecomposition of the resource state $\phi_\mathrm{PBT}$ given by\footnote{We consider a slightly more general setting where the resource state is a mixed state than a usual setting where the resource state is a pure state.
This extra argument is given to explicitly provide a covariant protocol corresponding to any given dPBT protocol.}
\begin{align}
    \phi_\mathrm{PBT} &= \sum_{k} \ketbra{\phi^{(k)}},\label{eq:phi_PBT_decomposition}\\
    \ket{\phi^{(k)}} &= (O_{\mcA^N}^{(k)} \otimes \1_{\mcB^N}) \bigotimes_{a=1}^{N} \ket{\phi^+}_{\mcA_a \mcB_a},\label{eq:def_Ok}
\end{align}
where $O^{(k)}$ is a linear operator on $\mcA^N$ satisfying the normalization condition $\sum_k \Tr[O^{(k) \dagger}O^{(k)}] = d^N$, and $\ket{\phi^+}$ is the maximally entangled state defined by $\ket{\phi^+} \coloneqq {1\over \sqrt{d}}\sum_{i=1}^{d} \ket{ii}$ using the computational basis $\{\ket{i}\}_{i=1}^{d}$ of $\CC^d$.
Using this expression, the teleportation fidelity~(3) is written by
\begin{align}
\label{eq:the_channel_fidelity_dPBT}
    &f(\1_d, \Lambda)\nonumber\\
    &= \frac1{d^2} \sum^N_{a=1} \sum_{k} \Tr[(O_{\mcA^N}^{(k)} \otimes \1_{\mcA_0}) \sigma_a
    (O_{\mcA^N}^{(k)} \otimes \1_{\mcA_0})^\dagger \Pi_a ]\\
    &= \frac1{d^2} \sum^N_{a=1} \Tr(\tilde{\Pi}_a \sigma_a),
\end{align}
where $\sigma_a$ and $\tilde{\Pi}_a$ are defined by
\begin{align}
    \sigma_a &\coloneqq \frac1{d^{N-1}} 
    \left( \1_{\overline{\mcA_a}} \otimes \ketbra{\phi^+}_{\mcA_a \mcA_0} \right),\\
    \tilde{\Pi}_a &\coloneqq \sum_{k} (O_{\mcA^N}^{(k)} \otimes \1_{\mcA_0})^\dagger \Pi_a (O_{\mcA^N}^{(k)} \otimes \1_{\mcA_0}),\label{eq:def_tilde_Pi_a}
\end{align}
and $\overline{\mcA_a}$ is defined by $\overline{\mcA_a}\coloneqq \bigotimes_{i\neq a} \mcA_i$.
The set of operators $\{\tilde{\Pi}_a\}_{a=1}^{N}$ defined in Eq.~\eqref{eq:def_tilde_Pi_a} corresponds to the set of operators $\{\Pi_a\}_{a=1}^{N}$ satisfying $\Pi_a\geq 0$ and $\sum_a \Pi_a = \1$ if and only if
\begin{align}
    \tilde{\Pi}_a \geq 0, \quad \sum_a \tilde{\Pi}_a = X_{\mcA^N} \otimes \1_{\mcA_0},
\end{align}
where $X$ is defined by
\begin{align}
    X \coloneqq \sum_k O^{(k) \dagger} O^{(k)}.
\end{align}
Therefore, the optimal channel fidelity for a resource state given by Eqs.~\eqref{eq:phi_PBT_decomposition} and \eqref{eq:def_Ok} is written as the following semidefinite programming (SDP):
\begin{align}
\label{eq:sdp_PBT}
\begin{split}
    &\max {1\over d^2} \sum_{a=1}^{N} \Tr(\tilde{\Pi}_a \sigma_a),\\
    \mathrm{s.t.} \; & \tilde{\Pi}_a\geq 0, \quad \sum_a \tilde{\Pi}_a = X_{\mcA^N} \otimes \1_{\mcA_0}, \quad \tr(X)=d^N.
\end{split}
\end{align}
Since the operator $\sigma_a$ satisfies the unitary group symmetry and symmetric group covariance given by
\begin{align}
\label{eq:sigma_a_unitary_group_symmetry}
    &[U^{\otimes N}_{\mcA^N} \otimes U^*_{\mcA_0}, \sigma_a] = 0 \quad \forall U\in \SU(d),\\
\label{eq:sigma_a_symmetric_group_symmetry}
    &[(V_\pi)_{\mcA^N} \otimes \1_{\mcA_0}] \sigma_a [(V_\pi)_{\mcA^N} \otimes \1_{\mcA_0}]^\dagger = \sigma_{\pi(a)} \quad \forall \pi \in \mfS_N,
\end{align}
where $V_\pi$ is the permutation operator defined in Appendix~\ref{appendix_sec:shur_weyl_duality} and $*$ denotes the complex conjugate in the computational basis, we obtain
\begin{align}
    \max {1\over d^2} \sum_{a=1}^{N} \Tr(\tilde{\Pi}_a \sigma_a) = \max {1\over d^2} \sum_{a=1}^{N} \Tr(\tilde{\Pi}_a' \sigma_a),
\end{align}
where $\tilde{\Pi}_a'$ is $\SU(d) \times \mfS_N$-twirled operator defined by
\begin{align}
    \tilde{\Pi}_a' \coloneqq {1\over N!} \sum_{\pi\in\mfS_N} &\int_{\SU(d)} \dd U [(U^{\otimes N} V_\pi)_{\mcA^N} \otimes U^*_{\mcA_0}] \nonumber\\
    &\times \Pi_{\pi^{-1}(a)} [(U^{\otimes N} V_\pi)_{\mcA^N} \otimes U^*_{\mcA_0}]^\dagger.
\end{align}
The set of operators $\{\tilde{\Pi}_a'\}_a$ satisfies $\tilde{\Pi}_a' \geq 0$ and $\sum_{a=1}^{N}\tilde{\Pi}_a' = X'_{\mcA^N} \otimes \1_{\mcA_0}$ for the operator $X'$ defined by
\begin{align}
    X'
    &\coloneqq {1\over N!} \sum_{\pi\in\mfS_N} \int_{\SU(d)} \dd U (U^{\otimes N} V_\pi) X (U^{\otimes N} V_\pi)^\dagger,
\end{align}
satisfying the unitary group and symmetric group symmetry given by
\begin{align}
    &[U^{\otimes N}, X'] = 0 \quad \forall U\in\SU(d),\\
    &[V_\pi, X'] = 0 \quad \forall \pi\in\mfS_N.
\end{align}
Therefore, $X'$ can be written as
\begin{align}
    \label{eq:Xprime}
    X'&= {d^{N}} \sum_{\mu} {w_\mu^2 \over d_\mu m_\mu}  P_\mu,
\end{align}
where $w_\mu$ is defined by
\begin{align}
    w_\mu \coloneqq \sqrt{\Tr(X' P_\mu) \over d^{N}} =  \sqrt{\Tr(X P_\mu) \over d^{N}},
\end{align}
and $P_\mu$ is the Young projector (see Appendix~\ref{appendix_sec:shur_weyl_duality}).
Since the operator $X$ can be written by the resource state $\phi_\mathrm{PBT}$ as
\begin{align}
    X = d^N [\Tr_{\mcA^N}(\phi_\mathrm{PBT}^*)]_{\mcB^N \to \mcA^N},
\end{align}
where the subscript $\mcB^N \to \mcA^N$ represents a relabelling the space from $\mcB^N$ to $\mcA^N$, we obtain
\begin{align}
\label{eq:wmu}
    w_\mu = \sqrt{\Tr[\phi_\mathrm{PBT}^*[(P_\mu)_{\mcA^N} \otimes \1_{\mcB^N}]]}.
\end{align}
Thus, the solution of the SDP \eqref{eq:sdp_PBT} is upper bounded by the following SDP:
\begin{align}
\label{eq:covariant_sdp_PBT}
\begin{split}
    &\max {1\over d^2} \sum_{a=1}^{N} \Tr(\tilde{\Pi}_a \sigma_a),\\
    \mathrm{s.t.} \; & \tilde{\Pi}_a\geq 0, \quad \sum_a \tilde{\Pi}_a = X'_{\mcA^N} \otimes \1_{\mcA_0},\\
    &\tr(X') = d^N, \quad X' \text{ is given by Eq.~\eqref{eq:Xprime}},
\end{split}
\end{align}
which corresponds to the SDP \eqref{eq:sdp_PBT} for a pure resource state given by
\begin{align}
\label{eq:covariant_resource_state}
    \ket{\tilde{\phi}_\mathrm{PBT}^{(\vec{w})}} &\coloneqq (O_{\mcA^N}^{(\vec{w})} \otimes \1_{\mcB^N}) \bigotimes_{a=1}^{N} \ket{\phi^+}_{\mcA_a \mcB_a},\\
\label{eq:covariant_O}
    O_{\mcA^N}^{(\vec{w})} &\coloneqq \sqrt{d^N} \sum_{\mu \in \young{d}{N}} \frac{w_\mu}{\sqrt{d_\mu m_\mu} } P_\mu.
\end{align}
Reference \cite{leditzky2022optimality} shows that the SDP \eqref{eq:covariant_sdp_PBT} gives the optimal value when the POVM is given by the square-root measurement (or the pretty good measurement) \cite{belavkin1975optimal, kholevo1979asymptotically, hausladen1994pretty} for the state ensemble $\{(1/N, \sigma_a)\}_{a=1}^{N}$:
\begin{align}
\label{eq:srm}
    \Pi_a &= \sigma^{-1/2} \sigma_a \sigma^{-1/2},\\
    \sigma &\coloneqq \sum_a \sigma_a.
\end{align}
We call the dPBT protocol with the resource state given in \eqref{eq:covariant_resource_state} and the POVM given in \eqref{eq:srm} the covariant protocol since they are covariant with respect to representations of the symmetric group $\mfS_N$ and the unitary group $\SU(d)$ as follows:
\begin{align}
    (U^{\otimes N}_{\mcA^N} \otimes U^{*\otimes N}_{\mcB^N})\ket{\tilde{\phi}_\mathrm{PBT}^{(\vec{w})}} = \ket{\tilde{\phi}_\mathrm{PBT}^{(\vec{w})}} \quad &\forall U\in\SU(d),\\
    [(V_\pi)_{\mcA^N}\otimes (V_\pi)_{\mcB^N}]\ket{\tilde{\phi}_\mathrm{PBT}^{(\vec{w})}} = \ket{\tilde{\phi}_\mathrm{PBT}^{(\vec{w})}} \quad &\forall \pi\in\mfS_N,\\
    (U^{\otimes N}_{\mcA^N}\otimes U^*_{\mcA_0}) \Pi_a (U^{\otimes N}_{\mcA^N}\otimes U^*_{\mcA_0})^\dagger = \Pi_a \quad &\forall U\in\SU(d),\\
    [(V_\pi)_{\mcA^N} \otimes \1_{\mcA_0}]\Pi_a [(V_\pi)_{\mcA^N} \otimes \1_{\mcA_0}]^\dagger = \Pi_{\pi(a)} \quad &\forall \pi\in\mfS_N,
\end{align}
which can be shown using Eqs.~\eqref{eq:sigma_a_unitary_group_symmetry}, \eqref{eq:sigma_a_symmetric_group_symmetry}, and \eqref{eq:covariant_O}.

In conclusion, any dPBT protocol can be converted into a covariant protocol using the resource state defined in Eqs.~\eqref{eq:wmu} and \eqref{eq:covariant_resource_state} and the square-root measurement defined in Eq.~\eqref{eq:srm} without decreasing the teleportation fidelity [see Fig.~\ref{fig:equivalence_protocol}~(b)].
The teleportation fidelity for the obtained covariant protocol is given by the following Lemma.

\begin{Lem} {\rm \cite{leditzky2022optimality,mozrzymas2018optimal}}
\label{lem:covariant_PBT}
    The fidelity of $d$-dimensional dPBT for $N$ ports is given by
    \begin{align}
        F_{\mathrm{PBT}}(N,d) = \vec{w}^\sfT M_\mathrm{PBT}(N,d) \vec{w},
    \end{align}
    using  a vector $\vec{w} = (w_\mu)_{\mu \in \young{d}{N}} $ satisfying $w_\mu \geq 0$ for all $\mu \in \young{d}{N}$ and $\sum_{\mu \in \young{d}{n}} w_\mu^2 = 1$.
    The $\abs{\mathbb{Y}_{N}^d}\times \abs{\mathbb{Y}_{N}^d} $ matrix $M_\mathrm{PBT}(N,d)$ defined by
    \begin{align}
    \label{def:teleportation-matrix}
        (M_\mathrm{PBT}(N,d))_{\mu\nu} = {1\over d^2}\#[(\mu-\square) \cap (\nu-\square)],
    \end{align}
    where $\mu-\square$ and $\nu-\square$ are defined in Eq.~\eqref{eq:remove_box}, and $\# (\mathbb{X})$ is the cardinality of a set $\mathbb{X}$.
    In particular, the optimal fidelity of dPBT is given by
    \begin{align}
        F_\mathrm{PBT}(N,d) = \max_{\norm{\vec{w}}=1}\vec{w}^\sfT M_\mathrm{PBT}(N,d)\vec{w},
    \end{align}
    which equals to the maximal eigenvalue of $M_\mathrm{PBT}(N,d)$.
\end{Lem}

\begin{proof}
    See the details in Ref.~\cite{leditzky2022optimality}.
\end{proof}

\section{Parallel covariant protocol for unitary estimation}
\label{appendix_sec:covariant_unitary_estimation}
This section shows that any general adaptive protocol for unitary estimation can be converted into a canonical form of the protocol, which we call the parallel covariant protocol \cite{chiribella2005optimal, bisio2010optimal} without decreasing the fidelity [see Eqs.~\eqref{eq:def_of_probe_state}, \eqref{eq:v_alpha} and \eqref{eq:covariant_POVM} and Figs.~\ref{fig:equivalence_protocol}~(c) and (d)].
The fidelity of the parallel covariant protocol is shown in Lem.~\ref{lem:covariant_unitary_estimation}.

Reference \cite{bisio2010optimal} shows that the optimal fidelity of dSAR can be achieved by a covariant parallel ``measure-and-prepare'' protocol.
We show a similar construction to convert any unitary estimation protocol into a covariant parallel protocol.
This construction is based on the Choi representation of quantum operations and quantum supermaps \cite{choi1975completely, jamiolkowski1972linear, chiribella2008quantum}.

We briefly review the Choi representation.
We consider a quantum operation $\Lambda: \mcL(\mcI) \to \mcL(\mcO)$, where $\mcI$ and $\mcO$ are the Hilbert spaces corresponding to the input and output systems, and $\mcL(\mcX)$ is the set of linear operators on a linear space $\mcX$.
The Choi matrix $J_{\Lambda} \in \mcL(\mcI\otimes \mcO)$ is defined by
\begin{align}
    J_{\Lambda}\coloneqq \sum_{i,j} \ketbra{i}{j}_{\mcI} \otimes \Lambda(\ketbra{i}{j})_{\mcO},
\end{align}
where $\{\ket{i}\}_i$ is the computational basis of $\mcI$, and the subscripts $\mcI$ and $\mcO$ represent the Hilbert spaces where each term is defined.
The Choi matrix of the unitary operation $\mcU(\cdot) = U(\cdot) U^\dagger$ is given by $J_{\mcU} = \dketbra{U}$, where $\dket{U}$ is the dual ket defined by
\begin{align}
\label{eq:def_dual_vector}
    \dket{U}\coloneqq \sum_{i} \ket{i} \otimes U\ket{i}.
\end{align}
In the Choi representation, the composition of a quantum operation $\Lambda$ with a quantum state $\rho$ and that of quantum operations $\Lambda_1, \Lambda_2$ are represented in a unified way using a link product $\ast$ as
\begin{align}
    \Lambda(\rho) &= J_{\Lambda} \ast \rho,\\
    J_{\Lambda_2\circ \Lambda_1} &= J_{\Lambda_1} \ast J_{\Lambda_2},
\end{align}
where the link product $\ast$ for $A\in\mcL(\mcX\otimes \mcY)$ and $B\in \mcL(\mcY\otimes \mcZ)$ is defined as \cite{chiribella2008quantum}
\begin{align}
    A \ast B \coloneqq \Tr_\mcY [(A^{\sfT_\mcY}\otimes \1_{\mcZ})(\1_{\mcX} \otimes B)],
\end{align}
and $A^{\sfT_\mcY}$ is the partial transpose of $A$ over the subspace $\mcY$.

The most general protocol for the unitary estimation within the quantum circuit framework is given by \cite{chiribella2008quantum, bisio2010optimal}
\begin{align}
\label{eq:general_adaptive_unitary_estimation}
    p(\hat{U}|U) = \mathrm{Tr}\big[&M_{\hat{U}}(\mathcal{U}_{\mcI_n \to \mcO_n} \otimes \1_{\mcX_n}) \circ \Lambda_{n-1}  \nonumber\\
    &\circ \cdots \circ \Lambda_1 \circ (\mathcal{U}_{\mcI_1\to\mcO_1} \otimes \1_{\mcX_1}) \circ \phi\big],
\end{align}
where $\mcI_i$ and $\mcO_i$ for $i\in\{1, \ldots, n\}$ are Hilbert spaces defined by $\mcI_i \simeq \mcO_i \simeq \CC^d$, $\mcX_i$ for $i\in\{1, \ldots, n\}$ are auxiliary Hilbert spaces, $\Lambda_i$ are completely positive and trace preserving (CPTP) maps whose input and output spaces are given by $\Lambda_i: \mcL(\mcO_i \otimes \mcX_i) \to \mcL(\mcI_{i+1} \otimes \mcX_{i+1})$, $\phi$ is a quantum state on the Hilbert space $\mcI_1\otimes \mcX_1$, and $\{M_{\hat{U}} \dd \hat{U}\}_{\hat{U}}$ is a POVM on the Hilbert space $\mcO_n \otimes \mcX_n$ [see Fig.~\ref{fig:equivalence_protocol} (c)].
In the Choi representation, Eq.~\eqref{eq:general_adaptive_unitary_estimation} can be written as
\begin{align}
    p(\hat{U}|U) = T_{\hat{U}} \ast [J_{\mathcal{U}}^{\otimes n}]_{\mcI^n \mcO^n},
\end{align}
where $\{T_{\hat{U}} \dd \hat{U}\}_{\hat{U}}$ is called a quantum tester \cite{chiribella2009optimal, chiribella2008memory,bavaresco2021strict, bavaresco2022unitary} defined by
\begin{align}
    T_{\hat{U}}\coloneqq \phi \ast J_{\Lambda_1} \ast \cdots \ast J_{\Lambda_{n-1}} \ast M_{\hat{U}}^\sfT \in \mcL(\mcI^n \otimes \mcO^n), 
\end{align}
and $\mcI^n$ and $\mcO^n$ are the joint Hilbert spaces defined by $\mcI^n\coloneqq \bigotimes_{i=1}^{n} \mcI_i$ and $\mcO^n\coloneqq \bigotimes_{i=1}^{n} \mcO_i$, respectively.
The quantum tester satisfies the positivity and normalization conditions given by \cite{chiribella2009optimal, chiribella2008memory,bavaresco2021strict, bavaresco2022unitary}
\begin{align}
\label{eq:tester_condition}
    T_{\hat{U}}\geq 0, \quad \Tr[\int \dd \hat{U} T_{\hat{U}}] = d^n.
\end{align}
In terms of the quantum tester $\{T_{\hat{U}} \dd \hat{U}\}_{\hat{U}}$, the average-case fidelity of unitary estimation is given by
\begin{align}
    F_\mathrm{est}^{(\mathrm{ave})} = {1\over d^2} \int \dd U \int \dd \hat{U} T_{\hat{U}} \ast J_\mathcal{U}^{\otimes n} \abs{\Tr(\hat{U}^\dagger U)}^2.
\end{align}
Since
\begin{align}
    (V^{\otimes n}_{\mcI^n} \otimes W^{\otimes n}_{\mcO^n}) T_{\hat{U}} (V^{\otimes n}_{\mcI^n} \otimes W^{\otimes n}_{\mcO^n})^\dagger \ast J_{\mathcal{U}}^{\otimes n} = T_{\hat{U}} \ast J_{\mathcal{W}^\sfT \circ \mathcal{U} \circ \mathcal{V}}^{\otimes n}\nonumber\\
    \forall V, W\in\SU(d)
\end{align}
holds, we obtain
\begin{align}
    F_\mathrm{est}^{(\mathrm{ave})} = {1\over d^2} \int \dd U \int \dd \hat{U} T'_{\hat{U}} \ast J_\mathcal{U}^{\otimes n} \abs{\Tr(\hat{U}^\dagger U)}^2,
\end{align}
where $T'_{\hat{U}}$ is a $\SU(d)$-twirled operator given by
\begin{align}
    T'_{\hat{U}} \coloneqq \int \dd V \int \dd W (V^{\otimes n}_{\mcI^n} \otimes W^{\otimes n}_{\mcO^n}) T_{W^\sfT \hat{U} V} (V^{\otimes n}_{\mcI^n} \otimes W^{\otimes n}_{\mcO^n})^\dagger.
\end{align}
Therefore, the average-case fidelity of the original adaptive protocol can be achieved by any protocol having the probability distribution given by
\begin{align}
\label{eq:tilde_p}
    \tilde{p}(\hat{U}|U) \coloneqq T'_{\hat{U}} \ast J_\mathcal{U}^{\otimes n}.
\end{align}

We show that a parallel protocol can implement the probability distribution \eqref{eq:tilde_p} using an argument shown in Refs.~\cite{chiribella2009optimal, chiribella2008memory, bavaresco2022unitary}.
Defining $T'$ by
\begin{align}
    T' &\coloneqq \int \dd \hat{U} T'_{\hat{U}}\\
    &= \int \dd \hat{U} \int \dd V \int \dd W (V^{\otimes n}_{\mcI^n} \otimes W^{\otimes n}_{\mcO^n}) T_{\hat{U}} (V^{\otimes n}_{\mcI^n} \otimes W^{\otimes n}_{\mcO^n})^\dagger,
\end{align}
the operator $T'$ satisfies the following $\SU(d)$ symmetry:
\begin{align}
\label{eq:T'_unitary_symmetry}
    [V^{\otimes n}_{\mcI^n} \otimes W^{\otimes n}_{\mcO^n}, T'] = 0 \quad \forall V, W\in\SU(d).
\end{align}
From Eq.~\eqref{eq:tester_condition}, $T'$ satisfies the following positivity and normalization conditions:
\begin{align}
    T' \geq 0, \quad \Tr(T') = d^n.
\end{align}
We define a quantum state $\ket{\phi'_\mathrm{est}}$ and a POVM $\{M'_{\hat{U}} \dd U\}$ by \cite{chiribella2009optimal, chiribella2008memory, bavaresco2022unitary}
\begin{align}
    \ket{\phi'_\mathrm{est}}&\coloneqq \sqrt{T'}^\sfT \dket{\1}_{\mcI^n \mcO^n},\\
    M'_{\hat{U}} &\coloneqq (T^{\prime -1/2} T'_{\hat{U}} T^{\prime -1/2})^\sfT,
\end{align}
where the normalization condition of $\ket{\phi'_\mathrm{est}}$ can be checked as follows:
\begin{align}
    &\Tr[\ketbra{\phi'_\mathrm{est}}] \nonumber\\
    &= T' \ast \dketbra{\1}_{\mcI^n \mcO^n}\\
    &= \int \dd \hat{U} \int \dd V \int \dd W (V^{\otimes n}_{\mcI^n} \otimes W^{\otimes n}_{\mcO^n})\nonumber\\
    &\hspace{60pt}\times T_{W^\sfT \hat{U} V} (V^{\otimes n}_{\mcI^n} \otimes W^{\otimes n}_{\mcO^n})^\dagger \ast \dketbra{\1} \\
    &= \int \dd \hat{U} \int \dd V \int \dd W p(W^\sfT \hat{U} V|W^\sfT V) \\
    &= 1.
\end{align}
The probability distribution given in Eq.~\eqref{eq:tilde_p} can be implemented by measuring the quantum state $(\1_{\mcI^n} \otimes U^{\otimes n}_{\mcO^n})\ket{\phi'_\mathrm{est}}$ with the POVM $\{M'_{\hat{U}} \dd \hat{U}\}_{\hat{U}}$ since
\begin{align}
    &\Tr[M'_{\hat{U}} (\1_{\mcI^n} \otimes U^{\otimes n}_{\mcO^n})\ketbra{\phi'_\mathrm{est}} (\1_{\mcI^n} \otimes U^{\otimes n}_{\mcO^n})^\dagger]\nonumber\\
    &= \mathrm{Tr}\Big[(T^{\prime -1/2} T'_{\hat{U}} T^{\prime -1/2})^\sfT (\1_{\mcI^n} \otimes U^{\otimes n}_{\mcO^n})\sqrt{T'}^\sfT \nonumber\\
    &\hspace{60pt} \times \dketbra{\1}_{\mcI^n \mcO^n} \sqrt{T'}^\sfT (\1_{\mcI^n} \otimes U^{\otimes n}_{\mcO^n})^\dagger\Big]\\
    &= \Tr[T_{\hat{U}}^{\prime \sfT} (\1_{\mcI^n} \otimes U^{\otimes n}_{\mcO^n}) \dketbra{\1}_{\mcI^n \mcO^n} (\1_{\mcI^n} \otimes U^{\otimes n}_{\mcO^n})^\dagger]\label{eq:use_unitary_symmetry}\\
    &= T'_{\hat{U}} \ast J_{\mathcal{U}}^{\otimes n}
\end{align}
holds, where we use Eq.~\eqref{eq:T'_unitary_symmetry} in Eq.~\eqref{eq:use_unitary_symmetry}.

Due to the unitary group symmetry \eqref{eq:T'_unitary_symmetry}, $\sqrt{T'}^{\sf T}$ can be written as
\begin{align}
    \sqrt{T'}^{\sf T} = \bigoplus_{\alpha, \beta\in\young{d}{n}} \1_{\mcU_{\alpha, 1}} \otimes \1_{\mcU_{\beta, 2}} \otimes (\sqrt{X_{\alpha\beta}})_{\mcS_{\alpha, 1} \mcS_{\beta, 2}},
\end{align}
where $(\mcU_{\alpha, 1}, \mcS_{\alpha, 1}), (\mcU_{\beta, 2}, \mcS_{\beta, 2})$ are the Hilbert spaces given by the decomposition \eqref{eq:decomp_space} for $\mcO^n$ and $\mcI^n$, respectively, and $X_{\alpha\beta}$ is a positive semidefinite matrix.
Therefore, $\ket{\phi'_\mathrm{est}}$ can be written as
\begin{align}
\label{eq:def_of_probe_state}
    &\ket{\phi'_\mathrm{est}} = \ket{\tilde{\phi}_\mathrm{est}^{(\vec{v})}} \coloneqq \bigoplus_{\alpha \in \young{d}{n}} \frac{v_\alpha}{\sqrt{d_\alpha}} \dket{W_\alpha},\\
    &\dket{W_\alpha} \coloneqq \sum_{s=1}^{d_\alpha} \ket{\alpha, s}_{\mcU_{\alpha,1}} \otimes \ket{\alpha, s}_{\mcU_{\alpha,2}} \otimes \ket{\mathrm{arb}}_{\mcS_{\alpha,1}\mcS_{\alpha,2}},
\end{align}
where $v_\alpha$ and $\ket{\mathrm{arb}}$ are given by
\begin{align}
    v_\alpha &\coloneqq \sqrt{d_\alpha \sum_{i,j=1}^{m_\alpha}\bra{\alpha, i}\otimes \bra{\alpha, i} X_{\alpha \alpha} \ket{\alpha,j}\otimes \ket{\alpha,j}},\\
    \ket{\mathrm{arb}} &\coloneqq {\sqrt{d_\alpha} \over v_\alpha} \sqrt{X_{\alpha\alpha}} \sum_{i=1}^{m_\alpha} \ket{\alpha, i}_{\mcS_{\alpha, 1}} \otimes \ket{\alpha, i}_{\mcS_{\alpha, 2}}.
\end{align}
As shown in Lem.~\ref{lem:covariant_unitary_estimation}, $\ket{\mathrm{arb}}$ does not affect the fidelity of unitary estimation, so it can be chosen arbitrarily.
The coefficient $v_\alpha$ is given from $T$ as
\begin{align}
\label{eq:v_alpha}
    v_\alpha = \sqrt{\Tr[T \dketbra{\1}_{\mcI^n \mcO^n} (P_{\alpha})_{\mcI^n} \otimes (P_{\alpha})_{\mcO^n}]},
\end{align}
where $P_\alpha$ is the Young projector (see Appendix~\ref{appendix_sec:shur_weyl_duality}).
The optimal POVM for the probe state \eqref{eq:def_of_probe_state} is given by the covariant form \cite{chiribella2005optimal}:
\begin{align}
\label{eq:covariant_POVM}
    M_{\hat U} &\coloneq  \ketbra{\eta_{\hat{U}} }{\eta_{\hat{U}} }, \\
    \ket{\eta_{\hat{U}} } &\coloneq 
   \label{eq:def_of_eta_U} 
    \bigoplus_{\alpha \in \mathbb{Y}^d_n } [(\hat{U}_\alpha)_{\mcU_{\alpha, 1}} \otimes \1_{\mcS_{\alpha, 1} \mcU_{\alpha, 2} \mcS_{\alpha, 2}} ] \sqrt{d_\alpha} \dket{W_\alpha},
\end{align}
where $\hat{U}_\alpha$ is the irreducible representation $\alpha$ of $\hat{U}$.
By relabelling the Hilbert spaces $\mcI^n$ and $\mcO^n$ as $\mcA^n$ and $\mcB^n$, respectively, the obtained protocol is given as Fig.~\ref{fig:equivalence_protocol}~(d).
We call the unitary estimation protocol with the probe state given in \eqref{eq:def_of_probe_state} and the POVM given in \eqref{eq:covariant_POVM} the parallel covariant protocol since they are covariant with respect to representations of the unitary group $\SU(d)$ as follows:
\begin{align}
    (V^{\otimes n}_{\mcA^n}\otimes V^{*\otimes n}_{\mcB^n})\ket{\tilde{\phi}_\mathrm{est}^{(\vec{v})}} = \ket{\tilde{\phi}_\mathrm{est}^{(\vec{v})}} \quad \forall V\in\SU(d),\\
    (V^{\otimes n}_{\mcA^n}\otimes W^{\otimes n}_{\mcB^n}) M_{\hat{U}} (V^{\otimes n}_{\mcA^n}\otimes W^{\otimes n}_{\mcB^n})^\dagger = M_{V\hat{U}W^\sfT}\nonumber\\
    \forall V, W\in\SU(d).
\end{align}

In conclusion, any given protocol for unitary estimation can be converted into a parallel covariant protocol for the probe state given by Eqs.~\eqref{eq:def_of_probe_state} and \eqref{eq:v_alpha} and the covariant POVM given in Eq.~\eqref{eq:covariant_POVM}.
Note that this construction can also be applied to the case when the original unitary estimation protocol is given by an indefinite causal order protocol \cite{hardy2007towards, oreshkov2012quantum, chiribella2013quantum}.
The average-case fidelity of the parallel covariant protocol for unitary estimation is given as the following Lemma.

\begin{Lem} {\rm \cite{yang2020optimal}}
\label{lem:covariant_unitary_estimation}
    The fidelity of the parallel covariant protocol using the probe state \eqref{eq:def_of_probe_state} and the POVM \eqref{eq:covariant_POVM} for $d$-dimensional unitary estimation using $n$ calls is given by
    \begin{align}
        F_\mathrm{est}(n,d) = \vec{v}^\sfT M_\mathrm{est}(n,d) \vec{v},
    \end{align}
    using a vector $\vec{v} = (v_\alpha)_{\alpha\in\young{d}{n}}$ satisfying $v_\alpha\geq 0$ for all $\alpha\in\young{d}{n}$ and $\sum_{\alpha\in\young{d}{n}} v_\alpha^2 = 1$. The $\abs{\mathbb{Y}_n^d}\times \abs{\mathbb{Y}_n^d} $ matrix $M_{\mathrm{est}}(n,d)$ is defined by
    \begin{align}
    \label{def:estimation-matrix}
        (M_\mathrm{est}(n,d))_{\alpha \beta} = \frac1{d^2} \# [(\alpha+\square) \cap (\beta+\square) \cap \young{d}{n+1}],
    \end{align}
    where $\alpha+\square$ and $\beta+\square$ are defined in Eq.~\eqref{eq:add_box}.
    In particular, the optimal  fidelity of unitary estimation is given by
    \begin{align}
        F_\mathrm{est}(n,d) = \max_{\norm{\vec{v}}=1}\vec{v}^\sfT M_\mathrm{est}(n,d)\vec{v},
    \end{align}
    which equals to the maximal eigenvalue of $M_\mathrm{est}(n,d)$.
\end{Lem}

\begin{proof}\footnote{The proof presented in Ref.~\cite{yang2020optimal} contains small inconsistencies; hence, for completeness, we present a full proof here.}
    From Eq.~\eqref{eq:def_of_probe_state} and \eqref{eq:def_of_eta_U}, the probability distribution of obtaining the estimator $\hat{U}$ is 
    \begin{align}
     &p(\hat{U}|U)\nonumber\\
    &= \Tr[M_{\hat{U}} (U^{\otimes n}\otimes \1_d^{\otimes n})\ketbra{\tilde{\phi}^{(\vec{v})}_{\mathrm{est}}}(U^{\otimes n}\otimes \1_d^{\otimes n})^\dagger]\\
    &= \Bigg\lvert \sum_{\alpha \in \mathbb{Y}^d_n} v_\alpha \dbra{W_\alpha } 
    [(\hat{U}_{\alpha})_{\mcU_{\alpha, 1}} \otimes \1_{\mcS_{\alpha, 1}} \otimes \1_{\mcU_{\alpha, 2}} \otimes \1_{\mcS_{\alpha, 2}}]^\dagger \nonumber\\
    &\hspace{30pt} \times [(U_{\alpha})_{\mcU_{\alpha, 1}} \otimes \1_{\mcS_{\alpha, 1}} \otimes \1_{\mcU_{\alpha, 2}} \otimes \1_{\mcS_{\alpha, 2}}] \dket{W_\alpha}   \Bigg\rvert^2\\
    &= \abs{ \sum_{\alpha \in \mathbb{Y}^d_n} v_\alpha \sum^{d_\alpha}_{s=1} \bra{\alpha, s} \hat{U}^\dagger_\alpha U_\alpha \ket{\alpha ,s }}^2 \\
    &= \abs{\sum_{\alpha \in \mathbb{Y}^d_n} v_\alpha \chi^{\alpha}( \hat{U}^\dagger U) }^2,
    \end{align}
    where $\chi^\alpha(U)$ is the character of the irreducible representation $\alpha$ defined by $\chi^\alpha(U) \coloneq \Tr(U_\alpha)$. The  fidelity is given by 
    \begin{align}
        &F_{\mathrm{est}}(n,d) \nonumber\\
        &=  \int \dd{U} \int \dd{\hat{U}} 
        \frac1{d^2} \abs{ \Tr(\hat{U}^\dagger U )}^2
        \abs{\sum_{\alpha \in \mathbb{Y}^d_n} v_\alpha \chi^{\alpha}( \hat{U}^\dagger U) }^2\\
        &=\frac1{d^2} \int \dd{U} \int \dd{\hat{U}} 
        \abs{\chi(\hat{U}^\dagger U) \sum_{\alpha \in \mathbb{Y}_n^d } v_\alpha \chi^\alpha (\hat{U}^\dagger U)}^2,
    \end{align}
        where $\chi$ is defined by $\chi(U)\coloneqq \Tr(U)$, which corresponds to the character of the irreducible representation $\square$.
        Since the characters satisfy the following properties:
    \begin{align}
        &\chi(V) \chi^\alpha(V) 
        = \sum_{\mu \in \alpha + \Box \cap \mathbb{Y}_{n+1}^d} \chi^\mu (V)\quad \forall V \in \mathrm{SU}(d), \\
        &\int \dd{V} \chi^\mu(V) \chi^{\nu,*}(V) = \delta_{\mu, \nu} \quad \forall \mu, \nu \in \young{d}{n},
    \end{align}
    the average-case fidelity is given by
    \begin{align}
        &F_{\mathrm{est}}(n,d)\nonumber\\
        &= \frac1{d^2} \int \dd{U} \int \dd{\hat{U}} 
        \abs{\sum_{\alpha \in \mathbb{Y}^d_n } v_\alpha \sum_{\mu \in \alpha + \Box} \chi^\alpha(\hat{U}^\dagger U) }^2
        \\
        &= \frac1{d^2} \int \dd{U} \int \dd{\hat{U}}
        \sum_{\alpha, \beta \in \mathbb{Y}^d_n} v_\alpha v_\beta \nonumber\\
        &\hspace{30pt} \times \sum_{\mu \in \alpha + \Box} \sum_{\nu \in \beta + \Box}
        \chi^\alpha(\hat{U}^\dagger U) \chi^{\beta,*} (\hat{U}^\dagger U) \\
        &= \frac1{d^2} \int \dd{U}  
        \qty[\sum_{\alpha, \beta \in \mathbb{Y}^d_n} v_\alpha v_\beta 
        \sum_{\mu \in \alpha + \Box} \sum_{\nu \in \beta + \Box} \delta_{\mu,\nu} ] \\
        &=
        \label{Result_of_calc}
        \frac1{d^2} \qty[\sum_{\alpha, \beta \in \mathbb{Y}^d_n} v_\alpha v_\beta 
        \sum_{\mu \in \alpha + \Box} \sum_{\nu \in \beta + \Box} \delta_{\mu,\nu} ].
    \end{align}
    The third equality follows from the invariance of the Haar measure given by $\dd{\hat{U} }= \dd(\hat{U}^\dagger U)$ \cite{mele2024introduction}.
    The fourth equality arises from the normality of the Haar measure, i.e., $\int \dd{U} = 1$, which explicitly demonstrates that average-case fidelity coincides with the worst-case fidelity for the covariant protocol.
    Rearranging this equation, the fidelity $F_{\mathrm{est}}(n,d)$ is expressed with the estimation matrix $M_\mathrm{est}(n,d)$ defined in Eq.~\eqref{def:estimation-matrix} as 
    \begin{equation}
        F_{\mathrm{est}}(n,d) =  \vec{v}^\sfT M_{\mathrm{est}}(n,d) \vec{v},
    \end{equation}
    where $\vec{v}$ is a unit vector supported by $\mathbb{Y}^d_n$.
\end{proof}

\section{Proof of Thm.~\ref{thm:equivalence}: One-to-one correspondence between dPBT and unitary estimation}
\label{appendix_sec:proof_equivalence}
In this section, we prove the main theorem of this paper.

\begin{proof}[Proof of Thm.~\ref{thm:equivalence}]
    We prove this Theorem by constructing a transformation between covariant protocols for dPBT using $N=n+1$ ports and unitary estimation using $n$ calls of the input unitary operation, which does not decrease the fidelity.
    Since the covariant protocol is shown to achieve the optimal fidelity for dPBT \cite{mozrzymas2018optimal, leditzky2022optimality}, the transformation from dPBT to unitary estimation shows
    \begin{align}
        F_\mathrm{PBT}(n+1, d) \leq F_\mathrm{est}(n,d).
    \end{align}
    Similarly, since the parallel covariant protocol is shown to achieve the optimal fidelity for unitary estimation \cite{chiribella2005optimal, bisio2010optimal}, the transformation from unitary estimation to dPBT shows 
    \begin{align}
        F_\mathrm{est}(n,d) \leq F_\mathrm{PBT}(n+1, d).
    \end{align}
    Therefore, we obtain
    \begin{align}
        F_\mathrm{PBT}(n+1, d) = F_\mathrm{est}(n,d).
    \end{align}
    In addition, since an explicit construction of the covariant protocols for dPBT and unitary estimation from any given protocol are shown in Appendixes~\ref{appendix_sec:covariant_PBT} and~\ref{appendix_sec:covariant_unitary_estimation}, this transformation shows a transformation between any given dPBT protocol and unitary estimation protocol.
    Below, we show the transformation between the covariant protocols for dPBT and unitary estimation.

    As shown in Lemmas~\ref{lem:covariant_PBT} and~\ref{lem:covariant_unitary_estimation}, the fidelity of the covariant dPBT and untiary estimation protocols are given by
    \begin{align}
        F_{\mathrm{PBT}}(N,d) &= \vec{w}^\sfT M_\mathrm{PBT}(N,d) \vec{w},\\
        F_\mathrm{est}(n,d) &= \vec{v}^\sfT M_\mathrm{est}(n,d) \vec{v}.
    \end{align}
    Reference \cite{mozrzymas2018optimal} shows that the matrix $M_{\mathrm{PBT}}(n+1,d)$ in Eq.~\eqref{def:teleportation-matrix} can be written as
    \begin{align}
        M_{\mathrm{PBT}}(n+1,d) = {1\over d^2} R^\sfT(n,d) R(n,d),
    \end{align}
    where $R(n,d)$ is a $\abs{\mathbb{Y}^d_n} \times \abs{\mathbb{Y}^d_{n+1}}$ matrix defined by
    \begin{align}
        (R(n,d))_{\alpha, \mu} = \begin{cases}
            1 & (\mu \in \alpha + \square)\\
            0 & (\text{otherwise})
        \end{cases} \quad \forall \alpha\in\young{d}{n}, \mu\in\young{d}{n+1}.
    \end{align}
    From Eq.~\eqref{def:estimation-matrix}, the matrix $M_{\mathrm{est}}(n,d)$ can also be written as
    \begin{align}
        M_{\mathrm{est}}(n,d)= {1\over d^2} R(n,d) R^\sfT(n,d).
    \end{align}
    Thus, from a given covariant protocol for dPBT parametrized by $\vec{w}$, one can obtain a parallel covariant protocol for unitary estimation parametrized by
    \begin{align}
    \label{eq:PBT_to_untiary_estimation}
        \vec{v} \coloneqq {R (n,d) \vec{w} \over \|R (n,d) \vec{w}\|}
    \end{align}
    without decreasing the fidelity since
    \begin{align}
        &{\vec{v}^\sfT M_\mathrm{est}(n,d) \vec{v} \over \vec{w}^\sfT M_\mathrm{PBT}(N,d) \vec{w}}\nonumber\\
        &= {\vec{w}^\sfT R^\sfT(n,d) R(n,d) R^\sfT(n,d) R(n,d) \vec{w} \over \vec{w}^\sfT R^\sfT (n,d) R (n,d) \vec{w}\ \vec{w}^\sfT R^\sfT(n,d) R(n,d) \vec{w}} \geq 1
    \end{align}
    holds, where we use $\vec{w}\vec{w}^\sfT \leq \1$.
    Similarly, from a given parallel covariant protocol for unitary estimation parametrized by $\vec{v}$, one can obtain a covariant protocol for dPBT parametrized by
    \begin{align}
    \label{eq:untiary_estimation_to_PBT}
        \vec{w} \coloneqq {R^\sfT (n,d) \vec{v} \over \|R^\sfT (n,d) \vec{v}\|}
    \end{align}
    without decreasing the fidelity since
    \begin{align}
        &{\vec{w}^\sfT M_\mathrm{PBT}(N,d) \vec{w} \over \vec{v}^\sfT M_\mathrm{est}(n,d) \vec{v}}\nonumber\\
        &= {\vec{v}^\sfT R(n,d) R^\sfT(n,d) R(n,d) R^\sfT(n,d) \vec{v} \over \vec{v}^\sfT R(n,d) R^\sfT(n,d) \vec{v} \vec{v}^\sfT R(n,d) R^\sfT(n,d) \vec{v}} \geq 1
    \end{align}
    holds, where we use $\vec{v}\vec{v}^\sfT \leq \1$.
\end{proof}

\section{Proof of Lem.~\ref{lemma:asymptotically_optimal_unitary_estimation}: Asymptotically optimal fidelity of unitary estimation}
\label{appendix_sec:asymptotically_optimal_unitary_estimation}
Reference \cite{yang2020optimal} shows a unitary estimation protocol achieving the fidelity given by
\begin{align}
    F_\text{est}(n,d)\geq 1-O(d^4)n^{-2}.\label{eq:fidelity_lower_bound}
\end{align}
This section shows that this scaling is tight, as shown below:
\begin{Thm}
\label{thm:fidelity_upper_bound}
    Suppose there exists a sequence of protocols, for $d\in \{2, \ldots\}$, $\epsilon<{1\over 200}$ and an unknown unitary operator $U\in\SU(d)$, using $n_{d,\epsilon}$ calls of $U, U^\dagger, \mathrm{c}U\coloneqq \ketbra{0}\otimes \1_d + \ketbra{1}\otimes U, \mathrm{c}U^\dagger\coloneqq \ketbra{0}\otimes \1_d + \ketbra{1}\otimes U^\dagger$ to obtain the estimator $\hat{U}$ satisfying $f(\hat{U}, \mcU) >1-\epsilon^2$ with probability $\mathrm{Prob}[f(\hat{U}, \mcU)> 1-\epsilon^2]\geq {2\over 3}$.
    Then, $n_{d,\epsilon}$ must satisfy $n_{d,\epsilon} \geq \Omega(d^2)/\epsilon$.
\end{Thm}

\begin{proof}[Proof of Lem.~\ref{lemma:asymptotically_optimal_unitary_estimation}]
If unitary estimation achieves the average-case channel fidelity $F_\mathrm{est}(n,d)=1-{1\over 3} \epsilon^2$ using $n$ calls of the input unitary operator $U\in \SU(d)$, the probability of obtaining the estimator $\hat{U}$ satisfying $f(\hat{U}, \mcU)>1-\epsilon^2$ should be $\geq {2\over 3}$ since
\begin{align}
    F_\mathrm{est}(n,d)&\leq 1-\mathrm{Prob}[f(\hat{U}, \mcU)\leq 1-\epsilon^2]\epsilon^2\\
    &= 1-[1-\mathrm{Prob}[f(\hat{U}, \mcU)> 1-\epsilon^2]]\epsilon^2
\end{align}
holds.
Then, from Thm.~\ref{thm:fidelity_upper_bound}, we obtain $n\geq \Omega(d^2)/\epsilon$, i.e., $\epsilon \geq \Omega(d^2)/n$.
Therefore, we have the following upper bound on the fidelity of unitary estimation:
\begin{align}
    F_\mathrm{est}(n,d) = 1-{1\over 3}\epsilon^2 \leq 1-\Omega(d^4) n^{-2}.
\end{align}
\end{proof}

Reference \cite{haah2023query} shows a lower bound on the diamond-norm error of unitary estimation, which is given as follows:
\begin{Prop}[Theorem 1.2 in Ref.~\cite{haah2023query}]
\label{prop:diamond_norm_lower_bound}
    Suppose there exists a sequence of protocols, for $d\in \{2, \ldots\}$, $\epsilon<{1\over 8}$ and an unknown unitary operator $U\in\SU(d)$, using $n_{d,\epsilon}$ calls of $U, U^\dagger, \mathrm{c}U\coloneqq \ketbra{0}\otimes \1_d + \ketbra{1}\otimes U, \mathrm{c}U^\dagger\coloneqq \ketbra{0}\otimes \1_d + \ketbra{1}\otimes U^\dagger$ to obtain the estimator $\hat{U}$ satisfying $\|\hat{\mcU}-\mcU\|_\diamond <\epsilon$ with probability $\mathrm{Prob}[\|\hat{\mcU}-\mcU\|_\diamond <\epsilon]\geq {2\over 3}$.
    Then, $n_{d,\epsilon}$ must satisfy $n_{d,\epsilon} \geq \Omega(d^2)/\epsilon$.
\end{Prop}
To extend Proposition~\ref{prop:diamond_norm_lower_bound} for the worst-case fidelity as shown in Thm.~\ref{thm:fidelity_upper_bound}, we first review the proof shown in Ref.~\cite{haah2023query}.
It first shows a lower bound on the query complexity of unitary estimation with a constant error given by $\Omega(d^2)$.
The proof of this lower bound uses the reduction of unitary estimation to the discrimination task among the set of unitaries $\mathcal{N}_d$ in which the different unitary operator has a constant distance.
The construction of such a set is shown in Proposition~\ref{prop:diamond_norm_net}, and the lower bound of the query complexity of the discrimination task is shown in Proposition~\ref{prop:discrimination}.
Then, we consider the unitary estimation with the error $\epsilon$.
Since the error $\epsilon$ on $U\in\SU(d)$ can be translated to a constant error on $U^{1/\alpha}$ for $\alpha = \Theta(\epsilon)$, the unitary estimation with the error $\epsilon$ can be translated to the estimation of $R\in\SU(d)$ with a constant error using $R^\alpha$.
To uniquely define $R^\alpha$ for $\alpha\in [-1,1]$, we restrict $R$ to be a Hermitian unitary (Definition~\ref{def:reflection}).
It is shown that any protocol using $n$ calls of $R^\alpha$ can be simulated by using $\Theta(\alpha n)$ calls of $\mathrm{c}R$ (Lem.~\ref{lem:fractional_query}).
By combining this result with the lower bound on the query complexity of unitary estimation with a constant error, we obtain $\Theta(\alpha n) = \Omega(d^2)$, i.e., $n = \Theta(d^2/\epsilon)$.

\begin{Def}[Definition 4.4 in Ref.~\cite{haah2023query}]
\label{def:reflection}
    A unitary $R\in\SU(d)$ is called a reflection if $R^2=\1_d$ holds.
    For a reflection $R\in\SU(d)$ and $\alpha\in[-1,1]$, we define
    \begin{align}
        R^\alpha \coloneqq {1\over 2}(\1_d+R) + e^{-i\pi\alpha}{1\over 2}(\1_d-R).
    \end{align}
\end{Def}

\begin{Prop}[Proposition 4.1 in Ref.~\cite{haah2023query}]
\label{prop:diamond_norm_net}
    There exists a sequence of sets of reflections $\mathcal{N}_d = \{U_x\}_{x=1}^{N_d} \subset \SU(d)$ for $d\in\{2, \ldots\}$ with $N_d \geq \exp(d^2/64)$ such that
    \begin{align}
        \|\mcU_x - \mcU_y\|_\diamond \geq {1\over 4} \quad \forall x\neq y.
    \end{align}
\end{Prop}

\begin{Prop}[Proposition 4.2 + Lemma 4.3 in Ref.~\cite{haah2023query}]
\label{prop:discrimination}
    Let $\mathcal{N}_d = \{U_x\}_{x=1}^{N_d} \subset \SU(d)$ be a set of unitary operators defined for $d\in\{2, \ldots\}$.
    Suppose there exists a sequence of protocols for $d\in\{2 , \ldots \}$ using $n_d$ queries of the input unitary operator $U_x$ for unknown $x\in\{1, \ldots, N_d\}$ to obtain the estimator $\hat{x}$ with the success probability
    \begin{align}
        {1\over N_d}\sum_{x=1}^{N_d} \mathrm{Prob}[\hat{x}=x|x] \geq \exp(-cn_d)
    \end{align}
    for $d\geq 100/c' \geq 1$ and $N_d\geq \exp(c' d^2)$ for some constants $c, c'>0$.
    Then, $n_d$ should satisfy
    \begin{align}
        n_d \geq c'' d^2
    \end{align}
    for some constant $c''$ that depends only on $c$ and $c'$.
\end{Prop}

\begin{Lem}[Lemma 4.5 in Ref.~\cite{haah2023query}]
\label{lem:fractional_query}
    For a reflection $R\in \SU(d)$ and $U \coloneqq R^\alpha$, let $C\in\SU(D)$ be a unitary operator implemented by a quantum circuit using $n$ queries to $U$ or $U^\dagger$.
    Then, there exists a unitary operator $C_R\in\SU(2^n D)$ implemented by a quantum circuit with $n$ auxiliary qubits using $50+100\alpha n$ queries to $\mathrm{c}R$ such that
    \begin{align}
    \label{eq:fractional_query}
        C_R (\ket{\psi}\otimes \ket{0}^{\otimes n}) = \nu \ket{\widetilde{\mathrm{out}}}+\ket{\Phi^\perp},
    \end{align}
    where $\ket{\widetilde{\mathrm{out}}}$ is a normalized vector satisfying
    \begin{align}
    \label{eq:closeness}
        \left\|\ket{\widetilde{\mathrm{out}}}-C\ket{\psi}\right\|_2\leq \exp(-99\alpha\pi n-20),
    \end{align}
    for the $2$-norm $\|\cdot\|_2$ defined by $\|\ket{\psi}\|_2\coloneqq \sqrt{\braket{\psi}}$,
    $\nu\in \CC$ is a complex number satisfying
    \begin{align}
    \label{eq:post_selection}
        \abs{\nu}\geq 0.99 \exp(-\alpha\pi n/2),
    \end{align}
    and $\ket{\Phi^\perp}$ satisfies $(\1_D\otimes \bra{0}^{\otimes n})\ket{\Phi^\perp} = 0$.
\end{Lem}

We extend Proposition~\ref{prop:diamond_norm_net} for the case of channel fidelity.
To this end, we define the Bures distance \cite{zyczkowski2006geometry} $d_B(U, V)$ between two unitary operators $U, V\in\SU(d)$ by
\begin{align}
    d_B(U, V) \coloneqq \sqrt{1-f(U, \mathcal{V})},
\end{align}
which satisfies the following properties:
\begin{align}
    d_B(U,W)&\leq d_B(U, V)+d_B(V, W),\\
    d_B(U, V) &= d_B(AUB, AVB) \quad \forall A, B\in\SU(d),\\
    d_B(U, V) &\geq {1\over \sqrt{2} d} \min_{\phi\in \mathrm{U}(1)}\|U-\phi V\|_\mathrm{tr},\label{eq:bures_trace}
\end{align}
where $\mathrm{U}(1)$ is the set of complex numbers with a unit norm, $\|\cdot\|_\mathrm{tr}$ is the trace norm defined by
\begin{align}
    \|\rho\|_\mathrm{tr} &\coloneqq \Tr\sqrt{\rho^\dagger \rho}.
\end{align}
The inequality \eqref{eq:bures_trace} is shown below.
First, the Bures distance $d_B(U, V)$ is given by
\begin{align}
    d_B(U, V) = {1\over \sqrt{2d}} \min_{\phi\in \mathrm{U}(1)}\|U-\phi V\|_\mathrm{F},
\end{align}
where $\|\cdot\|_\mathrm{F}$ is the Frobenius norm defined by
\begin{align}
    \|\rho\|_\mathrm{F} \coloneqq \sqrt{\Tr(\rho^\dagger \rho)}.
\end{align}
Then, the inequality \eqref{eq:bures_trace} holds since
\begin{align}
    \|\rho\|_\mathrm{F} = \sqrt{\sum_{i=1}^{d} \lambda_i^2} \geq \sqrt{{1\over d}\left(\sum_{i=1}^{d} \lambda_i\right)^2} = {\|\rho\|_\mathrm{tr}\over \sqrt{d}},
\end{align}
where $\lambda_i\geq 0$ are singular values of $\rho$.
Then, we show the following Proposition:

\begin{Prop}
\label{prop:bures_net}
    There exists a sequence of sets of reflections $\mathcal{N}_d = \{U_x\}_{x=1}^{N_d} \subset \SU(d)$ for $d\in\{2, \ldots\}$ with $N_d \geq \exp(d^2/288)$ such that
    \begin{align}
        d_B(U_x, U_y) \geq {1\over 100} \quad \forall x\neq y.
    \end{align}
\end{Prop}
\begin{proof}
    We construct the sequence of the sets $\mathcal{N}_d$ similarly to the proof of Proposition~\ref{prop:diamond_norm_net} in Ref.~\cite{haah2023query}.
    For $d\in\{2, 3\}$, $\mathcal{N}_d$ defined by
    \begin{align}
        \mathcal{N}_d \coloneqq \{\1_d, X_2\oplus \1_{d-2}\}
    \end{align}
    satisfies the properties shown in Proposition~\ref{prop:bures_net}, where $X_2$ is a Pauli X matrix defined by $X_2\coloneqq \ketbra{+}-\ketbra{-}$ for $\ket{\pm}\coloneqq (\ket{0}\pm \ket{1})/\sqrt{2}$.
    We construct $\mathcal{N}_d$ for $d\geq 4$ as shown below.
    We define $r$ and $b$ by
    \begin{align}
        r\coloneqq \left\lfloor{d-1 \over 3}\right\rfloor \geq d/6, \quad b\coloneqq d-2r \geq r \quad \forall d\geq 4.
    \end{align}
    There exists a sequence of sets of $2r$-dimensional density operators $\mathcal{D}_r = \{\rho_x\}_{x=1}^{N_d}$, where $D_r \geq \exp(r^2/8)$ and $\rho_x$ has rank $r$, exactly half of the eigenvalues are $1/r$ and the other half are zero, and 
    \begin{align}
        \|\rho_x - \rho_y\|_\mathrm{tr} \geq {1\over 4} \quad \forall x\neq y
    \end{align}
    as shown in Lemma 8 of Ref.~\cite{haah2016sample}.
    Defining the set of reflections
    \begin{align}
        \mathcal{N}_d \coloneqq \{U_x\coloneqq V_x\oplus \1_b \coloneqq (2r \rho_x - \1_{2r})\oplus \1_b\}_{x=1}^{N_d},
    \end{align}
    $\mathcal{N}_d$ satisfies the properties shown in Proposition~\ref{prop:bures_net} as shown below.
    The number of elements in $\mathcal{N}_d$ satisfies
    \begin{align}
        N_d \geq \exp(r^2/8) \geq \exp(d^2/288) \quad \forall d\geq 4.
    \end{align}
    The Bures distance between two different elements $U_x, U_y \in \mathcal{N}_d$ is given by
    \begin{align}
        &d_B(U_x, U_y)\nonumber\\
        &\geq {1\over \sqrt{2}d} \min_{\phi\in\mathrm{U}(1)} \|U_x-\phi U_y\|_\mathrm{tr}\\
        &= {1\over \sqrt{2}d} \min_{\phi\in\mathrm{U}(1)} \|(V_x-\phi V_y) \oplus (1-\phi)\1_b\|_\mathrm{tr}\\
        &\geq {1\over \sqrt{2}d} \min_{\phi\in\mathrm{U}(1)} \max\{\|\phi \1_{2r}-V_y^\dagger V_x\|_\mathrm{tr}, \|(1-\phi)\1_b\|_\mathrm{tr}\}\\
        &\geq {1\over \sqrt{2}d} \min_{\phi\in\mathrm{U}(1)} \max\{\|\1_{2r} - V_y^\dagger V_x\|_\mathrm{tr} - \|(\phi-1) \1_{2r}\|_\mathrm{tr}, \nonumber\\
        &\hspace{150pt}\|(1-\phi)\1_b\|_\mathrm{tr}\}\\
        &= {1\over \sqrt{2}d} \min_{\phi\in\mathrm{U}(1)} \max\{\|\1_{2r} - V_y^\dagger V_x\|_\mathrm{tr} - 2r\abs{\phi-1}, \nonumber\\
        &\hspace{180pt} b \abs{\phi-1}\}\\
        &\geq {1\over \sqrt{2}d} {b\over 2r+b} \|\1_{2r} - V_y^\dagger V_x\|_\mathrm{tr}\\
        &= {1\over \sqrt{2}d} {b\over 2r+b} \|V_y - V_x\|_\mathrm{tr}\\
        &= {2r\over \sqrt{2d}} {b\over 2r+b} \|\rho_y - \rho_x\|_\mathrm{tr}\\
        &\geq {r\over 2\sqrt{2}d} {b\over 2r+b}\\
        &\geq {1\over 36\sqrt{2}}\\
        &\geq {1\over 100}.
    \end{align}
\end{proof}

\begin{proof}[Proof of Thm.~\ref{thm:fidelity_upper_bound}]
    We show Thm.~\ref{thm:fidelity_upper_bound} following a similar argument for the proof of Proposition~\ref{prop:diamond_norm_lower_bound} shown in Ref.~\cite{haah2023query}.
    The proof shown here is almost the same as shown in Ref.~\cite{haah2023query}, but we put a proof here for completeness.
    
    Suppose there exists a quantum circuit $\mcA$ using $n_{d,\epsilon}$ queries to the oracles $\mathrm{c}U$ and $\mathrm{c}U^\dagger$ to obtain the estimator $\hat{U}$ satisfying $f(\hat{U}, \mcU) > 1-\epsilon^2$, i.e., $d_B(\hat{U}, U)<\epsilon$ for $\epsilon< 1/200$.
    If $U$ is given by $U=R^\alpha$ for $\alpha = \lfloor{1\over 200\epsilon}\rfloor^{-1}$, we can obtain the estimator $\hat{R} \coloneqq \hat{U}^{1/\alpha}$ satisfying $d_B(\hat{R}, R)< 1/200$ with probability $\geq {2\over 3}$ since
    \begin{align}
        d_B(\hat{R}, R)
        &= d_B(\hat{U}^{1/\alpha}, U^{1/\alpha})\\
        &\leq \sum_{p=1}^{1/\alpha} d_B(\hat{U}^p U^{1/\alpha-p}, \hat{U}^{p-1} U^{1/\alpha-p+1})\\
        &= {1\over \alpha} d_B(\hat{U}, U)
        < {\epsilon \over \alpha}
        \leq {1\over 200}
    \end{align}
    holds.

    Next, we show that the quantum circuit $\mcA$ can be simulated with a black-box access to $\mathrm{cc}R$.
    Suppose a quantum circuit $\mcA$ consists of a unitary operator $C\in\SU(D)$ followed by a quantum measurement $\mcM$ \cite{nielsen2010quantum}.
    Using Lem.~\ref{lem:fractional_query}, we obtain a unitary operator $C_R$ using $n'_{d,\epsilon} = 50+100\alpha n_{d,\epsilon}$ queries to $\mathrm{cc}R$ satisfying Eq.~\eqref{eq:fractional_query}.
    We define a quantum circuit $\mcA_R$, which first measures the output state of auxiliary qubits in the computational basis.
    If the measurement outcome is $(0, \ldots, 0)$, we perform the measurement $\mcM$ on the target system and declare failure otherwise.
    From Eq.~\eqref{eq:post_selection}, the postselection probability is given by 
    \begin{align}
        \abs{\nu}^2\geq 0.99^2 \exp(-\alpha\pi n_{d,\epsilon}).
    \end{align}
    We denote the probability distribution of the measurement outcome of $\mcA$ and $\mcA_R$ conditioned on the postselection by $\mathbf{a}$ and $\mathbf{a}_R$, respectively.
    Then, from Eq.~\eqref{eq:closeness}, the total variation distance $\mathrm{dist}_{\mathrm{TV}}(\mathbf{a}, \mathbf{a}_R)$ of $\mathbf{a}$ and $\mathbf{a}_R$ is given by \cite{bernstein1993quantum}
    \begin{align}
        \mathrm{dist}_{\mathrm{TV}}(\mathbf{a}, \mathbf{a}_R) \leq 4\left\|\ket{\widetilde{\mathrm{out}}}-C\ket{\psi}\right\|_2\leq 4\exp(-99\alpha\pi n_{d,\epsilon}-20).
    \end{align}
    Thus, the success probability for $\mcA_R$ to obtain the estimator $\hat{R}$ satisfying $d_B(\hat{R}, R)<1/200$ is given by
    \begin{align}
        &\mathrm{Prob}_{\mcA_R}[d_B(\hat{R}, R)<1/200]\nonumber\\
        &= \mathrm{Prob}_{\mcA_R}[\text{post selection}]\nonumber\\
        &\hspace{15pt}\times \mathrm{Prob}_{\mcA_R}[d_B(\hat{R}, R)<1/200|\text{post selection}]\\
        &\geq 0.99^2 \exp(-\alpha\pi n_{d,\epsilon})\nonumber\\
        &\hspace{15pt}\times \left[\mathrm{Prob}_{\mcA}[d_B(\hat{R}, R)<1/200]- \mathrm{dist}_{\mathrm{TV}}(\mathbf{a}, \mathbf{a}_R)\right]\\
        &\geq 0.99^2 \exp(-\alpha\pi n_{d,\epsilon}) \nonumber\\
        &\hspace{15pt}\times \left[{2\over 3} - 4\exp(-99\alpha\pi n_{d,\epsilon}-20)\right]\\
        &\geq {1\over 2} \exp(-\alpha\pi n_{d,\epsilon}).
    \end{align}
    
    If we have a promise that $R$ is taken from a set $\mathcal{N}\subset \SU(d)$, we can discriminate an element $\mathrm{cc}R$ in $\mathrm{cc}\mathcal{N} \coloneqq \{\mathrm{cc}R|R\in\mathcal{N}\}\subset \SU(4d)$ with probability $\geq {1\over 2}\exp(-\alpha\pi n_{d,\epsilon}) \geq \exp(\delta_2-\delta_1 n'_{d,\epsilon})$ using $n'_{d,\epsilon} = 50+100\alpha n_{d,\epsilon}$ queries to $\mathrm{cc}R$ for constants $\delta_1\geq 0$ and $\delta_2\in\mathbb{R}$ that do not depend on $d$ and $\epsilon$.
    From Proposition~\ref{prop:discrimination}, we obtain
    \begin{align}
        n'_{d,\epsilon} = \Omega(d^2),
    \end{align}
    i.e.,
    \begin{align}
        n_{d,\epsilon} = {n'_{d,\epsilon}-50 \over 100 \alpha} \geq \Omega(d^2)/\epsilon
    \end{align}
    holds.
\end{proof}

\section{Proof of Cor.~\ref{cor:optimal_unitary_estimation_small_n}: Optimal unitary estimation protocol for $n\leq d-1$}
\label{appendix_sec:optimal_unitary_estimation}
In this section, we explicitly construct the optimal unitary estimation protocol for $n\leq d-1$ shown in Cor.~\ref{cor:optimal_unitary_estimation_small_n}.
To this end, we apply the transformation shown in Appendix~\ref{appendix_sec:proof_equivalence} to the optimal covariant dPBT protocol shown in Refs.~\cite{ishizaka2008asymptotic, mozrzymas2018optimal}.

The optimal deterministic port-based teleportation using $N$ ports to teleport a $d$-dimensional state for the case of $N\leq d$ is constructed in Ref.~\cite{mozrzymas2018optimal}, which shows the following fidelity:
\begin{align}
    F_\mathrm{PBT}(N,d) = {N\over d^2},
\end{align}
where the optimal resource state is the covariant state \eqref{eq:covariant_resource_state} parametrized by $\vec{w}$ given as
\begin{align}
    w_\mu = {m_\mu \over \sqrt{N!}},\label{eq:optimal_PBT}
\end{align}
where $m_\mu$ is the multiplicity of the irreducible representation $\mu$ of the unitary group (see Appendix~\ref{appendix_sec:shur_weyl_duality}).
By using the one-to-one correspondence between the unitary estimation and PBT shown in Thm.~\ref{thm:equivalence}, we obtain the optimal unitary estimation with the covariant probe state \eqref{eq:def_of_probe_state} parametrized by $\vec{v}$ given as
\begin{align}
    v_\alpha
    &\propto \sum_{\mu} (R(n,d))_{\alpha, \mu} w_\mu\\
    &\propto \sum_{\mu\in\alpha+\square} m_\mu\\
    &= N m_\alpha,
\end{align}
where we used Lem.~\ref{lem:symmetric_group_induced_irrep_dimension}.
Since $\sum_{\alpha\in\young{d}{n}} m_\alpha^2 = n! = (N-1)!$ holds for $n\leq d-1$, we obtain
\begin{align}
    v_\alpha = {m_\alpha \over \sqrt{n!}}.\label{eq:optimal_estimation}
\end{align}
The optimal fidelity of the unitary estimation is given by
\begin{align}
    F_\mathrm{est}(d,n) = F_\mathrm{PBT}(d,N=n+1) = {n+1\over d^2}.
\end{align}

\section{Proof of Thm.~\ref{thm:optimality_of_parallel_unitary_inversion}: Optimality of parallel unitary inversion for $n\leq d-1$}
\label{appendix_sec:feasilibity_proof}
In this section, we show Thm.~\ref{thm:optimality_of_parallel_unitary_inversion}, stating that parallel protocol achieves the optimal fidelity of unitary inversion among the most general transformation of quantum operations including the ones with indefinite causal order.
To this end, we first review the framework of quantum supermaps representing the transformation of quantum operations.
Then, we show the proof of Thm.~\ref{thm:optimality_of_parallel_unitary_inversion}.

The quantum supermap \cite{chiribella2008transforming, chiribella2008quantum} is given as a linear map $\mathcal{C}: \bigotimes_{i=1}^{n} [\mcL(\mcI_i) \to \mcL(\mcO_i)] \to [\mcL(\mcP) \to \mcL(\mcF)]$, where $\mcI_i$ and $\mcO_i$ are the Hilbert spaces representing the input and output systems of the $i$-th input quantum operation, $\mcP$ and $\mcF$ are the Hilbert spaces representing the input and output systems of the output quantum operation.
Since the quantum supermaps should preserve the completely positive (CP) and trace-preserving (TP) maps, it satisfies the completely CP preserving (CCPP) condition \cite{quintino2019probabilistic, gour2019comparison} given by
\begin{align}
    (\mathcal{C} \otimes I)(\Lambda) \text{ is CP} \quad \forall \text{ CP maps } \Lambda,
\end{align}
and TP preserving (TPP) condition \cite{quintino2019probabilistic, gour2019comparison} given by
\begin{align}
    \mathcal{C}(\Lambda) \text{ is TP} \quad \forall \text { TP maps } \Lambda.
\end{align}
There exist quantum supermaps satisfying both the CCPP and TPP conditions but not implementable within the quantum circuit framework \cite{chiribella2008quantum, wechs2021quantum}, called quantum supermaps with indefinite causal order \cite{hardy2007towards, oreshkov2012quantum, chiribella2013quantum}.
The most general quantum supermap including the one with indefinite causal order can be represented by the Choi matrix \cite{chiribella2008quantum} $C\in \mcL(\mcP\otimes \mcI^n \otimes \mcO^n \otimes \mcF)$ satisfying $C\geq 0$, $L_\mathrm{GEN}(C) = 0$ and $\Tr(C) = d^n$, where $L_\mathrm{GEN}$ is a linear map corresponding to the TPP condition, and $\mcI^n$ and $\mcO^n$ are the joint Hilbert spaces defined by $\mcI^n\coloneqq \bigotimes_{i=1}^{n} \mcI_i$ and $\mcO^n\coloneqq \bigotimes_{i=1}^{n} \mcO_i$, respectively.

Using the framework of the quantum supermap, we show Thm.~\ref{thm:optimality_of_parallel_unitary_inversion}.

\begin{proof}[Proof of Thm.~\ref{thm:optimality_of_parallel_unitary_inversion}]
Since the optimal fidelity of parallel unitary inversion is the same as that of unitary estimation \cite{quintino2022deterministic}, we obtain
\begin{align}
    F_\mathrm{inv}^{(\mathrm{PAR})}(n,d) = {n+1\over d^2}
\end{align}
for $n\leq d-1$ from Cor.~\ref{cor:optimal_unitary_estimation_small_n}.
We conclude the proof of Thm.~\ref{thm:optimality_of_parallel_unitary_inversion} by showing an upper bound on the optimal fidelity for indefinite causal order protocol.
To this end, we introduce a semidefinite programming (SDP) to obtain the optimal fidelity of unitary inversion.
The optimal fidelity of deterministic unitary inversion using general protocols is given by the following SDP \cite{quintino2022deterministic}:
\begin{align}
\begin{split}
  &\max \Tr(C\Omega)\\
  \text{s.t. } & \mcL(\mcP\otimes \mcI^n \otimes \mcO^n \otimes \mcF)\ni C\geq 0,\\
  & L_\mathrm{GEN}(C)=0,\\
  &\Tr(C) = d^n.
\end{split}
\label{eq:sdp}
\end{align}
where $\Omega\in \mcL(\mcP\otimes \mcI^n \otimes \mcO^n \otimes \mcF)$ is a positive operator called the performance operator.
The performance operator $\Omega$ is given by
\begin{align}
  \Omega
  \coloneqq \frac{1}{d^2}\int \dd U_d \dketbra{U_d}_{\mcI^n \mcO^n}^{\otimes n} \otimes \dketbra{U_d}_{\mcF\mcP},
\end{align}
where $\dd U_d$ is the Haar measure on $\mathrm{SU}(d)$, $\dket{U_d}$ is the dual vector of $U_d\in\SU(d)$ given in Eq.~\eqref{eq:def_dual_vector}, and the subscripts $\mcI^n \mcO^n$ and $\mcP \mcF$ represent the Hilbert spaces where the corresponding operators act.
As shown in Refs.~\cite{quintino2022deterministic,yoshida2023reversing}, the performance operator $\Omega$ satisfies the unitary group symmetry given by
\begin{align}
    \label{eq:Omega_unitary_group_symmetry}
    [\Omega, U^{\otimes n+1}_{\mcI^n \mcF} \otimes V^{\otimes n+1}_{\mcO^n \mcP}] = 0 \quad \forall U, V\in\SU(d).
\end{align}
Thus, it can be represented by a linear combination of $\{(E^{\mu}_{ij})_{\mcI^n \mcF} \otimes (E^{\nu}_{kl})_{\mcO^n \mcP}\}_{\mu, \nu\in\young{d}{n+1}, i,j\in\{1, \ldots, m_\mu\}, k,l\in\{1, \ldots, m_\nu\}}$ using the operator $E^{\mu}_{ij}$ defined in Eq.~\eqref{eq:def_E}, as follows \cite{quintino2022deterministic,yoshida2023reversing}\footnote{Note that Ref.~\cite{yoshida2023reversing} defines $d_\mu$ and $m_\mu$ in the opposite notation.}:
\begin{align}
    \Omega = {1\over d^2} \sum_{\mu\in\young{d}{n+1}}\sum_{i,j=1}^{m_\mu} {(E^{\mu}_{ij})_{\mcI^n \mcF} \otimes (E^{\mu}_{ij})_{\mcO^n \mcP} \over d_\mu}.
\end{align}
In addition, the performance operator $\Omega$ satisfies the symmetric group symmetry given by
\begin{align}
    \label{eq:Omega_symmetric_group_symmetry}
    [\Omega, (V_\sigma)_{\mcI^n} \otimes (V_\sigma)_{\mcO^n} \otimes \1_{\mcP \mcF}] = 0 \quad \forall \sigma\in\mfS_n,
\end{align}
where $V_\sigma$ is the representation of the symmetric group defined in Eq.~\eqref{eq:def_permutation_operator}.
The dual problem of the SDP \eqref{eq:sdp} is given by \cite{bavaresco2021strict}
\begin{align}
\begin{split}
    &\min \lambda\\
    \text{s.t. } & W\in\mcL(\mcP \otimes \mcI^n \otimes \mcO^n)\\
    & \Omega \leq \lambda (W\otimes \1_{\mcF}),\\
    & \Tr_{\mcO_i} W = \Tr_{\mcI_i \mcO_i} W \otimes \frac{\1_{\mcI_i}}{d}\;\;\;\forall i\in\{1, \ldots, n\},\\
    & \Tr W =d^n,
\end{split}\label{eq:dual_sdp}
\end{align}
and any feasible solution of the dual problem~\eqref{eq:dual_sdp} gives an upper bound of the optimal fidelity of deterministic unitary inversion using general protocols.

We construct a feasible solution of the dual SDP~\eqref{eq:dual_sdp} with
\begin{align}
    \lambda &= \frac{n+1}{d^2},\label{eq:def_lambda}
\end{align}
which shows that the optimal fidelity of unitary inversion using general protocols is bounded by
\begin{align}
    F_\mathrm{inv}^{(\mathrm{GEN})}(n,d) \leq {n+1\over d^2}.
\end{align}
Since $F_\mathrm{inv}^{(\mathrm{PAR})}(n,d)\leq F_\mathrm{inv}^{(\mathrm{SEQ})}(n,d)\leq F_\mathrm{inv}^{(\mathrm{GEN})}(n,d)$ hold, we obtain
\begin{align}
    &{n+1\over d^2} = F_\mathrm{inv}^{(\mathrm{PAR})}(n,d)\leq F_\mathrm{inv}^{(\mathrm{SEQ})}(n,d)
    \nonumber\\
    &\hspace{60pt}\leq F_\mathrm{inv}^{(\mathrm{GEN})}(n,d) \leq {n+1\over d^2},
\end{align}
i.e., $F_\mathrm{inv}^{(\mathrm{PAR})}(n,d)=F_\mathrm{inv}^{(\mathrm{SEQ})}(n,d)=F_\mathrm{inv}^{(\mathrm{GEN})}(n,d)={n+1\over d^2}$ holds.
Due to the unitary group and symmetric group symmetries~\eqref{eq:Omega_unitary_group_symmetry} and \eqref{eq:Omega_symmetric_group_symmetry} of the performance operator $\Omega$, the constraints in the dual SDP~\eqref{eq:dual_sdp} is invariant under the twirling operations given by
\begin{align}
    W&\mapsto \int \dd U \int\dd V (U^{\otimes n}_{\mcI^n} \otimes V^{\otimes n+1}_{\mcO^n \mcP}) W (U^{\otimes n}_{\mcI^n} \otimes V^{\otimes n+1}_{\mcO^n\mcP})^\dagger,\\
    W&\mapsto {1\over n!} \sum_{\sigma\in \mfS_{n}} [(V_\sigma)_{\mcI^n} \otimes (V_\sigma)_{\mcO^n} \otimes \1_{\mcP}] \nonumber\\
    &\hspace{60pt}\times W [(V_\sigma)_{\mcI^n} \otimes (V_\sigma)_{\mcO^n} \otimes \1_{\mcP}]^\dagger.
\end{align}
Thus, without loss of generality, we can assume that $W$ satisfies the unitary group and symmetric group symmetries given by
\begin{align}
    \label{eq:W_unitary_group_symmetry}
    &[W, U^{\otimes n}_{\mcI^n} \otimes V^{\otimes n+1}_{\mcO^n\mcP}] = 0 \quad \forall U, V\in\SU(d),\\
    \label{eq:W_symmetric_group_symmetry}
    &[W, (V_\sigma)_{\mcI^n} \otimes (V_\sigma)_{\mcO^n} \otimes \1_{\mcP}] = 0 \quad \forall \sigma\in\mfS_n.
\end{align}
We assume the following ansatz on $W$:
\begin{align}
    W = \sum_{\alpha\in \young{d}{n}} \sum_{\mu\in \alpha+\square} w_{\alpha\mu} \sum_{a,b=1}^{m_\alpha} (E^\alpha_{ab})_{\mcI^n} \otimes (E^\mu_{a^\alpha_\mu b^\alpha_\mu})_{\mcO^n\mcP},\label{eq:ansatz_W}
\end{align}
where $w_{\alpha\mu}\geq 0$.
Note that the operator $(E^\mu_{a^\alpha_\mu b^\alpha_\mu})_{\mcO^n\mcP}$ is well-defined for all $\mu\in\alpha+\square$ since $\alpha+\square \subset \young{d}{n+1}$ holds for all $\alpha\in\young{d}{n}$ when $n \leq d-1$.
The ansatz~\eqref{eq:ansatz_W} satisfies the unitary group symmetry~\eqref{eq:W_unitary_group_symmetry} by definition of the operators $E^\alpha_{ab}$.
The ansatz~\eqref{eq:ansatz_W} also satisfies the symmetric group symmetry~\eqref{eq:W_symmetric_group_symmetry} since
\begin{align}
    &[(V_\sigma)_{\mcI^n} \otimes (V_\sigma)_{\mcO^n} \otimes \1_{\mcP}] W [(V_\sigma)_{\mcI^n} \otimes (V_\sigma)_{\mcO^n} \otimes \1_{\mcP}]^{\dagger}\nonumber\\
    &= \sum_{\alpha\in\young{d}{n}} \sum_{\mu\in\alpha+\square} 
    w_{\alpha\mu}(\1_{\mcU_\alpha})_{\mcI^n} \otimes (\1_{\mcU_\mu})_{\mcO^n\mcP} \nonumber\\
    &\hspace{30pt} \otimes (\sigma_\alpha\ketbra{\alpha, a}{\alpha, b}\sigma_\alpha^{\dagger})_{\mcS_\alpha} \otimes (\sigma_\mu\ketbra{\mu, a^\alpha_\mu}{\mu, b^\alpha_\mu}\sigma_\mu^{\dagger})_{\mcS_\mu}\\
    &= \sum_{\alpha\in\young{d}{n}} \sum_{\mu\in\alpha+\square} w_{\alpha\mu} \sum_{a,b=1}^{m_\alpha} (\1_{\mcU_\alpha})_{\mcI^n} \otimes (\1_{\mcU_\mu})_{\mcO^n\mcP} \nonumber\\
    &\hspace{30pt}\otimes (\ketbra{\alpha, a}{\alpha, b})_{\mcS_\alpha} \otimes (\ketbra{\mu, a^\alpha_\mu}{\mu, b^\alpha_\mu})_{\mcS_\mu}\\
    &= W
\end{align}
holds, where we use the equality
\begin{align}
    \sum_{a=1}^{m_\alpha} \sigma_\alpha \ket{\alpha, a} \otimes \sigma_\mu \ket{\mu, a^{\alpha}_\mu} &= \sum_{a=1}^{m_\alpha} \sigma_\alpha \ket{\alpha, a} \otimes \sigma_\mu^* \ket{\mu, a^{\alpha}_\mu} \\
    &= \sum_{a=1}^{m_\alpha} \ket{\alpha, a} \otimes \ket{\mu, a^{\alpha}_\mu},
\end{align}
which is derived from the fact that $\sigma_\mu$ is a real matrix in the Young-Yamanouchi basis\footnote{See, e.g., Ref.~\cite{ceccherini2010representation}.} and the action of $\sigma\in \mfS_n$ on $\ket{\mu, a^{\alpha}_\mu}$ is unitarily equivalent to the action of $\sigma$ on $\ket{\alpha, a}$ as shown in Eq.~(\ref{eq:subgroup_adapted}).
Due to the permutation symmetry (\ref{eq:W_symmetric_group_symmetry}), the constraints of the dual SDP \eqref{eq:dual_sdp} hold if
\begin{align}
    \Omega &\leq \lambda (W \otimes \1_{\mcF}),\label{eq:dual_problem_cond0}\\
    \Tr_{\mcO_n} W &= \Tr_{\mcI_n \mcO_n} W \otimes \frac{\1_{\mcI_n}}{d}, \label{eq:dual_problem_cond1}\\
    \Tr W&=d^n.\label{eq:dual_problem_cond2}
\end{align}

The right-hand side of the inequality~\eqref{eq:dual_problem_cond0} is given by
\begin{align}
    &\lambda(W\otimes \1_{\mcF})\nonumber\\
    &= \lambda \sum_{\alpha\in \young{d}{n}} \sum_{\mu, \nu\in \alpha+\square} w_{\alpha\mu} \sum_{a,b=1}^{m_\alpha} (E^\nu_{a^\alpha_\nu b^\alpha_\nu})_{\mcI^n\mcF} \otimes (E^\mu_{a^\alpha_\mu b^\alpha_\mu})_{\mcO^n\mcP}\\
    &\geq \lambda \sum_{\alpha\in \young{d}{n}} \sum_{\mu\in \alpha+\square} w_{\alpha\mu} \sum_{a,b=1}^{m_\alpha} (E^\mu_{a^\alpha_\mu b^\alpha_\mu})_{\mcI^n\mcF} \otimes (E^\mu_{a^\alpha_\mu b^\alpha_\mu})_{\mcO^n\mcP}\\
    &= \lambda \sum_{\mu\in \young{d}{n+1}} \sum_{\alpha \in \mu-\square} w_{\alpha\mu} \sum_{a,b=1}^{m_\alpha} (E^\mu_{a^\alpha_\mu b^\alpha_\mu})_{\mcI^n\mcF} \otimes (E^\mu_{a^\alpha_\mu b^\alpha_\mu})_{\mcO^n\mcP}\\
    &= \lambda \sum_{\mu\in\young{d}{n+1}} \sum_{\alpha \in\mu-\square} m_\alpha w_{\alpha\mu} (\1_{\mcU_\mu})_{\mcI^n\mcF} \otimes (\1_{\mcU_\mu})_{\mcO^n\mcP} \nonumber\\
    &\hspace{120pt} \otimes \ketbra{\phi_{\alpha, \mu}}{\phi_{\alpha, \mu}}_{\mcS_\mu \mcS_\mu},
\end{align}
where $\ket{\phi_{\alpha, \mu}} \in \mcS_\mu\otimes \mcS_\mu$ is defined by
\begin{align}
\ket{\phi_{\alpha, \mu}}&\coloneqq \frac{1}{\sqrt{m_\alpha}} \sum_{a=1}^{m_\alpha} \ket{\mu, a_\mu}_{\mcS_\mu} \otimes \ket{\mu, a_\mu}_{\mcS_\mu}.
\end{align}
On the other hand, the left-hand side of the inequality~\eqref{eq:dual_problem_cond0} is given by
\begin{align}
    \Omega &= \frac{1}{d^2}\sum_{\mu\in\young{d}{n+1}} \sum_{i,j=1}^{m_\mu} {(E^\mu_{ij})_{\mcI^n\mcF} \otimes (E^\mu_{ij})_{\mcO^n\mcP} \over d_\mu}\\
    &= {1\over d^2} \sum_{\mu\in\young{d}{n+1}} \frac{1}{d_\mu} (\1_{\mcU_\mu})_{\mcI^n\mcF} \otimes (\1_{\mcU_\mu})_{\mcO^n\mcP} \otimes \ketbra{\phi_{\mu}}{\phi_{\mu}}_{\mcS_\mu \mcS_\mu},
\end{align}
where $\ket{\phi_{\mu}} \in \mcS_\mu\otimes \mcS_\mu$ is defined by
\begin{align}
\ket{\phi_{\mu}}&\coloneqq \frac{1}{\sqrt{m_\mu}} \sum_{i=1}^{m_\mu} \ket{\mu, i}_{\mcS_\mu} \otimes \ket{\mu, i}_{\mcS_\mu}.
\end{align}
Since $\ket{\phi_\mu} = \sum_{\alpha\in\mu-\square} c_{\alpha, \mu} \ket{\phi_{\alpha, \mu}}$ and $\sum_{\alpha\in\mu-\square} |c_{\alpha, \mu}|^2=1$ holds for $c_{\alpha, \mu}\coloneqq \sqrt{m_\alpha/m_\mu}$, we have the following inequality:
\begin{align}
    \sum_{\alpha\in\mu-\square} \ketbra{\phi_{\alpha, \mu}}{\phi_{\alpha, \mu}} \geq \ketbra{\phi_{\mu}}{\phi_{\mu}}.
\end{align}
Therefore, by taking $w_{\alpha\mu}$ as
\begin{align}
    w_{\alpha\mu} = \frac{m_\mu}{(n+1) m_\alpha d_\mu},\label{eq:def_w_alpha_mu}
\end{align}
and $\lambda$ as in Eq.~\eqref{eq:def_lambda}, we obtain the inequality~\eqref{eq:dual_problem_cond0} as shown below:
\begin{align}
   &\lambda (W\otimes \1_{\mcF})\nonumber\\
    &\geq {1\over d^2}\sum_{\mu\in\young{d}{n+1}} \frac{m_\mu}{d_\mu} \sum_{\alpha \in\mu-\square} (\1_{\mcU_\mu})_{\mcI^n\mcF} \otimes (\1_{\mcU_\mu})_{\mcO^n\mcP} \nonumber\\
    &\hspace{120pt} \otimes \ketbra{\phi_{\alpha, \mu}}{\phi_{\alpha, \mu}}_{\mcS_\mu \mcS_\mu}\\
    &\geq \frac{1}{d^2}\sum_{\mu\in\young{d}{n+1}} \frac{1}{d_\mu} (\1_{\mcU_\mu})_{\mcI^n\mcF} \otimes (\1_{\mcU_\mu})_{\mcO^n\mcP} \otimes \ketbra{\phi_{\mu}}{\phi_{\mu}}_{\mcS_\mu \mcS_\mu}\\
    &= \Omega.
\end{align}

The constraint (\ref{eq:dual_problem_cond1}) is confirmed as follows.  By defining $\tau\coloneqq (n, n+1) \in \mfS_{n+1}$, $\Tr_{\mcO_n} W$ is given by
\begin{align}
    &\Tr_{\mcO_n} W \nonumber\\
    &= [\Tr_{\mcP} [(\1_{\mcI^n} \otimes (V_\tau)_{\mcO^n\mcP}) W (\1_{\mcI^n} \otimes (V_\tau)_{\mcO^n\mcP})^{\dagger}]]_{\mcO_n \to \mcP},
\end{align}
where subscript $\mcO_n \to \mcP$ represents the relabelling of the Hilbert spaces. This is further calculated as
\begin{align}
    &\Tr_{\mcO_n} W\nonumber\\
    &=
    \frac{1}{n+1}\sum_{\alpha\in\young{d}{n}} \sum_{\mu\in\alpha+\square} \frac{m_\mu}{m_\alpha d_\mu} \sum_{a,b=1}^{m_\alpha} (E^\alpha_{ab})_{\mcI^n} \nonumber\\
    &\hspace{75pt} \otimes [\Tr_{\mcP} [V_\tau (E^\mu_{a_\mu b_\mu}) V_\tau^{\dagger}]_{\mcO^n\mcP}]_{\mcO_n \to \mcP}\\
    &=
    \frac{1}{n+1}\sum_{\alpha\in\young{d}{n}} \sum_{\mu\in\alpha+\square} \frac{m_\mu}{m_\alpha d_\mu} \sum_{a,b=1}^{m_\alpha} \sum_{i,j=1}^{m_\mu} (E^\alpha_{ab})_{\mcI^n} \nonumber\\
    &\hspace{75pt} \otimes [\tau_\mu]_{i a_\mu} [\tau_\mu^\dagger]_{b_\mu j} [\Tr_{\mcP} (E^\mu_{ij})_{\mcO^n\mcP}]_{\mcO_n \to \mcP}\\
    &=
    \frac{1}{n+1}\sum_{\alpha\in\young{d}{n}} \sum_{\mu\in\alpha+\square} \frac{m_\mu}{m_\alpha} \sum_{a,b=1}^{m_\alpha} \sum_{i,j=1}^{m_\mu} (E^\alpha_{ab})_{\mcI^n} \nonumber\\
    &\hspace{30pt} \otimes \sum_{\beta\in \mu-\square} \sum_{c,d=1}^{d_\beta}  [\tau_\mu]_{c_\mu^\beta a_\mu^\alpha} [\tau_\mu^\dagger]_{b_\mu^\alpha d_\mu^\beta} \frac{(E^\beta_{cd})_{\mcO^{n-1}\mcP}}{d_\beta}\\
    &=
    \sum_{\alpha, \beta\in\young{d}{n}}
    \sum_{a,b=1}^{m_\alpha} \sum_{c,d=1}^{d_\beta} A^{\alpha\beta}_{abcd}(E^\alpha_{ab})_{\mcI^n} \otimes  (E^\beta_{cd})_{\mcO^{n-1}\mcP}
    ,
\end{align}
where $[\tau_\mu]_{ij}\coloneqq \bra{\mu, i}\tau_\mu \ket{\mu, j}$, $a_\mu^\alpha$ represents the index of the standard tableaux $s_{a_\mu^\alpha}^{\mu}$ obtained by adding a box \fbox{$n+1$} to the standard tableaux $s^\alpha_a$, and the coefficient $A^{\alpha\beta}_{abcd}$ is defined by
\begin{align}
    A^{\alpha\beta}_{abcd}\coloneqq \frac{1}{n+1} \sum_{\mu\in(\alpha+\square) \cap (\beta+\square)} \frac{m_\mu}{m_\alpha}[\tau_\mu]_{c_\mu^\beta a_\mu^\alpha} [\tau_\mu^\dagger]_{b_\mu^\alpha d_\mu^\beta}.\label{eq:def_coefficient}
\end{align}
To proceed with the calculation, we show the following Lemma for the coefficient $A^{\alpha\beta}_{abcd}$:
\begin{Lem}
\label{lem:coefficient}
When $n\leq d-1$ holds, coefficient $A^{\alpha\beta}_{abcd}$ defined by Eq.~(\ref{eq:def_coefficient}) is given by
\begin{align}
    A^{\alpha\beta}_{abcd} = \frac{1}{n} \sum_{\lambda\in(\alpha-\square) \cap (\beta-\square)} \sum_{p,q=1}^{m_\lambda} \frac{m_\beta}{m_\lambda} \delta_{a, p^\lambda_\alpha} \delta_{b, q_\lambda^\alpha} \delta_{c, p^\lambda_\beta} \delta_{d, q^\lambda_\beta}.
\end{align}
\end{Lem}
We prove Lem.~\ref{lem:coefficient} in the end of this section.
Using this relation, $\Tr_{\mcO_n} W$ is further calculated as
\begin{align}
    &\Tr_{\mcO_n} W\nonumber\\
    &=
    \frac{1}{n}\sum_{\alpha, \beta\in\young{d}{n}} \sum_{\lambda\in (\alpha-\square) \cap (\beta-\square)}
    \sum_{p,q=1}^{m_\lambda} \frac{m_\beta}{m_\lambda}\left(E^\alpha_{p_\alpha^\lambda q_\alpha^\lambda}\right)_{\mcI^n} \nonumber\\
    &\hspace{135pt}\otimes  \frac{\left(E^\beta_{p_\beta^\lambda q_\beta^\lambda}\right)_{\mcO^{n-1}\mcP}}{d_\beta}\\
    &=
    \frac{1}{n} \sum_{\lambda\in \young{d}{n-1}} \sum_{\beta\in \lambda+\square}
    \frac{m_\beta}{m_\lambda d_\beta}\sum_{p,q=1}^{m_\lambda} \sum_{\alpha\in\lambda+\square} \left(E^\alpha_{p_\alpha^\lambda q_\alpha^\lambda}\right)_{\mcI^n} \nonumber\\
    &\hspace{135pt} \otimes  \left(E^\beta_{p_\beta^\lambda q_\beta^\lambda}\right)_{\mcO^{n-1}\mcP}\\
    &=
    \frac{1}{n} \sum_{\lambda\in \young{d}{n-1}} \sum_{\beta\in \lambda+\square}
    \frac{m_\beta}{m_\lambda d_\beta}\sum_{p,q=1}^{m_\lambda} \sum_{\alpha\in\lambda+\square} \left(E^\lambda_{pq}\right)_{\mcI^{n-1}} \nonumber\\
    &\hspace{105pt} \otimes  \left(E^\beta_{p_\beta^\lambda q_\beta^\lambda}\right)_{\mcO^{n-1}\mcP} \otimes \1_{\mcI_n}.
\end{align}
Thus, the constraint (\ref{eq:dual_problem_cond1}) holds.

The constraint (\ref{eq:dual_problem_cond2}) is satisfied since
\begin{align}
    \Tr W &= \frac{1}{n+1}\sum_{\alpha\in\young{d}{n}} \sum_{\mu\in\alpha+\square} m_\mu d_\alpha\\
    &= \sum_{\alpha\in\young{d}{n}} m_\alpha d_\alpha\\
    &= \dim \bigoplus_{\alpha\in\young{d}{n}} \mcU_\alpha \otimes \mcS_\alpha\\
    &= \dim (\CC^d)^{\otimes n}\\
    &= d^n
\end{align}
holds, where we use Lem.~\ref{lem:symmetric_group_induced_irrep_dimension} in Appendix~\ref{appendix_sec:shur_weyl_duality}.

Therefore, $(\lambda, W)$ given in Eqs.~\eqref{eq:def_lambda}, \eqref{eq:ansatz_W} and \eqref{eq:def_w_alpha_mu} is a feasible solution of the dual SDP \eqref{eq:dual_sdp}.
As shown above, this gives a matching upper bound on the optimal fidelity of unitary inversion using general protocols, which completes the proof.
\end{proof}

\begin{proof}[Proof of Lem.~\ref{lem:coefficient}]
We first consider the operator $\{E^{\beta}_{cd}\}$ for $\beta\in\young{D}{n}$ and $c,d\in\range{m_\beta}$ defined on $(\CC^D)^{\otimes n+1}$ for $D\geq n+1$.  For $\tau=(n,n+1)\in\mfS_{n+1}$, 
\begin{align}
    \Tr_{n+1} [V_\tau(E^{\beta}_{cd} \otimes \1_{\CC^D})V_\tau^{\dagger}] = \Tr_{n} E^{\beta}_{cd} \otimes \1_{\CC^D}
\end{align}
holds.  The left-hand side is calculated as
\begin{align}
    &\Tr_{n+1} [V_\tau(E^{\beta}_{cd} \otimes \1_{\CC^D})V_\tau^{\dagger}]\nonumber\\
    &=
    \sum_{\mu\in\beta+\square} \Tr_{n+1} [V_\tau (E^\mu_{c^\beta_\mu d^\beta_\mu}) V_\tau^\dagger]\\
    &=
    \sum_{\mu\in\beta+\square} \sum_{i,j=1}^{m_\mu} [\tau_\mu]_{i c_\mu^\beta} [\tau_\mu^\dagger]_{d_\mu^\beta j} \Tr_{n+1} (E^\mu_{ij})\\
    &=
    \sum_{\mu\in\beta+\square} \sum_{\alpha\in\mu-\square} \sum_{a,b=1}^{m_\alpha}
    [\tau_\mu]_{a_\mu^\alpha c_\mu^\beta} [\tau_\mu^\dagger]_{d_\mu^\beta b_\mu^\alpha} \frac{d_\mu}{d_\alpha} E^{\alpha}_{ab}\\
    &=
    \sum_{\alpha\in\young{D}{n}}
    \sum_{\mu\in(\alpha+\square) \cap (\beta+\square)}
    \sum_{a,b=1}^{m_\alpha}
    [\tau_\mu]_{a_\mu^\alpha c_\mu^\beta} [\tau_\mu^\dagger]_{d_\mu^\beta b_\mu^\alpha} \frac{d_\mu}{d_\alpha} E^{\alpha}_{ab}\\
    &=
    \sum_{\alpha\in\young{D}{n}}
    \sum_{\mu\in(\alpha+\square) \cap (\beta+\square)}
    \sum_{a,b=1}^{m_\alpha}
    [\tau_\mu]_{c_\mu^\beta a_\mu^\alpha} [\tau_\mu^\dagger]_{b_\mu^\alpha d_\mu^\beta} \frac{d_\mu}{d_\alpha} E^{\alpha}_{ab}.
\end{align}
Here, we use the relation $\tau_\mu^\sfT = (\tau_\mu^\dagger)^* = \tau_\mu^* = \tau_\mu$.
The right-hand side is calculated as
\begin{align}
    &\Tr_{n} E^{\beta}_{cd} \otimes \1_{\CC^D} \nonumber\\
    &=
    \sum_{\lambda\in\beta-\square} \sum_{p,q=1}^{m_\lambda} \frac{d_\beta}{d_\lambda} \delta_{c, p^\lambda_\beta} \delta_{d, q^\lambda_\beta} E^{\lambda}_{pq} \otimes \1_{\CC^D}\\
    &=
    \sum_{\lambda\in\beta-\square} \sum_{p,q=1}^{m_\lambda}
    \sum_{\alpha\in\lambda+\square}
    \frac{d_\beta}{d_\lambda} \delta_{c, p^\lambda_\beta} \delta_{d, q^\lambda_\beta} E^{\alpha}_{p^{\lambda}_{\alpha} q^{\lambda}_{\alpha}}\\
    &=
    \sum_{\alpha\in\young{D}{n}}
    \sum_{\lambda\in(\alpha-\square) \cap (\beta-\square)} \sum_{p,q=1}^{m_\lambda}
    \frac{d_\beta}{d_\lambda} \delta_{c, p^\lambda_\beta} \delta_{d, q^\lambda_\beta} E^{\alpha}_{p^{\lambda}_{\alpha} q^{\lambda}_{\alpha}}.
\end{align}
Therefore, we obtain
\begin{align}
    &\sum_{\mu\in(\alpha+\square) \cap (\beta+\square)} \frac{d_\mu}{d_\alpha} [\tau_\mu]_{c_\mu^\beta a_\mu^\alpha} [\tau_\mu^\dagger]_{b_\mu^\alpha d_\mu^\beta} \nonumber\\
    &= \sum_{\lambda\in(\alpha-\square) \cap (\beta-\square)} \frac{d_\beta}{d_\lambda} \delta_{a, p^\lambda_\alpha} \delta_{b, q_\lambda^\alpha} \delta_{c, p^\lambda_\beta} \delta_{d, q^\lambda_\beta}.
\end{align}
Due to the hook-content formula \eqref{eq:hook-content} shown in Appendix~\ref{appendix_sec:shur_weyl_duality} for the dimensions of the irreducible representations of the unitary group and the symmetric group,
\begin{align}
    \lim_{D\to \infty} \frac{1}{D} \frac{d_\mu}{d_\alpha}&= \lim_{D\to \infty} (n+1)\frac{D+i-j}{D}\frac{m_\mu}{m_\alpha}=(n+1) \frac{m_\mu}{m_\alpha},\\
    \lim_{D\to \infty} \frac{1}{D} \frac{d_\beta}{d_\lambda}&=\lim_{D\to \infty} n\frac{D+k-l}{D}\frac{m_\beta}{m_\lambda}=n\frac{m_\beta}{m_\lambda}
\end{align}
hold, where $(i,j) ((k,l))$ is the position of the box added to $\alpha (\beta)$ to obtain $\mu (\nu)$.
Therefore, we obtain
\begin{align}
    A^{\alpha\beta}_{abcd}
    &=\frac{1}{n+1}\sum_{\mu\in(\alpha+\square) \cap (\beta+\square)} \frac{m_\mu}{m_\alpha} [\tau_\mu]_{c_\mu^\beta a_\mu^\alpha} [\tau_\mu^\dagger]_{b_\mu^\alpha d_\mu^\beta}  \\
    &= \frac{1}{n} \sum_{\lambda\in(\alpha-\square) \cap (\beta-\square)} \frac{m_\beta}{m_\lambda} \delta_{a, p^\lambda_\alpha} \delta_{b, q_\lambda^\alpha} \delta_{c, p^\lambda_\beta} \delta_{d, q^\lambda_\beta}.
\end{align}
\end{proof}

\bibliographystyle{IEEEtran}
\bibliography{main}

\begin{thebibliography}{10}
\providecommand{\url}[1]{#1}
\csname url@samestyle\endcsname
\providecommand{\newblock}{\relax}
\providecommand{\bibinfo}[2]{#2}
\providecommand{\BIBentrySTDinterwordspacing}{\spaceskip=0pt\relax}
\providecommand{\BIBentryALTinterwordstretchfactor}{4}
\providecommand{\BIBentryALTinterwordspacing}{\spaceskip=\fontdimen2\font plus
\BIBentryALTinterwordstretchfactor\fontdimen3\font minus \fontdimen4\font\relax}
\providecommand{\BIBforeignlanguage}[2]{{%
\expandafter\ifx\csname l@#1\endcsname\relax
\typeout{** WARNING: IEEEtran.bst: No hyphenation pattern has been}%
\typeout{** loaded for the language `#1'. Using the pattern for}%
\typeout{** the default language instead.}%
\else
\language=\csname l@#1\endcsname
\fi
#2}}
\providecommand{\BIBdecl}{\relax}
\BIBdecl

\bibitem{bennett1993teleporting}
\BIBentryALTinterwordspacing
C.~H. Bennett, G.~Brassard, C.~Cr\'epeau, R.~Jozsa, A.~Peres, and W.~K. Wootters, ``{Teleporting an unknown quantum state via dual classical and Einstein-Podolsky-Rosen channels},'' \emph{Phys. Rev. Lett.}, vol.~70, pp. 1895--1899, Mar 1993. [Online]. Available: \url{https://link.aps.org/doi/10.1103/PhysRevLett.70.1895}
\BIBentrySTDinterwordspacing

\bibitem{bennett1992communication}
\BIBentryALTinterwordspacing
C.~H. Bennett and S.~J. Wiesner, ``{Communication via one- and two-particle operators on Einstein-Podolsky-Rosen states},'' \emph{Phys. Rev. Lett.}, vol.~69, pp. 2881--2884, Nov 1992. [Online]. Available: \url{https://link.aps.org/doi/10.1103/PhysRevLett.69.2881}
\BIBentrySTDinterwordspacing

\bibitem{gottesman1999demonstrating}
\BIBentryALTinterwordspacing
D.~Gottesman and I.~L. Chuang, ``Demonstrating the viability of universal quantum computation using teleportation and single-qubit operations,'' \emph{Nature}, vol. 402, no. 6760, pp. 390--393, 1999. [Online]. Available: \url{https://doi.org/10.1038/46503}
\BIBentrySTDinterwordspacing

\bibitem{raussendorf2001one}
\BIBentryALTinterwordspacing
R.~Raussendorf and H.~J. Briegel, ``{A One-Way Quantum Computer},'' \emph{Phys. Rev. Lett.}, vol.~86, pp. 5188--5191, May 2001. [Online]. Available: \url{https://link.aps.org/doi/10.1103/PhysRevLett.86.5188}
\BIBentrySTDinterwordspacing

\bibitem{horodecki2009quantum}
\BIBentryALTinterwordspacing
R.~Horodecki, P.~Horodecki, M.~Horodecki, and K.~Horodecki, ``Quantum entanglement,'' \emph{Rev. Mod. Phys.}, vol.~81, pp. 865--942, Jun 2009. [Online]. Available: \url{https://link.aps.org/doi/10.1103/RevModPhys.81.865}
\BIBentrySTDinterwordspacing

\bibitem{ishizaka2008asymptotic}
\BIBentryALTinterwordspacing
S.~Ishizaka and T.~Hiroshima, ``{Asymptotic Teleportation Scheme as a Universal Programmable Quantum Processor},'' \emph{Phys. Rev. Lett.}, vol. 101, p. 240501, Dec 2008. [Online]. Available: \url{https://link.aps.org/doi/10.1103/PhysRevLett.101.240501}
\BIBentrySTDinterwordspacing

\bibitem{ishizaka2009quantum}
\BIBentryALTinterwordspacing
------, ``Quantum teleportation scheme by selecting one of multiple output ports,'' \emph{Phys. Rev. A}, vol.~79, p. 042306, Apr 2009. [Online]. Available: \url{https://link.aps.org/doi/10.1103/PhysRevA.79.042306}
\BIBentrySTDinterwordspacing

\bibitem{sedlak2019optimal}
\BIBentryALTinterwordspacing
M.~Sedl\'ak, A.~Bisio, and M.~Ziman, ``{Optimal Probabilistic Storage and Retrieval of Unitary Channels},'' \emph{Phys. Rev. Lett.}, vol. 122, p. 170502, May 2019. [Online]. Available: \url{https://link.aps.org/doi/10.1103/PhysRevLett.122.170502}
\BIBentrySTDinterwordspacing

\bibitem{bisio2010optimal}
\BIBentryALTinterwordspacing
A.~Bisio, G.~Chiribella, G.~M. D’Ariano, S.~Facchini, and P.~Perinotti, ``Optimal quantum learning of a unitary transformation,'' \emph{Physical Review A}, vol.~81, no.~3, p. 032324, 2010. [Online]. Available: \url{https://doi.org/10.1103/PhysRevA.81.032324}
\BIBentrySTDinterwordspacing

\bibitem{yang2020optimal}
\BIBentryALTinterwordspacing
Y.~Yang, R.~Renner, and G.~Chiribella, ``{Optimal Universal Programming of Unitary Gates},'' \emph{Phys. Rev. Lett.}, vol. 125, p. 210501, Nov 2020. [Online]. Available: \url{https://link.aps.org/doi/10.1103/PhysRevLett.125.210501}
\BIBentrySTDinterwordspacing

\bibitem{nielsen1997programmable}
\BIBentryALTinterwordspacing
M.~A. Nielsen and I.~L. Chuang, ``{Programmable Quantum Gate Arrays},'' \emph{Phys. Rev. Lett.}, vol.~79, no.~2, p. 321, 1997. [Online]. Available: \url{https://doi.org/10.1103/PhysRevLett.79.321}
\BIBentrySTDinterwordspacing

\bibitem{beigi2011simplified}
\BIBentryALTinterwordspacing
S.~Beigi and R.~König, ``Simplified instantaneous non-local quantum computation with applications to position-based cryptography,'' \emph{New Journal of Physics}, vol.~13, no.~9, p. 093036, sep 2011. [Online]. Available: \url{https://dx.doi.org/10.1088/1367-2630/13/9/093036}
\BIBentrySTDinterwordspacing

\bibitem{buhrman2016quantum}
\BIBentryALTinterwordspacing
H.~Buhrman, {\L}.~Czekaj, A.~Grudka, M.~Horodecki, P.~Horodecki, M.~Markiewicz, F.~Speelman, and S.~Strelchuk, ``{Quantum communication complexity advantage implies violation of a Bell inequality},'' \emph{Proceedings of the National Academy of Sciences}, vol. 113, no.~12, pp. 3191--3196, 2016. [Online]. Available: \url{https://doi.org/10.1073/pnas.1507647113}
\BIBentrySTDinterwordspacing

\bibitem{may2019quantum}
\BIBentryALTinterwordspacing
A.~May, ``Quantum tasks in holography,'' \emph{Journal of High Energy Physics}, vol. 2019, no.~10, pp. 1--39, 2019. [Online]. Available: \url{https://doi.org/10.1007/JHEP10%282019%29233}
\BIBentrySTDinterwordspacing

\bibitem{may2022complexity}
\BIBentryALTinterwordspacing
------, ``Complexity and entanglement in non-local computation and holography,'' \emph{Quantum}, vol.~6, p. 864, 2022. [Online]. Available: \url{https://doi.org/10.22331/q-2022-11-28-864}
\BIBentrySTDinterwordspacing

\bibitem{quintino2019probabilistic}
\BIBentryALTinterwordspacing
M.~T. Quintino, Q.~Dong, A.~Shimbo, A.~Soeda, and M.~Murao, ``Probabilistic exact universal quantum circuits for transforming unitary operations,'' \emph{Phys. Rev. A}, vol. 100, no.~6, p. 062339, 2019. [Online]. Available: \url{https://doi.org/10.1103/PhysRevA.100.062339}
\BIBentrySTDinterwordspacing

\bibitem{quintino2019reversing}
\BIBentryALTinterwordspacing
------, ``{Reversing Unknown Quantum Transformations: Universal Quantum Circuit for Inverting General Unitary Operations},'' \emph{Phys. Rev. Lett.}, vol. 123, p. 210502, Nov 2019. [Online]. Available: \url{https://link.aps.org/doi/10.1103/PhysRevLett.123.210502}
\BIBentrySTDinterwordspacing

\bibitem{quintino2022deterministic}
\BIBentryALTinterwordspacing
M.~T. Quintino and D.~Ebler, ``{Deterministic transformations between unitary operations: Exponential advantage with adaptive quantum circuits and the power of indefinite causality},'' \emph{Quantum}, vol.~6, p. 679, 2022. [Online]. Available: \url{https://doi.org/10.22331/q-2022-03-31-679}
\BIBentrySTDinterwordspacing

\bibitem{yoshida2023universal}
\BIBentryALTinterwordspacing
S.~Yoshida, A.~Soeda, and M.~Murao, ``Universal construction of decoders from encoding black boxes,'' \emph{Quantum}, vol.~7, p. 957, 2023. [Online]. Available: \url{https://doi.org/10.22331/q-2023-03-20-957}
\BIBentrySTDinterwordspacing

\bibitem{kopszak2021multiport}
\BIBentryALTinterwordspacing
P.~Kopszak, M.~Mozrzymas, M.~Studzi{\'n}ski, and M.~Horodecki, ``Multiport based teleportation--transmission of a large amount of quantum information,'' \emph{Quantum}, vol.~5, p. 576, 2021. [Online]. Available: \url{https://doi.org/10.22331/q-2021-11-11-576}
\BIBentrySTDinterwordspacing

\bibitem{studzinski2022efficient}
\BIBentryALTinterwordspacing
M.~Studziński, M.~Mozrzymas, P.~Kopszak, and M.~Horodecki, ``{Efficient Multi Port-Based Teleportation Schemes},'' \emph{IEEE Transactions on Information Theory}, vol.~68, no.~12, pp. 7892--7912, 2022. [Online]. Available: \url{https://doi.org/10.1109/TIT.2022.3187852}
\BIBentrySTDinterwordspacing

\bibitem{pereira2023continuous}
\BIBentryALTinterwordspacing
J.~L. Pereira, L.~Banchi, and S.~Pirandola, ``Continuous variable port-based teleportation,'' \emph{Journal of Physics A: Mathematical and Theoretical}, vol.~57, no.~1, p. 015305, 2023. [Online]. Available: \url{https://doi.org/10.1088/1751-8121/ad0ce2}
\BIBentrySTDinterwordspacing

\bibitem{ishizaka2015some}
\BIBentryALTinterwordspacing
S.~Ishizaka, ``Some remarks on port-based teleportation,'' 2015. [Online]. Available: \url{https://arxiv.org/abs/1506.01555}
\BIBentrySTDinterwordspacing

\bibitem{wang2016higher}
\BIBentryALTinterwordspacing
Z.-W. Wang and S.~L. Braunstein, ``Higher-dimensional performance of port-based teleportation,'' \emph{Scientific Reports}, vol.~6, no.~1, p. 33004, 2016. [Online]. Available: \url{https://doi.org/10.1038/srep33004}
\BIBentrySTDinterwordspacing

\bibitem{studzinski2017port}
\BIBentryALTinterwordspacing
M.~Studzi{\'n}ski, S.~Strelchuk, M.~Mozrzymas, and M.~Horodecki, ``Port-based teleportation in arbitrary dimension,'' \emph{Scientific reports}, vol.~7, no.~1, pp. 1--11, 2017. [Online]. Available: \url{https://doi.org/10.1038/s41598-017-10051-4}
\BIBentrySTDinterwordspacing

\bibitem{mozrzymas2018optimal}
\BIBentryALTinterwordspacing
M.~Mozrzymas, M.~Studzi{\'n}ski, S.~Strelchuk, and M.~Horodecki, ``Optimal port-based teleportation,'' \emph{New Journal of Physics}, vol.~20, no.~5, p. 053006, 2018. [Online]. Available: \url{https://doi.org/10.1088/1367-2630/aab8e7}
\BIBentrySTDinterwordspacing

\bibitem{mozrzymas2018simplified}
\BIBentryALTinterwordspacing
M.~Mozrzymas, M.~Studzi{\'n}ski, and M.~Horodecki, ``A simplified formalism of the algebra of partially transposed permutation operators with applications,'' \emph{Journal of Physics A: Mathematical and Theoretical}, vol.~51, no.~12, p. 125202, 2018. [Online]. Available: \url{https://doi.org/10.1088/1751-8121/aaad15}
\BIBentrySTDinterwordspacing

\bibitem{christandl2021asymptotic}
\BIBentryALTinterwordspacing
M.~Christandl, F.~Leditzky, C.~Majenz, G.~Smith, F.~Speelman, and M.~Walter, ``Asymptotic performance of port-based teleportation,'' \emph{Communications in Mathematical Physics}, vol. 381, pp. 379--451, 2021. [Online]. Available: \url{https://doi.org/10.1007/s00220-020-03884-0}
\BIBentrySTDinterwordspacing

\bibitem{studzinski2022square}
\BIBentryALTinterwordspacing
M.~Studzi{\'n}ski, M.~Mozrzymas, and P.~Kopszak, ``Square-root measurements and degradation of the resource state in port-based teleportation scheme,'' \emph{Journal of Physics A: Mathematical and Theoretical}, vol.~55, no.~37, p. 375302, 2022. [Online]. Available: \url{https://doi.org/10.1088/1751-8121/ac8530}
\BIBentrySTDinterwordspacing

\bibitem{leditzky2022optimality}
\BIBentryALTinterwordspacing
F.~Leditzky, ``Optimality of the pretty good measurement for port-based teleportation,'' \emph{Letters in Mathematical Physics}, vol. 112, no.~5, p.~98, 2022. [Online]. Available: \url{https://doi.org/10.1007/s11005-022-01592-5}
\BIBentrySTDinterwordspacing

\bibitem{strelchuk2023minimal}
\BIBentryALTinterwordspacing
S.~Strelchuk and M.~Studzi{\'n}ski, ``Minimal port-based teleportation,'' \emph{New Journal of Physics}, vol.~25, no.~6, p. 063012, 2023. [Online]. Available: \url{https://doi.org/10.1088/1367-2630/acdab4}
\BIBentrySTDinterwordspacing

\bibitem{grinko2023gelfand}
\BIBentryALTinterwordspacing
D.~{Grinko}, A.~{Burchardt}, and M.~{Ozols}, ``{Gelfand-Tsetlin basis for partially transposed permutations, with applications to quantum information},'' \emph{arXiv e-prints}, Oct. 2023. [Online]. Available: \url{https://arxiv.org/abs/2310.02252}
\BIBentrySTDinterwordspacing

\bibitem{grinko2023efficient}
\BIBentryALTinterwordspacing
D.~Grinko, A.~Burchardt, and M.~Ozols, ``Efficient quantum circuits for port-based teleportation,'' 2023. [Online]. Available: \url{https://arxiv.org/abs/2312.03188}
\BIBentrySTDinterwordspacing

\bibitem{wills2023efficient}
\BIBentryALTinterwordspacing
A.~Wills, M.-H. Hsieh, and S.~Strelchuk, ``Efficient algorithms for all port-based teleportation protocols,'' 2023. [Online]. Available: \url{https://arxiv.org/abs/2311.12012}
\BIBentrySTDinterwordspacing

\bibitem{fei2023efficient}
\BIBentryALTinterwordspacing
J.~Fei, S.~Timmerman, and P.~Hayden, ``Efficient quantum algorithm for port-based teleportation,'' 2023. [Online]. Available: \url{https://arxiv.org/abs/2310.01637}
\BIBentrySTDinterwordspacing

\bibitem{mozrzymas2024port}
\BIBentryALTinterwordspacing
M.~Mozrzymas, M.~Horodecki, and M.~Studzi{\'n}ski, ``{From port-based teleportation to Frobenius reciprocity theorem: partially reduced irreducible representations and their applications},'' \emph{Letters in Mathematical Physics}, vol. 114, no.~2, p.~56, 2024. [Online]. Available: \url{https://doi.org/10.1007/s11005-024-01800-4}
\BIBentrySTDinterwordspacing

\bibitem{kim2024asymptotic}
\BIBentryALTinterwordspacing
H.~E. Kim and K.~Jeong, ``Asymptotic teleportation schemes bridging between standard and port-based teleportation,'' 2024. [Online]. Available: \url{https://arxiv.org/abs/2403.04315}
\BIBentrySTDinterwordspacing

\bibitem{holevo2011probabilistic}
A.~S. Holevo, \emph{Probabilistic and statistical aspects of quantum theory}.\hskip 1em plus 0.5em minus 0.4em\relax Springer Science \& Business Media, 2011, vol.~1.

\bibitem{acin2001optimal}
\BIBentryALTinterwordspacing
A.~Ac\'{\i}n, E.~Jan\'e, and G.~Vidal, ``Optimal estimation of quantum dynamics,'' \emph{Phys. Rev. A}, vol.~64, p. 050302(R), Oct 2001. [Online]. Available: \url{https://link.aps.org/doi/10.1103/PhysRevA.64.050302}
\BIBentrySTDinterwordspacing

\bibitem{dariano2001using}
\BIBentryALTinterwordspacing
G.~M. D'Ariano, P.~Lo~Presti, and M.~G.~A. Paris, ``{Using Entanglement Improves the Precision of Quantum Measurements},'' \emph{Phys. Rev. Lett.}, vol.~87, p. 270404, Dec 2001. [Online]. Available: \url{https://link.aps.org/doi/10.1103/PhysRevLett.87.270404}
\BIBentrySTDinterwordspacing

\bibitem{fujiwara2001estimation}
\BIBentryALTinterwordspacing
A.~Fujiwara, ``{Estimation of SU(2) operation and dense coding: An information geometric approach},'' \emph{Phys. Rev. A}, vol.~65, p. 012316, Dec 2001. [Online]. Available: \url{https://link.aps.org/doi/10.1103/PhysRevA.65.012316}
\BIBentrySTDinterwordspacing

\bibitem{peres2002covariant}
\BIBentryALTinterwordspacing
A.~Peres and P.~Scudo, ``Covariant quantum measurements may not be optimal,'' \emph{Journal of Modern Optics}, vol.~49, no.~8, pp. 1235--1243, 2002. [Online]. Available: \url{https://doi.org/10.1080/09500340110118449}
\BIBentrySTDinterwordspacing

\bibitem{bagan2004entanglement}
\BIBentryALTinterwordspacing
E.~Bagan, M.~Baig, and {{R. Mu\~noz-Tapia}}, ``Entanglement-assisted alignment of reference frames using a dense covariant coding,'' \emph{Phys. Rev. A}, vol.~69, p. 050303(R), May 2004. [Online]. Available: \url{https://link.aps.org/doi/10.1103/PhysRevA.69.050303}
\BIBentrySTDinterwordspacing

\bibitem{bagan2004quantum}
\BIBentryALTinterwordspacing
------, ``Quantum reverse engineering and reference-frame alignment without nonlocal correlations,'' \emph{Phys. Rev. A}, vol.~70, p. 030301(R), Sep 2004. [Online]. Available: \url{https://link.aps.org/doi/10.1103/PhysRevA.70.030301}
\BIBentrySTDinterwordspacing

\bibitem{ballester2004estimation}
\BIBentryALTinterwordspacing
M.~A. Ballester, ``Estimation of unitary quantum operations,'' \emph{Phys. Rev. A}, vol.~69, p. 022303, Feb 2004. [Online]. Available: \url{https://link.aps.org/doi/10.1103/PhysRevA.69.022303}
\BIBentrySTDinterwordspacing

\bibitem{chiribella2004efficient}
\BIBentryALTinterwordspacing
G.~Chiribella, G.~M. D'Ariano, P.~Perinotti, and M.~F. Sacchi, ``{Efficient Use of Quantum Resources for the Transmission of a Reference Frame},'' \emph{Phys. Rev. Lett.}, vol.~93, p. 180503, Oct 2004. [Online]. Available: \url{https://link.aps.org/doi/10.1103/PhysRevLett.93.180503}
\BIBentrySTDinterwordspacing

\bibitem{chiribella2005optimal}
\BIBentryALTinterwordspacing
G.~Chiribella, G.~M. D'Ariano, and M.~F. Sacchi, ``Optimal estimation of group transformations using entanglement,'' \emph{Phys. Rev. A}, vol.~72, p. 042338, Oct 2005. [Online]. Available: \url{https://link.aps.org/doi/10.1103/PhysRevA.72.042338}
\BIBentrySTDinterwordspacing

\bibitem{hayashi2006parallel}
\BIBentryALTinterwordspacing
M.~Hayashi, ``{Parallel treatment of estimation of $\mathrm{SU}(2)$ and phase estimation},'' \emph{Physics Letters A}, vol. 354, no.~3, pp. 183--189, 2006. [Online]. Available: \url{https://doi.org/10.1063/1.1834432}
\BIBentrySTDinterwordspacing

\bibitem{kahn2007fast}
\BIBentryALTinterwordspacing
J.~Kahn, ``{Fast rate estimation of a unitary operation in $\mathrm{SU}(d)$},'' \emph{Phys. Rev. A}, vol.~75, p. 022326, Feb 2007. [Online]. Available: \url{https://link.aps.org/doi/10.1103/PhysRevA.75.022326}
\BIBentrySTDinterwordspacing

\bibitem{haah2023query}
\BIBentryALTinterwordspacing
J.~Haah, R.~Kothari, R.~O’Donnell, and E.~Tang, ``Query-optimal estimation of unitary channels in diamond distance,'' in \emph{2023 IEEE 64th Annual Symposium on Foundations of Computer Science (FOCS)}.\hskip 1em plus 0.5em minus 0.4em\relax IEEE, 2023, pp. 363--390. [Online]. Available: \url{https://doi.org/10.1109/FOCS57990.2023.00028}
\BIBentrySTDinterwordspacing

\bibitem{hardy2007towards}
\BIBentryALTinterwordspacing
L.~Hardy, ``Towards quantum gravity: a framework for probabilistic theories with non-fixed causal structure,'' \emph{Journal of Physics A: Mathematical and Theoretical}, vol.~40, no.~12, p. 3081, 2007. [Online]. Available: \url{https://doi.org/10.1088/1751-8113/40/12/S12}
\BIBentrySTDinterwordspacing

\bibitem{oreshkov2012quantum}
\BIBentryALTinterwordspacing
O.~Oreshkov, F.~Costa, and {\v{C}}.~Brukner, ``Quantum correlations with no causal order,'' \emph{Nature communications}, vol.~3, no.~1, p. 1092, 2012. [Online]. Available: \url{https://doi.org/10.1038/ncomms2076}
\BIBentrySTDinterwordspacing

\bibitem{chiribella2013quantum}
\BIBentryALTinterwordspacing
G.~Chiribella, G.~M. D'Ariano, P.~Perinotti, and B.~Valiron, ``Quantum computations without definite causal structure,'' \emph{Phys. Rev. A}, vol.~88, p. 022318, Aug 2013. [Online]. Available: \url{https://link.aps.org/doi/10.1103/PhysRevA.88.022318}
\BIBentrySTDinterwordspacing

\bibitem{horodecki1999general}
\BIBentryALTinterwordspacing
M.~Horodecki, P.~Horodecki, and R.~Horodecki, ``General teleportation channel, singlet fraction, and quasidistillation,'' \emph{Phys. Rev. A}, vol.~60, pp. 1888--1898, Sep 1999. [Online]. Available: \url{https://link.aps.org/doi/10.1103/PhysRevA.60.1888}
\BIBentrySTDinterwordspacing

\bibitem{raginsky2001fidelity}
\BIBentryALTinterwordspacing
M.~Raginsky, ``A fidelity measure for quantum channels,'' \emph{Physics Letters A}, vol. 290, no. 1-2, pp. 11--18, 2001. [Online]. Available: \url{https://doi.org/10.1016/S0375-9601%2801%2900640-5}
\BIBentrySTDinterwordspacing

\bibitem{nielsen2010quantum}
\BIBentryALTinterwordspacing
M.~A. Nielsen and I.~Chuang, \emph{{Quantum Computation and Quantum Information}}.\hskip 1em plus 0.5em minus 0.4em\relax Cambridge University Press, 2010. [Online]. Available: \url{https://doi.org/10.1017/CBO9780511976667}
\BIBentrySTDinterwordspacing

\bibitem{mele2024introduction}
\BIBentryALTinterwordspacing
A.~A. Mele, ``{Introduction to Haar Measure Tools in Quantum Information: A Beginner's Tutorial},'' \emph{Quantum}, vol.~8, p. 1340, 2024. [Online]. Available: \url{https://doi.org/10.22331/q-2024-05-08-1340}
\BIBentrySTDinterwordspacing

\bibitem{holevo1993note}
\BIBentryALTinterwordspacing
A.~S. Holevo, ``A note on covariant dynamical semigroups,'' \emph{Reports on mathematical physics}, vol.~32, no.~2, pp. 211--216, 1993. [Online]. Available: \url{https://doi.org/10.1016/0034-4877(93)90014-6}
\BIBentrySTDinterwordspacing

\bibitem{chiribella2008transforming}
\BIBentryALTinterwordspacing
G.~Chiribella, G.~M. D'Ariano, and P.~Perinotti, ``Transforming quantum operations: Quantum supermaps,'' \emph{Europhysics Letters}, vol.~83, no.~3, p. 30004, 2008. [Online]. Available: \url{https://doi.org/10.1209/0295-5075/83/30004}
\BIBentrySTDinterwordspacing

\bibitem{kitaev1997quantum}
\BIBentryALTinterwordspacing
A.~Y. Kitaev, ``Quantum computations: algorithms and error correction,'' \emph{Russian Mathematical Surveys}, vol.~52, no.~6, p. 1191, 1997. [Online]. Available: \url{https://doi.org/10.1070/RM1997v052n06ABEH002155}
\BIBentrySTDinterwordspacing

\bibitem{watrous2005notes}
\BIBentryALTinterwordspacing
J.~Watrous, ``{Notes on super-operator norms induced by Schatten norms},'' \emph{Quantum Info. Comput.}, vol.~5, no.~1, p. 58–68, 2005. [Online]. Available: \url{https://doi.org/10.26421/QIC5.1-6}
\BIBentrySTDinterwordspacing

\bibitem{arora2009computational}
\BIBentryALTinterwordspacing
S.~Arora and B.~Barak, \emph{{Computational Complexity: A Modern Approach}}.\hskip 1em plus 0.5em minus 0.4em\relax Cambridge University Press, 2009. [Online]. Available: \url{https://doi.org/10.1017/CBO9780511804090}
\BIBentrySTDinterwordspacing

\bibitem{zyczkowski2006geometry}
\BIBentryALTinterwordspacing
K.~Zyczkowski and I.~Bengtsson, \emph{Geometry of quantum states}.\hskip 1em plus 0.5em minus 0.4em\relax Cambridge University Press, 2006. [Online]. Available: \url{https://doi.org/10.1017/CBO9780511535048}
\BIBentrySTDinterwordspacing

\bibitem{yoshida2023reversing}
\BIBentryALTinterwordspacing
S.~Yoshida, A.~Soeda, and M.~Murao, ``{Reversing Unknown Qubit-Unitary Operation, Deterministically and Exactly},'' \emph{Phys. Rev. Lett.}, vol. 131, p. 120602, Sep 2023. [Online]. Available: \url{https://link.aps.org/doi/10.1103/PhysRevLett.131.120602}
\BIBentrySTDinterwordspacing

\bibitem{Chen2024Reversing}
\BIBentryALTinterwordspacing
Y.-A. Chen, Y.~Mo, Y.~Liu, L.~Zhang, and X.~Wang, ``{Quantum Algorithm for Reversing Unknown Unitary Evolutions},'' \emph{arXiv e-prints}, Mar. 2024. [Online]. Available: \url{https://arxiv.org/abs/2403.04704}
\BIBentrySTDinterwordspacing

\bibitem{Odake2024LowerBound}
\BIBentryALTinterwordspacing
T.~Odake, S.~Yoshida, and M.~Murao, ``{Analytical Lower Bound on Query Complexity for Transformations of Unknown Unitary Operations},'' \emph{Phys. Rev. Lett.}, vol. 135, p. 230603, Dec 2025. [Online]. Available: \url{https://link.aps.org/doi/10.1103/drp2-rzzw}
\BIBentrySTDinterwordspacing

\bibitem{kay2018tutorial}
\BIBentryALTinterwordspacing
A.~Kay, ``Tutorial on the quantikz package,'' 2018. [Online]. Available: \url{https://arxiv.org/abs/1809.03842}
\BIBentrySTDinterwordspacing

\bibitem{fulton1997young}
\BIBentryALTinterwordspacing
W.~Fulton, \emph{{Young Tableaux: With Applications to Representation Theory and Geometry}}.\hskip 1em plus 0.5em minus 0.4em\relax Cambridge University Press, 1997, no.~35. [Online]. Available: \url{https://doi.org/10.1017/CBO9780511626241}
\BIBentrySTDinterwordspacing

\bibitem{georgi2000lie}
\BIBentryALTinterwordspacing
H.~Georgi, \emph{{Lie Algebras In Particle Physics: From Isospin To Unified Theories}}.\hskip 1em plus 0.5em minus 0.4em\relax CRC Press, Boca Raton, 2000. [Online]. Available: \url{https://doi.org/10.1201/9780429499210}
\BIBentrySTDinterwordspacing

\bibitem{ceccherini2010representation}
\BIBentryALTinterwordspacing
T.~Ceccherini-Silberstein, F.~Scarabotti, and F.~Tolli, \emph{{Representation Theory of the Symmetric Groups: The Okounkov-Vershik Approach, Character Formulas, and Partition Algebras}}.\hskip 1em plus 0.5em minus 0.4em\relax Cambridge University Press, 2010, vol. 121. [Online]. Available: \url{https://doi.org/10.1017/CBO9781139192361}
\BIBentrySTDinterwordspacing

\bibitem{young1931quantitative}
\BIBentryALTinterwordspacing
A.~Young, ``{On Quantitative Substitutional Analysis VI},'' \emph{Proceedings of the London Mathematical Society}, vol.~34, pp. 196--230, 1931. [Online]. Available: \url{https://doi.org/10.1112/plms/s2-34.1.196}
\BIBentrySTDinterwordspacing

\bibitem{yamanouchi1937construction}
\BIBentryALTinterwordspacing
T.~Yamanouchi, ``{On the Construction of Unitary Irreducible Representations of the Symmetric Group},'' \emph{Proceedings of the Physico-Mathematical Society of Japan. 3rd Series}, vol.~19, pp. 436--450, 1937. [Online]. Available: \url{https://doi.org/10.11429/ppmsj1919.19.0_436}
\BIBentrySTDinterwordspacing

\bibitem{belavkin1975optimal}
\BIBentryALTinterwordspacing
V.~P. Belavkin, ``{Optimal Multiple Quantum Statistical Hypothesis Testing},'' \emph{Stochastics: An International Journal of Probability and Stochastic Processes}, vol.~1, no. 1-4, pp. 315--345, 1975. [Online]. Available: \url{https://doi.org/10.1080/17442507508833114}
\BIBentrySTDinterwordspacing

\bibitem{kholevo1979asymptotically}
\BIBentryALTinterwordspacing
A.~S. Kholevo, ``{On Asymptotically Optimal Hypothesis Testing in Quantum Statistics},'' \emph{Theory of Probability \& Its Applications}, vol.~23, no.~2, pp. 411--415, 1979. [Online]. Available: \url{https://doi.org/10.1137/1123048}
\BIBentrySTDinterwordspacing

\bibitem{hausladen1994pretty}
\BIBentryALTinterwordspacing
P.~Hausladen and W.~K. Wootters, ``{A `Pretty Good' Measurement for Distinguishing Quantum States},'' \emph{Journal of Modern Optics}, vol.~41, no.~12, pp. 2385--2390, 1994. [Online]. Available: \url{https://doi.org/10.1080/09500349414552221}
\BIBentrySTDinterwordspacing

\bibitem{choi1975completely}
\BIBentryALTinterwordspacing
M.-D. Choi, ``Completely positive linear maps on complex matrices,'' \emph{Linear algebra and its applications}, vol.~10, no.~3, pp. 285--290, 1975. [Online]. Available: \url{https://doi.org/10.1016/0024-3795(75)90075-0}
\BIBentrySTDinterwordspacing

\bibitem{jamiolkowski1972linear}
\BIBentryALTinterwordspacing
A.~Jamio{\l}kowski, ``Linear transformations which preserve trace and positive semidefiniteness of operators,'' \emph{Reports on Mathematical Physics}, vol.~3, no.~4, pp. 275--278, 1972. [Online]. Available: \url{https://doi.org/10.1016/0034-4877(72)90011-0}
\BIBentrySTDinterwordspacing

\bibitem{chiribella2008quantum}
\BIBentryALTinterwordspacing
G.~Chiribella, G.~M. D’Ariano, and P.~Perinotti, ``{Quantum Circuit Architecture},'' \emph{Phys. Rev. Lett.}, vol. 101, no.~6, p. 060401, 2008. [Online]. Available: \url{https://doi.org/10.1103/PhysRevLett.101.060401}
\BIBentrySTDinterwordspacing

\bibitem{chiribella2009optimal}
\BIBentryALTinterwordspacing
------, ``Optimal covariant quantum networks,'' in \emph{AIP Conference Proceedings}, vol. 1110, no.~1.\hskip 1em plus 0.5em minus 0.4em\relax American Institute of Physics, 2009, pp. 47--56. [Online]. Available: \url{https://doi.org/10.1063/1.3131375}
\BIBentrySTDinterwordspacing

\bibitem{chiribella2008memory}
\BIBentryALTinterwordspacing
------, ``{Memory Effects in Quantum Channel Discrimination},'' \emph{Phys. Rev. Lett.}, vol. 101, no.~18, p. 180501, 2008. [Online]. Available: \url{https://doi.org/10.1103/PhysRevLett.101.180501}
\BIBentrySTDinterwordspacing

\bibitem{bavaresco2021strict}
\BIBentryALTinterwordspacing
J.~Bavaresco, M.~Murao, and M.~T. Quintino, ``{Strict Hierarchy between Parallel, Sequential, and Indefinite-Causal-Order Strategies for Channel Discrimination},'' \emph{Phys. Rev. Lett.}, vol. 127, no.~20, p. 200504, 2021. [Online]. Available: \url{https://doi.org/10.1103/PhysRevLett.127.200504}
\BIBentrySTDinterwordspacing

\bibitem{bavaresco2022unitary}
\BIBentryALTinterwordspacing
------, ``Unitary channel discrimination beyond group structures: Advantages of sequential and indefinite-causal-order strategies,'' \emph{Journal of Mathematical Physics}, vol.~63, no.~4, p. 042203, 2022. [Online]. Available: \url{https://doi.org/10.1063/5.0075919}
\BIBentrySTDinterwordspacing

\bibitem{haah2016sample}
\BIBentryALTinterwordspacing
J.~Haah, A.~W. Harrow, Z.~Ji, X.~Wu, and N.~Yu, ``Sample-optimal tomography of quantum states,'' in \emph{Proceedings of the forty-eighth annual ACM symposium on Theory of Computing}, 2016, pp. 913--925. [Online]. Available: \url{https://doi.org/10.1109/TIT.2017.2719044}
\BIBentrySTDinterwordspacing

\bibitem{bernstein1993quantum}
\BIBentryALTinterwordspacing
E.~Bernstein and U.~Vazirani, ``Quantum complexity theory,'' in \emph{Proceedings of the twenty-fifth annual ACM symposium on Theory of computing}, 1993, pp. 11--20. [Online]. Available: \url{https://doi.org/10.1137/S0097539796300921}
\BIBentrySTDinterwordspacing

\bibitem{gour2019comparison}
\BIBentryALTinterwordspacing
G.~Gour, ``{Comparison of Quantum Channels by Superchannels},'' \emph{IEEE Transactions on Information Theory}, vol.~65, no.~9, pp. 5880--5904, 2019. [Online]. Available: \url{https://doi.org/10.1109/TIT.2019.2907989}
\BIBentrySTDinterwordspacing

\bibitem{wechs2021quantum}
\BIBentryALTinterwordspacing
J.~{Wechs}, H.~{Dourdent}, A.~A. {Abbott}, and C.~{Branciard}, ``{Quantum Circuits with Classical Versus Quantum Control of Causal Order},'' \emph{PRX Quantum}, vol.~2, p. 030335, 2021. [Online]. Available: \url{https://doi.org/10.1103/PRXQuantum.2.030335}
\BIBentrySTDinterwordspacing

\end{thebibliography}

\begin{IEEEbiographynophoto}{Satoshi Yoshida}
received the M.Sc. degree in physics at the University of Tokyo, Japan in 2023, where he is currently pursuing the Ph.D. degree.
His research interests include higher-order quantum computation, quantum learning, fault-tolerant quantum computation, and group theoretical analysis of quantum information processing.
\end{IEEEbiographynophoto}

\begin{IEEEbiographynophoto}{Yuki Koizumi}
is a graduate student in the Department of Applied Physics, Graduate School of Engineering, the University of Tokyo, where he is currently pursuing the M.Sc. degree. His research interests include quantum information processing and quantum computation, with a particular focus on fault-tolerant quantum algorithms, quantum learning theory, and related applications in quantum physics.
\end{IEEEbiographynophoto}

\begin{IEEEbiographynophoto}{Michał Studziński}
received his Ph.D. in Physics from the University of Gdańsk in 2015 and subsequently held postdoctoral positions at the National Quantum Information Centre in Sopot and at the Department of Applied Mathematics and Theoretical Physics, University of Cambridge (2016-2018). He is currently an Adjunct Professor at the University of Gdańsk and, since 2025, a Group Leader of Quantum Devices in Computer Science at the International Centre for Theory of Quantum Technologies (ICTQT). His research explores the role of symmetries and representation theory (e.g., Schur-Weyl duality) in simplifying the description of quantum information problems, as well as higher-order quantum structures. Among his key contributions are results on the second laws of quantum thermodynamics and universal programmable quantum processors.
\end{IEEEbiographynophoto}

\begin{IEEEbiographynophoto}{Marco Túlio Quintino}
received the M.S. degree from the Federal University of Minas Gerais, Brazil, in 2012, the Ph.D. degree from the University of Geneva, Switzerland, in 2016, and the habilitation degree from Sorbonne University, France, in 2025. From 2016 to 2020, he was a Post-Doctoral Researcher with the University of Tokyo, Japan, and from 2020 to 2022, he was a Post-Doctoral Researcher with the University of Vienna and the Institute for Quantum Optics and Quantum Information (IQOQI), Vienna, Austria. Since September 2022, he has been an associate professor (maître de conférences [HDR]) with the computer science department (LIP6) at Sorbonne University, Paris, France. His main research interests include quantum information and quantum computation, with particular focus on quantum correlations, causality in quantum theory, higher-order quantum computing, Bell nonlocality, quantum steering, entanglement, measurement incompatibility, quantum discrimination tasks, and semidefinite programming.
\end{IEEEbiographynophoto}

\begin{IEEEbiographynophoto}{Mio Murao}
received her M.S. and Ph.D. degrees from Ochanomizu University, Tokyo, Japan, in 1993 and 1996, respectively. She held postdoctoral positions at Harvard University (USA), Imperial College London (UK), and RIKEN (Japan). She joined the Department of Physics, Graduate School of Science, The University of Tokyo, where she was appointed Associate Professor in 2001 and promoted to Professor in 2015. Her research interests span a broad range of theoretical topics in quantum information and quantum physics. Her current research focuses on quantum algorithms for higher-order quantum computation, quantum learning, distributed quantum computation, and quantum-programmable physics.
\end{IEEEbiographynophoto}

\end{document}